\documentclass[runningheads]{llncs}
\usepackage[T1]{fontenc}
\usepackage{graphicx}
\usepackage[T1]{fontenc}   
\usepackage[utf8]{inputenc}
\usepackage[english]{babel} 
\usepackage{amsmath,amsfonts,amssymb,mathtools}
\usepackage{amsmath}
\usepackage{colortbl}
\usepackage{blkarray}
\usepackage{makecell}
\usepackage{booktabs}
\usepackage{mathabx}
\usepackage{mathrsfs}
\usepackage{url}
\usepackage{color}
\usepackage[bookmarks=false]{hyperref}
\usepackage{stmaryrd}
 \usepackage{algorithm}
 \usepackage[noend]{algpseudocode}
 \usepackage{tikz}
\usepackage{tikz-cd}
\usetikzlibrary{tikzmark,calc,arrows, shapes, decorations.pathreplacing, positioning, quotes}
\usepackage{subfig}
\usepackage{breakcites}
\usepackage{thmtools}
\usepackage{thm-restate}

\newtheorem{thm}{Theorem}
{\bfseries}{\itshape}
\newtheorem{fact}{Fact}{\bfseries}{\itshape}

\newtheorem{conjec}{Conjecture}
\newtheorem{modeling}{Modeling}
 \renewcommand{\leq}{\leqslant} 
\renewcommand{\geq}{\geqslant} 

\newcommand{\N}{\ensuremath{\mathbb{N}}}

\newcommand{\K}{\ensuremath{\mathbb{K}}}
\newcommand{\F}{\ensuremath{\mathbb{F}}}
\newcommand{\Fq}{\ensuremath{\mathbb{F}_q}}

\newcommand{\Fqm}{\ensuremath{\mathbb{F}_{q^m}}}

\newcommand{\mat}[1]{\ensuremath{\boldsymbol{#1}}}
\newcommand{\code}[1]{\ensuremath{\mathscr{#1}}}

\newcommand{\AC}{\code{A}}

\newcommand{\CC}{\code{C}}
\newcommand{\DC}{\code{D}}

\newcommand{\GC}{\code{G}}

\newcommand{\RC}{\code{R}}
\newcommand{\SC}{\code{S}}
\newcommand{\VC}{\code{V}}

\newcommand{\zerom}{\mat{0}}
\newcommand{\Am}{\mat{A}}
\newcommand{\Bm}{\mat{B}}
\newcommand{\Cm}{\mat{C}}
\newcommand{\Dm}{\mat{D}}

\newcommand{\Gm}{\mat{G}}
\newcommand{\Hm}{\mat{H}}
\renewcommand{\Im}{\mat{I}}
\newcommand{\Mm}{\mat{M}} 
\newcommand{\Nm}{\mat{N}}        
\newcommand{\Pm}{\mat{P}}

\newcommand{\Sm}{\mat{S}}
\newcommand{\Vm}{\mat{V}}
\newcommand{\Xm}{\mat{X}}
\newcommand{\Ym}{\mat{Y}}

\newcommand{\zerov}{\mat{0}}
\newcommand{\av}{\mat{a}}
\newcommand{\bv}{\mat{b}}
\newcommand{\cv}{\mat{c}}
\newcommand{\dv}{\mat{d}}

\newcommand{\mv}{\mat{m}}
\newcommand{\uv}{\mat{u}}
\newcommand{\vv}{\mat{v}}
\newcommand{\xv}{{\mat{x}}}
\newcommand{\yv}{{\mat{y}}}
\newcommand{\zv}{{\mat{z}}}

\newcommand{\cA}{\mathcal{A}}
\newcommand{\cB}{\mathcal{B}}
\newcommand{\cD}{\mathcal{D}}
\newcommand{\cI}{\mathcal{I}}
\newcommand{\cP}{\mathcal{P}}
\newcommand{\cR}{\mathcal{R}}
\newcommand{\cV}{\mathcal{V}}

\newcommand{\GRS}[3]{\text{\bf GRS}_{#1}(#2,#3)}
\newcommand{\Alt}[3]{\code{A}_{#1}(#2, #3)}
\newcommand{\Goppa}[2]{\code{G}(#1, #2)}

\newcommand{\trsp}[1]{{{#1}^{\intercal}}}

\newcommand{\Crel}{\code{C}_\text{rel}}
\newcommand{\Cmat}{\code{C}_\text{mat}}

\DeclareMathOperator{\im}{Im}

\newcommand{\starp}[2]{{#1} \star {#2}}

\newcommand{\sq}[1]{#1^{\star 2}}
\newcommand{\sqb}[1]{\left(#1\right)^{\star 2}}

\newcommand{\map}[4]{\left\{
\begin{array}{ccc}
#1 & \longrightarrow & #2 \\
#3 & \longmapsto     & #4
        \end{array}
      \right.}

\newcommand{\Mac}{\text{Mac}}
\newcommand{\scp}[2]{#1 \cdot #2}

\newcommand{\GL}{\mathbf{GL}}

\newcommand{\rank}{\mathbf{Rank}}

\newcommand{\eqdef}{\stackrel{\text{def}}{=}}
\newcommand{\ie}{\textit{i.e.}\,}
\newcommand{\Span}[2]{\left\langle \, #1 \, \right\rangle_{#2}}
\newcommand{\Fspan}[1]{\left\langle \, #1 \, \right\rangle_{\F}}

\newcommand{\Fqmspan}[1]{\left\langle \, #1 \, \right\rangle_{\Fqm}}
\newcommand{\floor}[1]{\left\lfloor #1 \right\rfloor}
\newcommand{\ceil}[1]{\left\lceil #1 \right\rceil}

\newcommand{\Iintv}[2]{\llbracket #1 , #2 \rrbracket}
\newcommand{\card}[1]{\lvert #1 \rvert}

\newcommand{\HS}{HS}
\newcommand{\HF}{HF}
\newcommand{\Sym}{\mathbf{Sym}}
\newcommand{\Skew}{\mathbf{Skew}}
\newcommand{\Minors}{\mathbf{Minors}}
\newcommand{\ext}[2]{#1_{#2}}
\newcommand{\dgv}{d_{\text{GV}}}

\newcommand{\ideal}[1]{\cI\left(#1\right)}
\newcommand{\dreg}{d_{\textrm{reg}}}
\newcommand{\Cmess}{\textrm{C}_{\textrm{mess}}}
\newcommand{\Ckey}{\textrm{C}_{\textrm{key}}}
\newcommand{\Cdist}{\textrm{C}_{\textrm{dist}}}

\newcommand{\bigO}[1]{\mathcal{O}\left(#1\right)}
\newcommand{\Th}[1]{\Theta\left(#1\right)}
\newcommand{\OO}[1]{\bigO{#1}}

\begin{document}
\title{A new approach based on quadratic forms to attack the McEliece cryptosystem}
\titlerunning{A new approach based on quadratic forms to attack McEliece}
\author{Alain Couvreur\inst{1} \and
Rocco Mora\inst{2} \and
Jean-Pierre Tillich\inst{2}}
\authorrunning{A. Couvreur, R. Mora, J.-P. Tillich}
\institute{Inria Saclay, LIX, CNRS UMR 7161, École Polytechnique, 1 rue Honoré d’Estienne d’Orves, 91120 Palaiseau Cedex \and Inria Paris, 2 rue Simone Iff, 75012 Paris, France \\
\email{\{alain.couvreur,rocco.mora,jean-pierre.tillich\}@inria.fr}}
\maketitle              \begin{abstract}
  We introduce a novel algebraic approach for attacking the
  McEliece cryptosystem which is currently at the $4$-th round of the
  NIST competition.  The contributions of the article are twofold.
  (1) We present a new distinguisher on alternant and Goppa codes
  working in a much broader range of parameters than \cite{FGOPT11}.  (2) With this approach we also
  provide a polynomial--time key recovery attack on
  alternant codes which are distinguishable with the distinguisher
  \cite{FGOPT11}.  
  
  These results are obtained by introducing a subspace of matrices
  representing quadratic forms. Those are associated with quadratic
  relations for the component-wise product in the dual of the Goppa
  (or alternant) code of the cryptosystem. It turns out that this subspace 
of matrices contains matrices of unusually small rank in the case of alternant or 
Goppa codes ($2$ or $3$ depending on the field characteristic)
revealing the secret polynomial structure
  of the code.
 MinRank solvers can then be used to recover the
  secret key of the scheme. We devise a dedicated algebraic modeling in
  characteristic $2$ where the Gröbner basis techniques to solve it can be analyzed.
This computation behaves differently
  when applied to the matrix space associated with a
  random code rather than with a Goppa or an alternant code. This
  gives a distinguisher of the latter code families, which contrarily 
to the one proposed in \cite{FGOPT11} working only in a tiny parameter regime is now
able to work for code rates above $\frac{2}{3}$. It applies to most of the 
instantiations of the McEliece cryptosystem in the literature. It coincides with the one of \cite{FGOPT11}
when the latter can be applied (and is therefore of polynomial complexity in this case). However, its complexity 
increases significantly when \cite{FGOPT11} does not apply anymore, but stays subexponential as long as the co-dimension of the code is sublinear in the length (with an asymptotic exponent which is below those of all known key recovery or message attacks). For the concrete parameters of the McEliece NIST submission \cite{ABCCGLMMMNPPPSSSTW20}, its 
complexity is way too complex to threaten the cryptosystem, but is smaller than known key recovery attacks for most 
of the parameters of the submission. This subspace of quadratic forms can also be used in a different manner
   to give a polynomial time attack of the McEliece cryptosystem
  based on generic alternant codes or Goppa codes provided that these codes are distinguishable by the method of
  \cite{FGOPT11}, and in the Goppa case we need the additional assumption that its degree is less than $q-1$, where $q$ is the
  alphabet size of the code.
\end{abstract}

\section{Introduction}
\subsection*{The McEliece Cryptosystem}
The McEliece encryption scheme \cite{M78}, which is only a few months
younger than RSA \cite{RSA78}, is a code-based cryptosystem built upon
the family of binary Goppa codes.  It is equipped with very fast
encryption and decryption algorithms and has very small ciphertexts
but large public key size. Contrarily to RSA which is broken by
quantum computers \cite{S94a}, it is also widely viewed as a viable
quantum-safe cryptosystem. A variation of this public key cryptosystem
intended to be IND-CCA secure and an associated key exchange protocol
\cite{ABCCGLMMMNPPPSSSTW20} is one of the three remaining code-based 
candidates in the fourth round of the NIST post-quantum competition on
post-quantum cryptography.  Its main selling point for being
standardized is that it is the oldest public key cryptosystem which
has resisted all possible attacks be they classical or quantum so far,
this despite very significant efforts to break it.

The consensus right now about this cryptosystem is that key-recovery
attacks that would be able to exploit the underlying algebraic
structure are way more expensive than message-recovery attacks that
use decoding algorithms for generic linear codes. Because of this
reason, the parameters of McEliece encryption scheme are chosen
according to the latest algorithms for decoding a linear code.  This
is also actually another selling point for this cryptosystem, since
despite significant efforts on improving the algorithms for decoding
linear codes, all the classical algorithms for performing this task
are of exponential complexity and this exponent has basically only
decreased by less than $20$ percent for most parameters of interest
after more than 60 years of research
\cite{P62,S88,D89,CC98,MMT11,BJMM12,MO15,BM17}.  The situation is even
more stable when it comes to quantum algorithms \cite{B10,KT17a}.

\subsection*{Key Recovery Attacks}
The best key recovery attack has not changed for many years. It was
given in \cite{LS01} and consists in checking all Goppa polynomials
and all possible supports with the help of \cite{S00}. Its complexity
is also exponential with an exponent which is much bigger than the one
obtained for message recovery attacks. There has been some progress on
this issue, not on the original McEliece cryptosystem, but on
variations of it. This concerns very high rate binary Goppa codes for
devising signature schemes \cite{CFS01}, non-binary Goppa codes over
large alphabets \cite{BLP10,BLP11a}, or more structured versions of
the McEliece system, based on quasi-cyclic alternant codes
\cite{BCGO09,BIGQUAKE} (a family of algebraic codes containing Goppa
codes retaining the essential algebraic features of Goppa codes) or on
quasi-dyadic Goppa codes such as \cite{MB09,BLM11,BBBCDGGHKNNPR17}.

The quasi-cyclic or quasi-dyadic alternant/Goppa codes have been
attacked in \cite{FOPT10,GL09,BC18} by providing a suitable algebraic
modeling for the secret key and then solving the algebraic system with
Gr\"obner bases techniques. This algebraic modeling tries
to recover the underlying polynomial structure of these codes coming
from the underlying generalized Reed-Solomon structure by using just
an arbitrary generator matrix of the alternant or Goppa code which is
given by the public key of the scheme. This is basically the secret
key of the scheme. It allows to decode the alternant or Goppa code and
therefore all possible ciphertexts. Recall that a generalized
Reed-Solomon code is defined by
\begin{definition}[Generalized Reed-Solomon (GRS) code
  ]\label{def:GRS}
  Let $\xv=(x_1,\dots,$ $x_n)\in\F^n$ be a vector of pairwise distinct
  entries and $\yv=(y_1,\dots,y_n)\in\F^n$ a vector of nonzero
  entries, where $\F$ is a finite field. The \emph{generalized
    Reed-Solomon (GRS) code} over $\F$ of dimension $k$ with
  \emph{support} $\xv$ and \emph{multiplier} $\yv$ is
	\[
      \GRS{k}{\xv}{\yv}\eqdef\{(y_1 P(x_1),\dots,y_n P(x_n)) \mid P
      \in \F[z], \deg P < k\}.
	\]
\end{definition}
Alternant codes are defined as subfield subcodes of GRS codes, meaning
that an alternant code $\AC$ of length $n$ is defined over some field
$\Fq$ whereas the underlying GRS code $\CC$ is defined over an
extension field $\Fqm$ of degree $m$.
The alternant code is defined in this case as the set of codewords of
the GRS code whose entries all belong to the subfield $\Fq$, {\em i.e}
$$
\AC = \CC \cap \Fq^n.
$$
Rather than trying to recover the polynomial structure of the
underlying GRS code, the algebraic attack in \cite{FOPT10}
actually recovers the polynomial structure of the {\em dual
  code}. Recall that the dual code of a linear code is defined by
\begin{definition}[dual code] The dual $\CC^\perp$ of a linear code
  $\CC$ of length $n$ over $\Fq$ is the subspace of $\Fq^n$ defined by
  $ \CC^\perp \eqdef \{\dv \in \Fq^n: \scp{\dv}{\cv}=0,\;\forall \cv
  \in \CC\}, $ where $\scp{\dv}{\cv} = \sum_{i=1}^n c_i d_i$ with
  $\cv=(c_i)_{1 \leq i \leq n}$ and $\dv=(d_i)_{1 \leq i \leq n}$.
\end{definition}
The dual code of an alternant code has also a polynomial structure
owing to the fact that the dual of a GRS code is actually a GRS code:
\begin{proposition}[{\cite[Theorem~4, p.~304]{MS86}\label{pr:dual_GRS}}]
  Let $\GRS{r}{\xv}{\yv}$ be a GRS code of length $n$. Its dual is
  also a GRS code. In particular
  $ \GRS{r}{\xv}{\yv}^\perp=\GRS{n-r}{\xv}{\yv^\perp}, $ with
  $
  \yv^\perp\eqdef\left(\frac{1}{\pi'_\xv(x_1)y_1},\dots,\frac{1}{\pi'_\xv(x_n)y_n}\right)$,
  where $\pi_\xv(z)\eqdef \prod_{i=1}^n (z-x_i)$ and $\pi'_\xv$ is its
  derivative.
\end{proposition}
It is actually the dual of the underlying GRS code which serves to
define the multiplier and the support of an alternant code as shown by
\begin{definition}[alternant code]
  Let $n\le q^m$, for some positive integer $m$. Let
  $\GRS{r}{\xv}{\yv}$ be the GRS code over $\Fqm$ of dimension $r$
  with support $\xv \in \Fqm^n$ and multiplier $\yv\in
  (\Fqm^*)^n$. The \emph{alternant code} with support $\xv$ and
  multiplier $\yv$, \emph{degree} $r$ over $\Fq$ is
	\[
      \Alt{r}{\xv}{\yv}\eqdef \GRS{r}{\xv}{\yv}^\perp \cap
      \F_q^n=\GRS{n-r}{\xv}{\yv^\perp} \cap \Fq^n.
	\]
	The integer $m$ is called \emph{extension degree} of the
    alternant code.
\end{definition}
It is much more convenient to recover with an algebraic modeling the
support and the multiplier of the dual of the underlying GRS code
because {\em any} codeword $\cv=(c_i)_{1 \leq i \leq n}$ of the
alternant code $\Alt{r}{\xv}{\yv}$ is readily seen to be orthogonal to
any codeword $\dv$ of $\GRS{r}{\xv}{\yv}$, {\em i.e.}
$\scp{\cv}{\dv}=0$.  The algebraic modeling of \cite{FOPT10} is based
on such equations where the unknowns are the entries of $\xv$ and
$\yv$.
Goppa codes can be recovered from this approach too, since they are
particular alternant codes:
\begin{definition}[Goppa code]
  Let $\xv\in\Fqm^n$ be a support vector and $\Gamma\in\Fqm[z]$ a
  polynomial of degree $r$ such that $\Gamma(x_i)\neq 0$ for all
  $i \in \{1,\dots,n\}$. The \emph{Goppa code} of degree $r$ with
  support $\xv$ and \emph{Goppa polynomial} $\Gamma$ is defined as
  $ \Goppa{\xv}{\Gamma}\eqdef\Alt{r}{\xv}{\yv},$ where
  $\yv\eqdef\left(\frac{1}{\Gamma(x_1)},\dots,\frac{1}{\Gamma(x_n)}\right).$
\end{definition}
The algebraic modeling approach of \cite{FOPT10} worked because the
quasi cyclic/ dyadic structure allowed to reduce drastically the number
of unknowns of the algebraic system when compared to the original
McEliece cryptosystem. A variant of this algebraic modeling was
introduced in \cite{FPP14} to attack certain parameters of the variant
of the McEliece cryptosystem \cite{BLP10,BLP11a} based on wild Goppa
codes/wild Goppa codes incognito. It only involves equations on the
multiplier $\yv$ of the Goppa code induced by the wild Goppa
structure. The McEliece cryptosystem based on plain binary Goppa codes
seems immune to both the approaches of \cite{FOPT10} and
\cite{FPP14}. The first one because the degree and the number of
variables of the resulting system are most certainly too big to make
such an approach likely to succeed if not at the cost of a very high
exponential complexity (but this has to be confirmed by a rigorous
analysis which is hard to perform because Gr\"obner bases techniques
perform here very differently from a generic system). The second one
because this modeling does not apply to binary Goppa codes. In
particular, it needs a very small extension degree and a code alphabet
size that are prime powers rather than prime.

It was also found that Gr\"obner bases techniques when applied to the
algebraic system \cite{FOPT10} behaved very differently when the
system corresponds to a Goppa code instead of a random linear code of
the same length and dimension. This approach led to \cite{FGOPT11}
that gave a way to distinguish high-rate Goppa codes from random
codes.  It is based on the kernel of a linear system related to the
aforementioned algebraic system.  It was shown there to have an
unexpectedly high dimension when instantiated with Goppa codes or the
more general family of alternant codes rather than with random linear
codes. Another interpretation was later on given to this distinguisher
in \cite{MP12}, where it was proved that the kernel dimension is
related to the dimension of the square of the dual of the Goppa
code. Very recently, \cite{MT22} revisited \cite{FGOPT11} and gave
rigorous bounds for the dimensions of the square codes of Goppa or
alternant codes and a better insight into the algebraic structure of
these squares. Recall here that the component-wise/Schur
product/square of codes is defined from the component-wise/Schur
product of vectors $\av=(a_i)_{1 \leq i \leq n}$ and
$\bv=(b_i)_{1 \leq i \leq n}$
$$\starp{\av}{\bv}\eqdef(a_1 b_1,\dots,a_n b_n)$$
by
\begin{definition}
  The \emph{component-wise product of codes} $\CC,\DC$ over $\F$
  with the same length $n$ is defined as
	\[
      \starp{\CC}{\DC}\eqdef \Span{\starp{\cv}{\dv} \mid \cv \in \CC,
        \dv \in \DC}{\F}.  \] If $\CC=\DC$, we call
    $\sq{\CC}\eqdef\starp{\CC}{\CC}$ the \emph{square code} of
    $\CC$.
\end{definition}

The reason why Goppa codes behave differently from random codes for
this product is essentially because the underlying GRS code behaves
very abnormally with respect to the component-wise product. Indeed,
\begin{proposition}[{\cite{CGGOT14}}] \label{pr: square_GRS}
  Let $\GRS{k}{\xv}{\yv}$ be a GRS code with support $\xv$, multiplier
  $\yv$ and dimension $k$. We have
  $\sq{\GRS{k}{\xv}{\yv}}=\GRS{2k-1}{\xv}{\starp{\yv}{\yv}}$.  Hence,
  if $k\le\frac{n+1}{2}$, $\dim_{\Fqm}\sq{(\GRS{k}{\xv}{\yv})}=2k-1$.
	\end{proposition}
On the other hand, random linear codes behave very differently,
because they attain with probability close to $1$ \cite{CCMZ15} the
general upper bound on the dimension given by
$ \dim_{\F} \sq{\CC} \le \min\left(n,\binom{\dim_{\F}
    \CC+1}{2}\right).$ In other words, the dimension of the square of
a random linear code scales quadratically as long as the dimension is
$k = \OO{\sqrt{n}}$ and attains after this the full dimension $n$,
whereas the dimension of the square of a GRS code of dimension $k$
increases only linearly in $k$. This peculiar property of GRS codes
survives in an attenuated form in the square of
the dual of an alternant/Goppa code as shown by \cite{MT22}.

This tool was also instrumental in another breakthrough in this area,
namely that for the first time a polynomial attack \cite{COT14,COT17}
was found on the McEliece scheme when instantiated with Goppa
codes. This was done by using square code considerations. However,
this attack required very special parameters to be carried out: (i)
the extension degree should be $2$, (ii) the Goppa code should be a
wild Goppa code. It is insightful to remark that this attack exploits
the unusually low dimension of the square of wild Goppa codes when
their dimension is low enough whereas the distinguisher of
\cite{FGOPT11} actually uses the small dimension of the square of the
{\em dual} of a Goppa or alternant code. The dual of such codes has a
much more involved structure, in particular it loses a lot of the nice
polynomial structure of the Goppa code (this was essential in the
attack performed in \cite{COT14}).  This is probably the reason why
for a long time the distinguisher of \cite{FGOPT11} has not turned
into an actual attack. However, recently in \cite{BMT23} it has been
found out that in certain cases (i) very small field size $q=2$ or
$q=3$ over which the code is defined, (ii) being a {\em generic
  alternant code} rather than being in the special case of Goppa code,
(iii) being in the region of parameters where the distinguisher of
\cite{FGOPT11} applies, then this distinguisher can actually be turned
into a polynomial-time attack. Note that \cite{BMT23} also made some
crucial improvements in the algebraic modeling of \cite{FOPT10} (in
particular by adding low-degree equations that take into account that
the multiplier and support of the alternant/Goppa code should satisfy
certain constraints).

\subsection*{A new approach}

\subsubsection*{A first idea: non generic quadratic relations on the
  extended dual alternant/Goppa code.}
We devise in this work a radically new approach toward attacking the
McEliece cryptosystem when it is based on alternant or Goppa codes.
This leads to two new contributions : (1) a new distinguisher on
alternant and Goppa codes and (2) a polynomial time key-recovery
attack on the alternant and part of the Goppa codes that are
distinguishable by \cite{FGOPT11}.
Both exploit the structure of the extension over a
larger field of the dual of an alternant/Goppa code. The extension of
a code over a field extension is given by
\begin{definition}[Extension of a code over a field extension]
  Let $\CC$ be a linear code over $\Fq$. We denote by $\CC_{\Fqm}$ the
  $\Fqm$-linear span of $\CC$ in $\Fqm^n$.
\end{definition}

\begin{definition}[Image of a code by the Frobenius map]
  Let $\CC \subseteq \Fqm$ be a code, we define $\CC^{(q)}$ as
  \[
    \CC^{(q)} \eqdef \{(c_1^q, \dots, c_n^q) ~|~ (c_1, \dots, c_n) \in \CC\}.
  \]
\end{definition}
It turns out that the extension of the dual of an alternant code
actually contains GRS codes and their images by the Frobenius map:
\begin{proposition}[{\cite{BMT23}}] \label{prop:dual_alt_fqm} 
  Let $\Alt{r}{\xv}{\yv}$ be an alternant code over $\Fq$. Then \\
  $ \left(\Alt{r}{\xv}{\yv}^\perp\right)_{\Fqm} =\sum_{j=0}^{m-1}
  \GRS{r}{\xv}{\yv}^{(q^j)}= \sum_{j=0}^{m-1}
  \GRS{r}{\xv^{q^j}}{\yv^{q^j}}.$
	\end{proposition}
Observe now that a GRS code contains non-zero codewords $\cv_1$,
$\cv_2$, $\cv_3$ satisfying a very peculiar property, namely
\begin{equation}\label{eq:quadratic_relation}
	\starp{\cv_1}{\cv_3} = \sq{\cv_2}.
\end{equation}
This can be seen by choosing
$\cv_1 = \yv \xv^a = (y_i x_i^a)_{1 \leq i \leq n}$,
$\cv_2 = \yv \xv^b = (y_i x_i^b)_{1 \leq i \leq n}$ and
$ \cv_3 = \yv \xv^c = (y_i x_i^c)_{1 \leq i \leq n}$
for any $a,b,c$ in $\Iintv{0}{r-1}$ satisfying $b = \frac{a+c}{2}$.
Such a relation is unlikely to hold in a random linear code of
dimension $k$, unless it is of rate $k/n$ close to $1$. Therefore the
dual code of our alternant or Goppa code contains very peculiar
codewords. The issue is now how to find them?

\subsubsection*{A new concept: the code of quadratic relations.}
Equation~\eqref{eq:quadratic_relation} can be viewed as a quadratic
relation between codewords.  There is a natural object that can be
brought in that encodes in a natural way quadratic relations
\begin{definition}[Code of quadratic relations] \label{def: crel} Let
  $\CC$ be an $[n,k]$ linear code over $\F$ and let
  $\cV=\{\vv_1,\dots,\vv_k\}$ be a basis of $\CC$. The \textbf{code of
    relations between the Schur's products with respect to $\cV$}
  is
	\[
      \Crel(\cV)\eqdef\{\cv=(c_{i,j})_{1\le i\le j \le k} \mid \sum_{i
        \le j} c_{i,j} \starp{\vv_i}{\vv_j}=0\} \subseteq
      \F^{\binom{k+1}{2}}.
  \]
	\end{definition}
Such an element $\cv=(c_{i,j})_{1\le i\le j \le k} $ of $\Crel(\cV)$
defines a quadratic form as
$$
Q_{\cv}(x_1,\cdots,x_k) = \sum_{ i\le j } c_{i,j} x_i x_j.
$$
When a basis $\cV$ containing the aforementioned $\cv_i$'s is chosen,
there exists an element in $\Crel(\cV)$ whose
associated quadratic form is of the form $x_i x_j - x_\ell^2$ (for
$\vv_i=\cv_1$, $\vv_j=\cv_3$, $\vv_\ell = \cv_2$). In other words,
this quadratic form is of rank $3$ (in odd characteristic). To find
such non--generic elements in $\Crel(\cV)$, it is convenient to
represent the elements of $\Crel(\cV)$ as matrices corresponding to
the bilinear map given by the polar form of the quadratic form,
i.e. the matrix $\Mm_{\cv}$ corresponding to $\cv \in \Crel(\cV)$ that
satisfies for all $\xv$ and $\yv$ in $\Fqm^k$
\begin{equation}
\label{eq:polarform}
\xv \Mm_{\cv} \trsp{\yv} = Q_{\cv}(\xv+\yv) -Q_{\cv}(\xv)-Q_{\cv}(\yv).
\end{equation}
This definition allows to have a matrix definition of the quadratic
form which works both in odd characteristic and characteristic $2$ and
which satisfies the crucial relation \eqref{eq:basis_change} when the
basis is changed. Note that $\Mm_{\cv}$ is symmetric in odd
characteristic, whereas it is skew-symmetric in characteristic $2$.

\begin{remark}
  By {\em skew symmetric} matrices in characteristic $2$ we mean
  symmetric matrices with zero diagonal.
\end{remark}

\begin{definition}[Matrix code of relations] \label{def: cmat_odd}
  Let $\CC$ be an $[n,k]$ linear code over $\F$ and let
  $\cV=\{\vv_1,\dots,\vv_k\}$ be a basis of $\CC$. The \textbf{matrix
    code of relations between the Schur's products with respect to
    $\cV$} is
	\[
      \Cmat(\cV)\eqdef\{\Mm_{\cv}=(m_{i,j})_{\substack{1\le i\le k \\
          1\le j\le k}} \mid \cv=(c_{i,j})_{1\le i\le j \le k} \in
      \Crel(\cV) \} \subseteq \Sym(k, \F),
	\]
	where 
		$\Mm_{\cv}$ is defined as
	$
	\begin{cases}
		m_{i,j}\eqdef  m_{j,i} \eqdef c_{i,j},&\quad 1\le i< j\le k,\\
		m_{i,i} \eqdef  2c_{i,i},&\quad 1\le i\le k.
	\end{cases}
	$
	\end{definition}
The previous discussion shows that if $\cV$ contains the triple
$\cv_1$, $\cv_2$, $\cv_3$, then there exists a matrix of rank $3$ in
the matrix code of relations in odd characteristic.  Note that the
matrix is of rank $2$ in characteristic $2$ since the polar form
corresponding to the quadratic form $Q(\xv) = x_i x_j - x_{\ell}^2$ is
given by
$(x_i + y_i)(x_j+y_j)-(x_\ell+y_\ell)^2- x_i x_j + x_{\ell}^2 - y_i
y_j + y_{\ell}^2 = x_i y_j + x_j y_i$.

Now the point is that even if we do not have a basis containing the
$\cv_i$'s, there are still rank $3$ (or $2$) matrices in the matrix
code of relations. This holds because a change of basis basically
amounts to a congruent matrix code. Indeed if $\cA$ and $\cB$ are two
different bases of the same code, there exists (see
Proposition~\ref{prop: congr_odd}) an invertible
$\Pm \in \F^{k \times k}$ such that
\begin{equation}
\label{eq:basis_change}
\Cmat(\cA)=\trsp{\Pm} \Cmat(\cB) \Pm.
\end{equation}
Therefore for any choice of basis, there exists a rank $3$ matrix in
the corresponding matrix code of relations.  Finding such matrices
can be viewed as a MinRank problem for rank $3$ with symmetric
matrices
\begin{problem}[Symmetric MinRank problem for rank $r$]
	Let $\Mm_1,\cdots,\Mm_K$ be $K$ symmetric matrices in $\F^{N \times N}$. Find an $\Mm \in \Fspan{\Mm_1,\cdots,\Mm_K}$ of rank $r$.
\end{problem}

Of course, the dimension of the matrix code could be so large that
there are rank $3$ (or $2$) matrices which are here by chance and
which are not induced by these unusual quadratic relations between
codewords of the GRS code. We will study this problem and will give in
Section \ref{sec:low_rank} bounds on the parameters of the problem
which rule out this possibility. Basically, the parameters that we
will encounter for breaking McEliece-type systems will avoid this
phenomenon.

\subsubsection*{A dedicated algebraic approach for finding rank $2$
  elements in a skew-symmetric matrix code.}
There are many methods which can be used to solve the MinRank problem,
be they combinatorial \cite{GC00}, based on an algebraic modeling and
solving them with Gröbner basis or XL type techniques, such as
\cite{KS99,FLP08,FSS10,VBCPS19,BBCGPSTV20} or hybrid methods
\cite{BBBGT22}. Basically all of them can be adapted to the symmetric
MinRank problem. One of the most attractive methods for solving the
problem for the parameters we have is the Support Minors approach
introduced in \cite{BBCGPSTV20}. Unfortunately due to the symmetric or
skew-symmetric form of the matrix space, solving the corresponding
system with the proposed XL type approach behaves very differently
from a generic matrix space and its complexity seems very delicate to
predict.  For this reason, we have devised another way of solving the
corresponding MinRank problem in characteristic $2$. First, we took
advantage that the algebraic system describing the variety of
skew-symmetric matrices of rank $\leq 2$ has already been studied in
the literature and Gr\"obner bases are known. Next, we add to this
Gr\"obner basis the linear equations expressing that the
skew-symmetric matrix should also belong to the matrix code of
relations. This allows us to understand the complexity of solving the
corresponding algebraic system.  It turns out that the Gröbner basis
computation behaves very differently when applied to the
skew-symmetric matrix space associated with a random code rather than
with a Goppa or an alternant code.  This clearly yields a way to
distinguish a Goppa code or more generally an alternant code from a
random code. Contrarily to the distinguisher that has been devised in
\cite{FGOPT11} which works only for a very restricted set of
parameters, this new distinguisher basically works already for rates above $\frac{2}{3}$.
This concerns an overwhelming proportion of
code parameters that have been proposed (and all parameters of the NIST submission \cite{ABCCGLMMMNPPPSSSTW20}).
Interestingly enough, for the code parameters where \cite{FGOPT11} works, our new distinguisher
coincides with it. Despite the fact that its complexity 
increases significantly when \cite{FGOPT11} does not apply anymore, it stays subexponential as long as the co-dimension of the code is sublinear in the length. Interestingly enough in this regime, its asymptotic exponent is below those of all known key recovery or message attacks. For the concrete parameters of the McEliece NIST submission \cite{ABCCGLMMMNPPPSSSTW20}, its 
complexity is  too complex to threaten the cryptosystem, but is smaller than known key recovery attacks for most 
of the parameters of the submission.

\subsubsection*{A new attack exploiting rank defective matrices in the
  matrix code of relations.}
There is another way to exploit this matrix code which consists in
observing that for a restricted set of code parameters (i) the degree
$r$ of the alternant code is less than $q+1$ or $q-1$ in the Goppa
case, (ii) the code is distinguishable with the method of
\cite{FGOPT11}, a rank defective matrix in the matrix code of
relations leaks information on the secret polynomial structure of
the code. This can be used to mount a simple attack by just (i)
looking for such matrices by picking enough random elements in the
matrix code and verifying if they are rank defective (ii) and then
exploiting the information gathered here to recover the support and
multiplier of the alternant/Goppa code.

\subsubsection*{Summary of the contributions.}

In a nutshell, our contributions are
\begin{itemize}
\item We introduce a new concept, namely the matrix code of quadratic
  relations which can be derived from the extended dual of the
  Goppa/alternant code for which we want to recover the polynomial
  structure. This is a subspace of symmetric or skew-symmetric
  matrices depending on the field characteristic over which the code
  is defined which has the particular feature of containing very
  low-rank matrices (rank $3$ in odd characteristic, rank $2$ in
  characteristic $2$) which are related to the secret key of the
  corresponding McEliece cryptosystem.
\item We devise a dedicated algebraic approach for finding these
  low-rank matrices in characteristic $2$ when this subspace of
  matrices is formed by skew-symmetric matrices. It takes advantage of
  the fact that we know a Gröbner basis for the algebraic system
  expressing the fact that a skew-symmetric matrix is of rank $\leq 2$
  based on the nullity of all minors of size greater than $2$. This
  system can be solved with the help of Gröbner bases techniques. It
  turns out that the solving process behaves differently when applied
  to the matrix code of quadratic relations associated with a
  random linear code rather than with a Goppa or an alternant
  code. This gives a way to distinguish a Goppa code or more generally
  an alternant code from a random code which contrarily to the
  distinguisher of \cite{FGOPT11,FGOPT13} works for virtually all code
  parameters relevant to cryptography (recall that the latter works
  only for very high rate Goppa or alternant codes). Moreover, the
  complexity of this system solving can be analyzed and an upper bound
  on the complexity of the distinguisher can be given. It is
  polynomial in the same regime of parameters when the distinguisher
  of \cite{FGOPT11} works. Even if its complexity increases significantly outside this regime, it is less complex
than all known attacks in the sublinear co-dimension regime. For the concrete 
NIST submission parameters \cite{ABCCGLMMMNPPPSSSTW20} its complexity is very far away from representing a threat, but is 
below the known key attacks for most of these parameters. This can be considered as a breakthrough in this area.

\item Rank defective elements in this matrix space also reveal
  something about the hidden polynomial structure of the Goppa or
  alternant code in a certain parameter regime, namely when (i) the
  degree $r$ of the alternant code is less than $q+1$ or $q-1$ in the
  Goppa case, (ii) the code is distinguishable with the method of
  \cite{FGOPT11}. We use this to give a polynomial-time attack in such
  a case by just looking for rank defective elements with a random
  search. This complements nicely the polynomial attack which has been
  found in \cite{BMT23} which also needs that the code is
  distinguishable with \cite{FGOPT11}, but works in the reverse
  parameter regime $r \geq q+1$ (and has also additional restrictions,
  code alphabet size either binary or ternary and it does not work for
  Goppa codes).  Note that in conjunction with the filtration of
  \cite{BMT23}, this new attack works for {\em any} distinguishable
  generic alternant code. This gives yet another example of a case
  when the distinguisher of \cite{FGOPT11} turns into an actual attack
  of the scheme.
\end{itemize}

 \section{Notation and preliminaries}
\label{sec:preliminaries}

\subsection{Notation}

\subsubsection*{General notation}

$\Iintv{a}{b}$ indicates the closed integer interval between $a$ and
$b$. We will make use of two notations for finite fields, $\F_q$
denotes the finite field with $q$ elements, but sometimes we do not
indicate the size of it when it is not important to do so and simply
write $\F$. Instead, a general field (not necessarily finite) is
denoted by $\K$ and its algebraic closure by $\overline{\K}$.

\subsubsection*{Vector and matrix  notation.}
Vectors are indicated by lowercase bold letters $\xv$ and matrices by
uppercase bold letters $\Mm$. Given a function $f$ acting on $\F$ and
a vector $\xv=(x_i)_{1\le i \le n} \in \F$, the expression $f(\xv)$ is
the component-wise mapping of $f$ on $\xv$,
i.e. $f(\xv)=(f(x_i))_{1\le i \le n}$. We will even apply this with
functions $f$ acting on $\F \times \F$: for instance for two vectors
$\xv$ and $\yv$ in $\F^n$ and two positive integers $a$ and $b$ we
denote by $\xv^a\yv^b$ the vector $(x_i^a y_i^b)_{1 \leq i \leq n}$.
We will use the same operation over matrices, but in order to avoid
confusion with the matrix product, we use for a matrix
$\Am=(a_{i,j})_{i,j}$ the notation $\Am^{(q)}$ which stands for the
entries of $\Am$ all raised to the power $q$, {i.e.}  the entry
$(i,j)$ of $\Am^{(q)}$ is equal to $a_{i,j}^q$. The scalar product
between $\xv=(x_i)_{1 \leq i \leq n}\in \F^n$ and
$\yv=(y_i)_{1 \leq i \leq n}\in \F^n$ is denoted by $\xv \cdot \yv$
and is defined by $\xv\cdot \yv = \sum_{i=1}^n x_i y_i$.

\subsubsection*{Symmetric and skew-symmetric matrices.}
The set of $k \times k$ symmetric matrices over $\F$ is denoted by
$\Sym(k,\F)$, whereas the corresponding set of skew-symmetric matrices
is denoted by $\Skew(k,\Fq)$. 

\subsubsection*{Vector spaces.}
Vector spaces are indicated by $\CC$. For two vector spaces $\CC$ and
$\DC$, the notation $\CC\oplus \DC$ means that the two vector spaces
are in direct sum, i.e.  that $\CC\cap\DC=\{0\}$. The $\F$-linear
space generated by $\xv_1,\dots,\xv_m \in \F^n$ is denoted by
$\Fspan{\xv_1,\dots,\xv_m}$.

\subsubsection*{Codes.}
A linear code $\CC$ of length $n$ and dimension $k$ over $\F$ is a $k$ dimensional subspace of $\F^n$. We refer to it as an $[n,k]$-code.

\subsubsection*{Ideals.}
Ideals are indicated by calligraphic $\cI$. Given a sequence $S$ of
polynomials, $\cI(S)$ refers to the polynomial ideal generated by such
sequence. Given the polynomials $f_1,\dots,f_m$, we denote by
$\ideal{f_1,\dots,f_m}$ 
the ideal
generated by them. The variety associated with a polynomial ideal
$\cI\subseteq \K[x_1,\dots,x_n]$ is indicated by $\Vm(\cI)$ and
defined as
$\Vm(\cI)=\{\av \in \overline{\K}^n \mid \forall f \in \cI,\;
f(\av)=0\}$.

\subsection{Distinguishable Alternant or Goppa Code }

We will frequently use here the term {\em distinguishable
  alternant/Goppa} (in the sense of \cite{FGOPT11}) code.  They are
defined as
\begin{definition}[Square--distinguishable alternant/Goppa code]
  A (generic) alternant code $\Alt{r}{\xv}{\yv}$ of length $n$ over
  $\Fq$ and extension degree $m$ is said to be {\em
    square--distinguishable} if
\begin{equation}\label{eq:distinguishable_alternant}
  n > \binom{rm+1}{2}-\frac{m}{2}(r-1)\left((2e_{\AC}+1)r-2\frac{q^{e_{\AC}+1}-1}{q-1}\right)
\end{equation}
where $e_{\AC}\eqdef \max\{i \in \mathbb{N} \mid r\ge q^i+1\}=
\floor{\log_q(r-1)}$.\\
A Goppa code $\Goppa{\xv}{\Gamma}$ of the same parameters is said to
be {\em square--distinguishable} if
	\begin{align} 
      n > &  \binom{rm+1}{2}-\frac{m}{2}(r-1)(r-2),\quad &\text{if $r < q-1$} \label{eq:distinguishable_Goppa_e=0}\\
      n > &\binom{rm+1}{2}-\frac{m}{2}r\left((2e_{\GC}+1)r-2(q-1)q^{e_{\GC}-1}-1\right), \quad &\text{otherwise,}
		\label{eq:prediction_Goppa_e>0}
	\end{align}
    where
    $e_{\GC}\eqdef \min \{i \in \mathbb{N} \mid r \le (q-1)^2q^{i}\}
    +1 = \ceil{\log_q\left(\frac{r}{(q-1)^2}\right)}+1$.
  \end{definition}
  This definition is basically due to the fact that there is a way to
  distinguish such codes from random codes in this case
  \cite{FGOPT11}. For our purpose, it is better to use the point of
  view of \cite{MT22} and to notice that they are distinguishable
  because the computation of the dimension of the square of the dual
  code leads to a result which is different from $n$ and
  $\binom{rm+1}{2}$ (which is the expected dimension of the square of
  a dual code of dimension $rm$). This is shown by
\begin{thm}[{\cite{MT22}}] \label{thm: MT22}
	For an alternant code $\Fq$ of length $n$ and extension degree $m$ we have
	\begin{equation}
		\label{eq:prediction_alternant}
		\dim_{\Fq} \sq{(\Alt{r}{\xv}{\yv}^\perp)} \le \min\left\{n, \binom{rm+1}{2}-\frac{m}{2}(r-1)\left((2e_{\AC}+1)r-2\frac{q^{e_{\AC}+1}-1}{q-1}\right)\right\}.
	\end{equation} where $e_{\AC}\eqdef \max\{i \in \mathbb{N} \mid r\ge q^i+1\}=\floor{\log_q(r-1)}$.\\
	For a Goppa code $\Goppa{\xv}{\Gamma}$  of length $n$ over $\Fq$ with Goppa polynomial $\Gamma(X) \in \Fqm[X]$ of degree $r$ we have
	\begin{align} 
      \dim \sq{(\Goppa{\xv}{\Gamma}^\perp)}& \le  \min\left\{n,\binom{rm+1}{2}-\frac{m}{2}(r-1)(r-2)\right\},\;\;\text{if $r < q-1$} \label{eq:prediction_Goppa_e=0}\\
      \dim \sq{(\Goppa{\xv}{\Gamma}^\perp)}& \le  \min \left\{n,\binom{rm+1}{2}-\frac{m}{2}r\left((2e_{\GC}+1)r-2(q-1)q^{e_{\GC}-1}-1\right)\right\}, \;\;\text{otherwise,}
		\label{eq:prediction_Goppa_e>0}
	\end{align}
	where	$e_{\GC}\eqdef \min \{i \in \mathbb{N} \mid r \le (q-1)^2q^{i}\} +1 = \ceil{\log_q\left(\frac{r}{(q-1)^2}\right)}+1$.
\end{thm}

 \section{Invariants of the Matrix Code of Quadratic Relations}
\label{sec:quadratic_form}

\subsection{Changing the basis}
The fundamental objects that we have introduced, namely the code of relations $\Crel(\cV)$ and the corresponding matrix code $\Cmat(\cV)$ both depend on the basis $\cV$ which is chosen. 
However, all these matrix codes are isometric for the rank metric, namely the metric $d$ between matrices given by
\[
d(\Xm, \Ym)\eqdef \rank(\Xm-\Ym).
\]
This holds because of the following result:
\begin{restatable}{proposition}{propcongr}  \label{prop: congr_odd}
Let $\cA$ and $\cB$ be two bases of a same $[n,k]$ $\F$-linear code $\CC$, with $\F$. Then $\Cmat(\cA)$ and $\Cmat(\cB)$ are isometric matrix codes, \ie there exists $\Pm\in \GL_k(\F)$ such that
\begin{equation}
\Cmat(\cA)=\trsp{\Pm}\Cmat(\cB) \Pm.
\end{equation}
The matrix $\Pm$ coincides with the change of basis matrix between $\cA$ and $\cB$.
\end{restatable}
This Proposition is proved in Appendix~\ref{sec:app-quadratic}.
This result implies that there are several fundamental quantities which stay invariant when considering different bases, such as for instance
\begin{itemize}
\item the distribution of ranks $\{n_i, 0 \leq i \leq k\}$ where $n_i$ is the number of matrices in $\Cmat(\cV)$ of rank $i$;
\item the dimension of $\Cmat(\cV)$.
\end{itemize}

We will sometime avoid specifying the basis, and simply write $\Cmat$, when referring to invariants for the code.

\subsection{Dimension}
We can be a little bit more specific concerning the dimension.
In general, two different bases of a same code provide different codes of relations. The corresponding dimension, instead, is an invariant:
\begin{proposition}\label{prop:dimension}
	Let $\CC\subseteq \F^n$ be an $[n,k]$ linear code with ordered basis $\cV$. Then
\begin{eqnarray*}
\dim_{\F} \Crel(\cV)&=&\binom{k+1}{2}-\dim_{\F} \sq{\CC} \\
\dim_{\F} \Cmat(\cV)& =&\dim_{\F} \Crel(\cV).\end{eqnarray*}
\end{proposition}
\begin{proof}
	The first point directly follows by applying the rank-nullity theorem with respect to the map $T\colon \F^{\binom{k+1}{2}}\rightarrow \F^n, \; T(\cv)=\sum_{i \le j} c_{i,j} \starp{\vv_i}{\vv_j}$:
	\[
      \binom{k+1}{2}=\dim_{\F}
      \F^{\binom{k+1}{2}}=\dim_{\F}\im(T)+\dim_{\F}\ker(T)=\dim_{\F}
      \sq{\CC}+\dim_{\F} \Crel(\cV).\]
    For the second point,
    consider the linear map 
    \[
        \map{\Crel(\cV)}{\Cmat(\cV)}{\cv}{\Mm_{\cv},}
      \]
      where $\Mm_{\cv}$ is defined in Definition~\ref{def: cmat_odd}.
      Note that $\Cmat(\cV)$ is defined as the image of the above
      map, hence the map is surjective by design.
      Let us prove that it is injective.
      In odd characteristic, it is straightforward to see that
      the kernel of this map is zero. 
      In even characteristic the kernel of this map is composed of
      ``diagonal'' relations, {\em i.e.}, relations of the form
      \begin{equation}\label{eq:relation_vi}
        \sum_{i=1}^k c_{i,i} \vv_i \star \vv_i = 0.
      \end{equation}
      Note that writing $\vv_i = (v_{i1}, \dots, v_{in})$ we have
      $\vv_i \star \vv_i = (v_{i1}^2, \dots, v_{in}^2)$ which is
      nothing but the vector $\vv_i^{(2)}$ obtained by applying the
      componentwise Frobenius map on the entries of $\vv_i$. Next, by
      the additivity of the Frobenius map,
      relation~(\ref{eq:relation_vi}) becomes
      \[
        \left(\sum_{i=1}^n c_{i,i}^{1/2} \vv_i \right)^{(2)} = 0
      \quad \text{and\ hence,} \quad
        \sum_{i=1}^n c_{i,i}^{1/2} \vv_i  = 0.
      \]
      The latter identity is a linear relation between the $\vv_i$'s
      which form a basis of $\cV$, hence, we deduce that $c_{i,i} = 0$
      for all $i$.  Thus, the kernel of the map is also zero in
      characteristic 2.  \qed
\end{proof}

 \section{Low-rank matrices in $\Cmat$}
\label{sec:low_rank}

\subsection{Low-rank matrices from quadratic relations in \cite{FGOPT13}}
\label{ss:low-rank}
By Proposition \ref{prop: congr_odd}, all the matrix codes
$\Cmat(\cB)$ are isometric for any choice of basis $\cB$. We will be
interested here in showing that the matrix code of quadratic relations
associated to the extension over $\Fqm$ of the dual of an alternant
code $\Alt{r}{\xv}{\yv}$ defined over $\Fq$ contains many low rank
matrices. This is due to the fact that this code contains the GRS
codes $\GRS{r}{\xv^{q^{i}}}{\yv^{q^{i}}}$ for all
$i \in \Iintv{0}{m-1}$ (Proposition \ref{prop:dual_alt_fqm}). This
will be clear if we choose the basis appropriately. We can namely
choose the ordered basis
\begin{equation}\label{eq:basis_choice}
\cA=(\yv,\xv\yv,\dots,\xv^{r-1}\yv,\dots,\yv^{q^{m-1}},(\xv\yv)^{q^{m-1}},\dots,(\xv^{r-1}\yv)^{q^{m-1}}).
\end{equation}
We call this the {\em canonical basis}. It will be convenient to 
denote the $r$ first basis elements by $\av_0\eqdef \yv$, $\av_1\eqdef \xv\yv,\dots,\av_{r-1}\eqdef \xv^{r-1}\yv$ and view the basis as
$$
\cA=(\av_0,\cdots,\av_{r-1},\av_0^q,\cdots,\av_{r-1}^{q},\cdots,\av_0^{q^{m-1}},\cdots,\av_{r-1}^{q^{m-1}}).
$$
There are simple quadratic relations between the $\av_i^{q^j}$ owing to the trivial algebraic relations introduced in \cite{FGOPT13}:
$\starp{(\xv^a \yv)^{q^l}}{(\xv^b \yv)^{q^u}}=\starp{(\xv^c \yv)^{q^l}}{(\xv^d \yv)^{q^u}}$ if $aq^l+bq^u=cq^l+dq^u$. This amounts to the 
quadratic relation between the basis elements
\begin{equation} \label{eq: alt_rel}
\starp{\av_a^{q^l}}{\av_b^{q^u}}-\starp{\av_c^{q^l}}{\av_d^{q^u}}=0.
\end{equation}
It is readily seen that matrix of $\Cmat(\cB)$ corresponding to this
quadratic relation is of rank $4$ with the exception of the case $c=d$
and $l=u$ where it is of rank $3$ (odd characteristic) or rank $2$
(characteristic $2$). Indeed, if we reorder the basis $\cB$ such that
it starts with $\av_a^{q^l}$, $\av_b^{q^l}$, $\av_c^{q^l}$, then it is
readily seen that the matrix $\Mm \in \Cmat(\cB)$ corresponding to
\eqref{eq: alt_rel} has only zeros with the exception of the first
$3\times 3$ block $\Mm'$ which is given by
$$
\Mm' = \begin{bmatrix} 0 & 1 & 0\\1 & 0 & 0 \\ 0 & 0 & -2 \end{bmatrix}\;\;\text{(odd characteristic),} \quad
\Mm' = \begin{bmatrix} 0 & 1 & 0\\1 & 0 & 0 \\ 0 & 0 & 0 \end{bmatrix}\;\;\text{(characteristic $2$).}
$$
This leads to the following fact
\begin{fact}\label{fact:simple}
Consider the alternant code $\Alt{r}{\xv}{\yv}$ of extension degree $m$ and let $\Cmat(\cA)$ be the corresponding matrix code associated to the basis choice \eqref{eq:basis_choice}. 
Let $l \in \Iintv{0}{m-1}$ and $a,b,c$ in $\Iintv{0}{r-1}$ be such that $a+b=2c$. Then the matrix of $\Cmat(\cA)$ corresponding to 
the quadratic relation $\starp{\av_a^{q^l}}{\av_b^{q^l}}-\sqb{\av_c^{q^l}}=0$ is of rank $3$ in odd characteristic and of rank $2$ in characteristic $2$.
\end{fact}
This already shows that there are many rank $2$ or $3$ matrices in
$\Cmat$ corresponding to an alternant code. But it will turn out some
subsets of the set of rank $\leq 2$ matrices of $\Cmat$ form a
vector space of matrices. Moreover, depending on the fact that the
alternant code has a Goppa structure we will have even more low rank
matrices as we show below. We namely have in characteristic $2$
\begin{restatable}{proposition}{propranktwosubspace} \label{prop: rank2_subspaces}
	Let $\Alt{r}{\xv}{\yv}$ be an alternant code of extension degree $m$ and order $r$ over a field of characteristic $2$. Then $\Cmat$ contains $\floor{\frac{r-1}{2}}$-dimensional subspaces of rank-($\le 2$) matrices. If $\Alt{r}{\xv}{\yv}$ is a binary Goppa code with a square-free Goppa polynomial, then $\Cmat$ contains $(r-1)$-dimensional subspaces of rank-($\le 2$) matrices.
\end{restatable}
This proposition is proved in Appendix \S\ref{ss:app-low-rank}.
We can also give a lower bound on the number of such matrices as shown by

\begin{restatable}{proposition}{propcountlowrankalt}\label{prop:count-low-rank-alt}
Let $\Alt{r}{\xv}{\yv}$ be an alternant code in characteristic $2$ and extension degree $m$. 
The matrix code of quadratic relationships $\Cmat$ contains at least $\Omega(m(q^{m(r-2)})$
matrices of rank 2.
\end{restatable}
In the particular case of binary Goppa codes associated to a square-free polynomial (i.e. the standard choice in a McEliece cryptosystem) we have
\begin{restatable}{proposition}{propcountlowrankGoppa}
Let $\Goppa{\xv}{\Gamma}$ be a binary Goppa code of extension degree $m$ with $\Gamma$ a square-free polynomial of degree $r$. Then $\Cmat$ contains at least 
\[m \frac{(q^{mr}-1)(q^{m(r-1)}-1)}{q^{2m}-1}\]
matrices of rank 2.
\end{restatable}
These propositions are proved in Appendix \S\ref{ss:app-low-rank}.
It also turns out that for the ``canonical'' choice mentioned above (namely when choosing the basis $\cA$ given in \eqref{eq:basis_choice}) under certain circumstances, $\Cmat$ contains the subspace of block diagonal skew symmetric matrices with blocks of size $r$
\begin{restatable}{proposition}{proprelbinGoppa} \label{prop: rel_binGoppa}
Let $\Goppa{\xv}{\Gamma}$ be a binary $[n,n-rm]$ Goppa code with $\Gamma$ a square-free polynomial of degree $r$ and let $\cA$ be the canonical basis of $\ext{\Goppa{\xv}{\Gamma}^\perp}{\Fqm}$ given in \eqref{eq:basis_choice} with $\yv = \frac{1}{\Gamma(\xv)}$. Then $\Cmat(\cA)$ contains the space of block-diagonal skew-symmetric matrices with $r\times r$ blocks.
\end{restatable}

\subsection{The random case}
We have described in the previous subsection a family of matrices in $\Cmat(\cA)$ with a small rank. In particular, we found rank 3 matrices for odd characteristic and rank 2 matrices for even characteristic. In the case of binary Goppa codes with square-free Goppa polynomial, the subspace generated by such rank 2 matrices is even bigger. Since the two codes $\Cmat(\cA)$ and $\Cmat(\cB)$ have the same weight distribution, the same number of low-rank matrices must exist for $\Cmat(\cB)$ as well. We may wonder if such low-rank matrices exist in the matrix code of relationships $\Cmat(\cR)$ of an $[n,rm]$ random $\Fqm$-linear code $\RC$ with basis $\cR$. This can be determined by computing the Gilbert-Varshamov distance $\dgv$ for spaces of symmetric (resp. skew-symmetric) matrices, which is the smallest $d$ such that
\begin{equation} \label{eq: GV_Sym}
\card{\Cmat(\cR)}\card{B_d^{(\Sym)}}\ge\card{\Sym(rm,\Fqm)},
\end{equation}
\begin{equation} \label{eq: GV_Skew}
\card{\Cmat(\cR)}\card{B_d^{(\Skew)}}\ge\card{\Skew(rm,\Fqm)},
\end{equation}
where $B_d^{(\Sym)}$ (resp. $B_d^{(\Skew)}$) is the ball of radius $d$
(with respect to the rank metric) of the space of symmetric
(resp. skew-symmetric) matrices.  The rationale of this definition is
that it can be proved that for a random linear code $\CC$ the
probability of having a non zero matrix of rank $\leq d$ in $\CC$ is
upper-bounded by the ratio
$\frac{\card{\CC}\card{B_d^{(\Sym)}}}{\card{\Sym(rm,\Fqm)}}$ in the
symmetric case. A similar bound holds in the skew-symmetric case.  In
a low dimension scenario, more precisely when $\binom{rm+1}{2}\le n$,
the code $\Cmat(\cR)$ is expected to be trivial. This corresponds
indeed to the square distinguishable regime. We will then assume
$\binom{rm+1}{2}>n$.
\begin{restatable}{proposition}{propGV} \label{prop: GV} Let
  $\RC\subset\Fqm^n$ be a random code of dimension $rm$ with basis
  $\cR$ and let $\binom{rm+1}{2}>n$. Under the assumption that
  $\Cmat(\cR)$ has the same the rank weight distribution as a random linear
  matrix code, it contains matrices of rank $\le d$ with non-negligible
  probability iff
\begin{align*}
  &n\le drm-\binom{d}{2} &\quad &\text{(symmetric case)}\\
  &n\le (d+1)rm-\binom{d+1}{2} &\quad &\text{(skew-symmetric case)}
 \end{align*}
\end{restatable}
This proposition is proved in \S \ref{ss:proof_prop: GV}.
In particular, we expect rank-3 symmetric matrices in $\Cmat(\cR)$ for
\begin{equation} \label{eq: n le 3rm-3}
n\le 3rm-3
\end{equation}
and rank-2 skew-symmetric matrices in $\Cmat(\cR)$ in characteristic 2 for
\[
n\le (2+1)rm-\binom{2+1}{2}=3rm-3
\]
as well. 
We observe that for all security levels of Classic McEliece \cite{ABCCGLMMMNPPPSSSTW20}, the code rate is such that $n=\alpha rm$ with $\alpha\in(3.5,5)$. This means that any algorithm that finds low-rank matrices in $\Cmat(\cR)$ represents a distinguisher between Goppa codes (and more in general alternant codes) and random linear codes for Classic McEliece rates.
 \section{A New Distinguisher of Alternant and Goppa Codes in Characteristic $2$}
\label{sec:modeling}
 
We are going to focus here on the particular case of characteristic
$2$ where we want to find rank $2$ matrices in the matrix code of
quadratic relations. We are going to consider a particular algebraic
modeling for finding matrices of this kind for which we can estimate
the running time of Gröbner bases algorithms for solving it. We will
show that the behavior of the Gröbner basis computation is quite
different when applied to the matrix code corresponding to an
alternant (or a Goppa) code rather than to the matrix code
corresponding to a random code of the same dimension and length as the
alternant/Goppa code. This provides clearly a distinguisher of an
alternant or Goppa code whose complexity can be estimated. Interestingly enough,
it coincides with the square distinguisher of \cite{FGOPT11} for the parameters
where the latter applies, but it also permits to distinguish other parameters and can distinguish 
Goppa or alternant codes of rate in the range $[\frac{2}{3},1]$, contrarily to the former
which works only for rate extremely close to $1$.

\subsection{A modeling coming from the Pfaffian ideal.} We are first
going to give an algebraic modelling expressing that a skew-symmetric
matrix $\Mm$ with arbitrary entries is of rank $\leq 2$. To do so, we
express the fact that all minors of size $4$ should be zero. This
implies that $\Mm$ should be of rank $\leq 2$, because any
skew-symmetric matrix is of even rank and therefore cannot have rank
$3$. In other words, let us consider the generic skew-symmetric matrix
$\Mm=(m_{i,j})_{i,j}\in \Skew(s, \Fqm)$, whose entries $m_{i,j}$ with
$1 \le i < j \le s$ are independent variables. Let
$\mv=(m_{i,j})_{1\le i<j\le s}$. We will write sometimes $m_{j,i}$
with $i<j$, this must just be seen as an alias for $m_{i,j}$ and not
as another variable. We denote by $\Minors(\Mm, d)$ the set of all
minors of $\Mm$ of size $d$. The set of specializations of $\Mm$ that
provide rank 2 matrices is the variety of the determinantal ideal
generated by $\Minors(\Mm, 3)$.  We refer the reader to
\cite[\S~15.1]{MS05} Since there do not exist rank 3 matrices in
$\Sym(s, \Fqm)$, the ideal generated by each possible $4\times 4$
minor of $\Mm$ leads to the same variety:
\[
\Vm(\cI( \Minors(\Mm, 3)))=\Vm(\cI(\Minors(\Mm, 4))).
\]
The homogeneous ideal $\cI(\Minors(\Mm, 2l))$ is not radical. The
determinant of a generic skew-symmetric matrix of size $2l\times 2l$
is the square of a polynomial of degree $l$, called \textit{Pfaffian}
\cite[\S~1.1]{W12}. It is well-known that the corresponding radical
ideal is generated by the square roots of a subset of minors, namely
those corresponding to a submatrix with the same subset for row and
column indexes. Note that such matrices are skew-symmetric as well,
and thus their determinant is the square of a Pfaffian polynomial. In
particular, we define
\begin{definition}[Pfaffian ideal for rank 2]
The \textit{Pfaffian ideal} of rank 2 for $\Mm$ in characteristic 2 is
\begin{equation} \label{eq: pfaffian_ideal}
\cP_2(\Mm) \eqdef \ideal{ m_{i,j}m_{k,l}+m_{i,k}m_{j,l}+m_{i,l}m_{j,k} \mid 1\le i<j<k<l\le s },
\end{equation}
\end{definition}
\begin{remark}
  Note that in the definition of the Pfaffian ideal \eqref{eq:
    pfaffian_ideal}, the 4-tuple $(i,j,k,l)$ is given by distinct
  values. Indeed, if two indexes are equal then the following expression
  \[m_{i,j}m_{k,l}+m_{i,k}m_{j,l}+m_{i,l}m_{j,k}\] vanishes
  identically. Thus these equations do not have to be considered.
\end{remark}
We have
\begin{proposition}[{\cite[Theorem 5.1]{HT92}}]
The basis $\{m_{i,j}m_{k,l}+m_{i,k}m_{j,l}+m_{i,l}m_{j,k} \mid 1\le i<j<k<l\le s \}$ is a Gr\"obner basis of $\cP_2(\Mm)$ with respect to a suitable order.
\end{proposition}
Another straightforward result is that
\begin{proposition}
We have $\Vm(\cP_2(\Mm))=\Vm( \cI(\Minors(\Mm, 4)))$.
\end{proposition}
\begin{proof}
One can verify that for any $f\in  \Minors(\Mm, 4)$, $f\in \cP_2(\Mm)$ and for any $f$ in the basis of $\cP_2(\Mm)$, $f\in\sqrt{ \cI(\Minors(\Mm, 4))}$. By Hilbert's Nullstellensatz, the thesis follows. \qed
\end{proof}

Our modeling takes advantage of the deep knowledge we have about this
ideal. We express now the fact that a matrix $\Mm$ of size $s$ belongs
to some matrix code $\Cmat$ associated to an $[n,k]$ code (which
implies that $s=n-k$ since we are looking at quadratic relations on
the {\em dual} code) by $t \eqdef \binom{s}{2}-\dim \Cmat$ linear
equations $L_1=0,\dots,L_t=0$ linking the $m_{i,j}$'s.  The linear
relations can be obtained as follows:
\begin{itemize}
\item We start from a parity--check
  matrix $\Hm \in \Fq^{s \times n}$ of the code.
\item We compute a basis of the code of quadratic relations described
  in Definition~\ref{def: crel} and deduce basis of the space space
  $\Cmat$ of symmetric (resp. skew symmetric) matrices as described in
  Definition~\ref{def: cmat_odd}.
\item Once we have a basis of $\Cmat$, we compute a basis of its dual which
  can be described as a basis of symmetric (resp. skew symmetric) matrices
  $\Dm_1, \dots, \Dm_t$ satisfying
  \[
    \forall \Mm \in \Cmat, \ \forall i \in \{1, \dots, t\},\quad
    \text{Tr}(\Dm_i \Mm) = 0.
  \]
  This provides the expected $t$ linear relations $L_1, \dots, L_t$ on
  symmetric (resp. skew symmetric) which are satisfied by the elements
  of $\Cmat$.
\end{itemize}
The algebraic modeling we use to express that an element $\Mm$ of
$\Cmat$ is of rank $\leq 2$ uses these $t$ linear equations and the
Gröbner basis of the Pfaffian ideal. In other words, we have the
following algebraic modeling
\begin{modeling}[$\Mm \in \Cmat$, $\rank(\Mm)\leq 2$]\label{mod:Pfaffian}\mbox{ }
\begin{itemize}
\item
$\binom{s}{4}$ quadratic equations $m_{i,j}m_{k,l}+m_{i,k}m_{j,l}+m_{i,l}m_{j,k}=0$ where $1\le i<j<k<l\le s$
\item
$t \eqdef \binom{s}{2}-\dim \Cmat$ linear equations $L_1=0,\dots,L_t=0$ linking the $m_{ij}$'s expressing the fact that 
$\Mm$ belongs to $\Cmat$.
\end{itemize}
\end{modeling}

\subsection{Gr\"obner bases and Hilbert series}
We will be interested in computing the Hilbert series of the ideal corresponding to Modeling \ref{mod:Pfaffian} because it will turn out to behave
differently depending on the code we use for defining the associated matrix code $\Cmat$. This will lead to a distinguisher of alternant or Goppa codes.
Given a homogeneous ideal $\cI \in \K[\zv]$, $\zv=(z_1,\dots,z_n)$, the Hilbert function of the ring $R=\K[\zv]/\cI$ is defined as
\[
\HF_R(d)\eqdef \dim_{\K} (R)= \dim_{\K} (\K[\zv]_d )- \dim_{\K} (\cI_d), 
\] 
where $\K[\zv]_d=\{f \in \K[\zv] \mid \deg(f)=d\}$ and $\cI_d=\cI \cap \K[\zv]_d$. 
Then the Hilbert series of $R$ is the formal series
\[
\HS_R(t)\eqdef \sum_{d\ge 0} \HF_R(d) t^d.
\]
We are interested in computing individual terms $\HF_R(d)$. This
can be done by computing the rank of the Macaulay
matrix at degree $d$ by taking $m$ generators of the ideal $\cI$ (see
Appendix~\ref{app: GB}). An upper bound on its cost can therefore be derived directly from \cite[Proposition 1]{BFS15}:
\begin{restatable}{proposition}{propcomplexity} \label{prop:complexity}
Let $F=\{f_1,\dots,f_m\}\subset\K[z_1,\dots,z_n]$ be a homogeneous system. Let $\cI$ be the corrresponding ideal. The term $\HF_{R}(d)$ of degree $d$  of the 
Hilbert function of $R=\K[\zv]/\cI$ can be computed in time
 bounded by
\[
\OO{md\binom{n+d-1}{d}^\omega},
\]
where $\omega$ is the linear algebra exponent.
\end{restatable}

Fortunately, the Hilbert function for our Pfaffian ideal is known.
We define the quotient ring
\[
R(\Mm)=\Fqm[\mv] / \cP_2(\Mm).
\] The Hilbert function (or equivalently the Hilbert series) of $R(\Mm)$ is well-known: 
\begin{proposition}[{\cite[(from) Theorem 1]{GK04}}] \label{prop: HS_pfaffian}
Let $\Mm=(m_{i,j})_{i,j}$ be the generic $s \times s$ skew-symmetric matrix over $\F$. Then $\dim \Vm(\cP_2(\Mm))=2s-3$ and
\[
\HF_{R(\Mm)}(d) = \binom{s+d-2}{d}^2-\binom{s+d-2}{d+1}\binom{s+d-2}{d-1},
\]
\[
\HS_{R(\Mm)}(z) = \frac{\sum_{d=0}^{s-3} \left(\binom{s-2}{d}^2-\binom{s-3}{d-1}\binom{s-1}{d+1}\right)z^d}{(1-z)^{2s-3}}.
\]
\end{proposition}
The term corresponding to $\HF_{R(\Mm)}(d)$ can also be rewritten as a Narayana number:
\[
\HF_{R(\Mm)}(d) = \frac{1}{s+d-1}\binom{s+d-1}{d+1}\binom{s+d-1}{d}.
\]
 Modeling \ref{mod:Pfaffian} adds linear equations to it expressing the fact that the matrix should also be in the 
matrix code of quadratic relations. There is one handy tool that allows to compute the Hilbert series obtained by enriching with polynomials an ideal whose Hilbert series is known. 
\begin{proposition}[{\cite[Lemma 3.3.2]{B04}}]\label{B04}
As long as there are no reductions to 0 in the F5 algorithm, the Hilbert function $\HF_{\K[\xv]/\ideal{f_1,\dots,f_m}}(d)$ satisfies the following recursive formula:
\[
\HF_{\K[\xv]/\cI(f_1,\dots,f_m)}(d)=\HF_{\K[\xv]/\ideal{f_1,\dots,f_{m-1}}}(d)-\HF_{\K[\xv]/\ideal{f_1,\dots,f_{m-1}}}(d-d_m)
\]
where $d_m=\deg(f_m)$.
\end{proposition}
Essentially, reductions to $0$ in F5 correspond to ``non generic''
reductions to $0$ and experimentally we have not observed this
behavior for Modeling \ref{mod:Pfaffian} when we add the linear
equations expressing that $\Mm$ belongs to
the matrix code $\Cmat$ of relations associated to a random linear
code.

\subsection{Analysis of the Hilbert series for the Pfaffian ideal}
We will from now on consider that the matrix code $\Cmat$ of quadratic
relations is associated to a code $\CC$ over $\Fqm$ of parameters
$[n,mr]$ which are the same as those of the extended dual code
$\ext{\Alt{r}{\xv}{\yv}^\perp}{\Fqm}$ of an alternant code
$\Alt{r}{\xv}{\yv}^\perp$ of length $n$ over $\Fq$ and extension
degree $m$ which we assume to be of generic dimension $k=n-mr$.  We
will from now on also assume that the $[n,mr]$ code $\CC$ we consider
satisfies
\begin{equation}\label{eq:hypsquare}
\dim \sq{\CC} = n.
\end{equation}
Equivalently, we suppose that the code is not square distinguishable
and will look for another and more powerful distinguisher. This
corresponds to the generic case of a random code as soon as
$\binom{rm+1}{2} \geq n$ and to duals of alternant codes/Goppa codes
that are not square--distinguishable.  Recall that, from
Proposition~\ref{prop:dimension},
$$
\dim_{\Fqm} \Cmat(\cV)=\binom{mr+1}{2}-\dim_{\F} \sq{\CC}=
\binom{mr}{2}+mr - n = \binom{mr}{2}-k,
$$
where $k \eqdef n - rm$ is given above and corresponds to the dimension of the
alternant code we are interested in.  Notice that $k$ is also the
cardinality of the set of independent linear equations expressing in
Modeling \ref{mod:Pfaffian} that the $rm \times rm$ matrix $\Mm$
belongs to $\Cmat$ since $\binom{rm}{2}-\dim \Cmat = k$. 
We are now going to show that the Hilbert function
of the ring $\Fqm[\mv]/(\cP(\Mm)+\langle L_i\rangle_i)$ differs
starting from some degree $\bar{d}$ depending on how the linear
relations $L_i$'s are defined (coming from $\Cmat$ associated to a
random $\CC$ or to the extended dual of an alternant or Goppa
code). We will assume that the parameters of our matrix code are such
that we do not expect a matrix or rank $2$ when $\CC$ is random, which
according to Proposition \ref{prop: GV} holds as soon as $n > 3rm -3$,
{\em i.e} essentially for $k/n > 2/3$.

\subsubsection{Random case.}

We assume that there are no reductions to $0$ in F5 and that we can
apply Proposition \ref{B04}
\begin{align*}
\HF_{\K[\zv]/\left(\cI+\ideal{L_1, \dots,L_{\ell}}\right)}(d)&=\HF_{\K[\zv]/\left(\cI+\ideal{L_1, \dots,L_{\ell-1}}\right)}(d)-\HF_{\K[\zv]/\left(\cI+\ideal{L_1, \dots,L_{\ell-1}} \right)}(d-1)\\
&= \ \cdots\\
&=\HF_{\K[\zv]/\cI}(d)-\sum_{i=0}^{\ell-1} \HF_{\K[\zv]/\left(\cI+\ideal{L_1, \dots,L_i} \right)}(d-1),
\end{align*}
which, by induction, leads to
\[
\HF_{\K[\zv]/\left(\cI+\ideal{L_1, \dots,L_{\ell}}\right)}(d)=\sum_{i=0}^{d} (-1)^i \binom{\ell}{i}\HF_{\K[\zv]/\cI}(d-i).
\]
This holds as long as there are no reductions to $0$ in F5. When there
are, we expect that the Hilbert series at this degree is zero, which
means that the induction formula should be
\[
\HF_{\K[\zv]/\left(\cI+\ideal{f}\right)}(d)=\max(\HF_{\K[\zv]/\cI }(d)-\HF_{\K[\zv]/\cI}(d-\bar{d}), 0 ).
\]
This leads to the following conjecture, experimentally supported.
\begin{conjec}[Random case] \label{conj: HF_random}
Let $L_1,\dots,L_k$ be the $k=n-rm$ linear relations relative to the matrix code $\Cmat$ associated to a random $[n,rm]$-code as above. Let  $\cP_2^+(\Mm)\eqdef \cP_2(\Mm)+\ideal{L_1, \dots,L_k }$. If $\HF_{\F[\mv]/\cP_2^+(\Mm)}(d')>0$ for all $d'<d$,
then
\begin{align}
\HF_{\F[\mv]/\cP_2^+(\Mm)}(d)&=\max\left(0,\sum_{i=0}^{d} (-1)^i \binom{k}{i}\HF_{\F[\mv]/\cP_2(\Mm)}(d-i)\right) \nonumber \\
&=\max\left(0, \sum_{i=0}^{d} \frac{(-1)^i}{rm+d-i-1}\binom{k}{i} \binom{rm+d-i-1}{d-i+1}\binom{rm+d-i-1}{d-i}\right).
\label{eq:HF_random}
\end{align}
Otherwise $\HF_{\F[\mv]/\cP_2^+(\Mm)}(d)=0$.
\end{conjec}
Because we assume that Modeling \ref{mod:Pfaffian} has only zero for solution in the case of a random code, there exists a $d$ such that $\HF_{\F[\mv]/\cP_2^+(\Mm)}(d)=0$.
Experiments (see Appendix~\ref{sec:HF_exps})  lead to conjecture the following behavior:
\begin{conjec} \label{conj: HF_asymptotic}
Let $\Cmat$ be the matrix code of relations originated by a random $[n,rm]$ code as above.
 Let $\cP_2^+(\Mm)$ the corresponding Pfaffian ideal and $\dreg = \min\{d : HF_{\F[\mv]/\cP_2^+(\Mm)}(d)=0\}$. Then
\[
 \dreg \sim c\frac{(rm)^2}{n-rm}
\]
for a constant $c$ equal or close to $\frac{1}{4}$.
\end{conjec}
The value $\dreg$ is known in the literature as the \textbf{degree of regularity}.
\subsubsection*{Alternant/Goppa case.}

In the alternant/Goppa case however the Hilbert series never vanishes because the variety of solutions has always positive dimension. We can even lower its dimension by a rather large
quantity. 
\begin{proposition}
Let $\Cmat$ be the matrix code of quadratic relations corresponding to the extended dual of an $[n,n-rm]$ binary Goppa code with a square-free Goppa polynomial. Let $\cP_2^+(\Mm)$ be the corresponding Pfaffian ideal. Then $\dim \Vm(\cP_2^+(\Mm))\ge 2r-3$.
\end{proposition}
\begin{proof}
Consider the matrix space $\DC$ of all skew-symmetric matrices that are 0 outside the top-left $r\times r $ diagonal block and let $\Mm'$ be the generic matrix in this space. We have that $\dim \Vm(\cP_2^+(\Mm'))=\dim \Vm(\cP_2(\Nm))$, where $\Nm$ is the generic skew-symmetric matrix of size $r\times r$. We recall from Proposition~\ref{prop: HS_pfaffian} that the dimension of the variety of the generic Pfaffian ideal $\cP_2(\Nm)$ is $2r-3$. Proposition \ref{prop: rel_binGoppa} states that $\Cmat(\cA)$ contains the subspace of block-diagonal skew-symmetric matrices (with $r\times r$ blocks). This implies $\DC\subseteq \Cmat(\cA)$ and thus
\[
\dim\Vm(\cP_2^+(\Mm))\ge \dim \Vm(\cP_2(\Mm'))=\dim \Vm(\cP_2(\Nm))=2r-3.
\]
\qed
\end{proof}
More in general, we can upper bound the dimension of the variety using the following proposition, whose proof is given in Appendix~\ref{sec:dim_var}.
\begin{restatable}{proposition}{dimvar} \label{prop : dimvar}
Let $\Cmat$ be the matrix code of quadratic relations corresponding to the extended dual of an $[n,n-rm]$ alternant code over a field of even characteristic. Let $\cP_2^+(\Mm)$ be the corresponding Pfaffian ideal. Then $\dim \Vm(\cP_2^+(\Mm))\ge r-2$.
\end{restatable}
\begin{remark}
  Equalities in the two previous propositions were met in the
  experiments we performed. Note that, comparing with
  Proposition~\ref{prop: rank2_subspaces}, the Pfaffian ideal contains
  subspaces of dimension roughly half the dimension of the variety.
\end{remark}
Now, as a consequence of the variety not being trivial, we have
\begin{proposition}\label{prop: HF>0}
		Let $\Cmat$ be the matrix code of quadratic relations corresponding to the extended dual of an $[n,n-rm]$ alternant code. Let $\cP_2^+(\Mm)$ be the corresponding Pfaffian ideal. For all $d\in \N$,
	$\HF_{\F[\mv]/\cP_2^+(\Mm)}(d)>0$.
\end{proposition}
\begin{proof}
	Assume by contradiction that $\exists d\in \N$ such that $\HF_{\F[\mv]/\cP_2^+(\Mm)}(d)=0$. Therefore
	\[
	\dim_{\Fqm} (\cP_2^+(\Mm))_d= \dim_{\Fqm} \Fqm[\mv]_d,
	\]
	\ie all the monomials of degree $d$ belong to $\cP_2^+(\Mm)$, in
    particular all the monomials $m_{i,j}^d$. This implies that the
    only element in the variety of $\cP_2^+(\Mm)$ is the zero matrix
    (with some multiplicity). This is in
    contradiction with the existence of rank 2 matrices in $\Cmat$
    that must therefore be solutions of the Pfaffian system. \qed
	\end{proof}
Computing the Hilbert function up to some degree $d$ provides a distinguisher as soon as it assumes a different value depending on whether it refers to random or alternant/Goppa codes. Thanks to Proposition~\ref{prop: HF>0}, this will happen at the latest at the degree of regularity $\dreg$ corresponding to a random code. 

\subsubsection{An extension of the distinguisher of \cite{FGOPT11}.}
All these considerations lead to a very simple distinguisher of alternant or more specifically of Goppa codes, we compute for a code $\HF_{\Fqm[\mv]/\cP_2^+(\Mm)}(d)$ at a certain degree (where $\cP_2^+(\Mm)$ is the associated Pfaffian ideal), and say that it does not behave like a random code if this Hilbert function evaluated  at degree $d$ does not coincide with the formula we expect from a random code which is given in Conjecture \ref{conj: HF_random}. This leads us to the following definition
\begin{definition}[$d$-distinguishable] An $[n,rm]$ $\Fqm$-linear code $\CC$ is said to be $d$-distinguishable from a generic $[n,rm]$ linear code over $\Fqm$ when the following holds
$$\HF_{\Fqm[\mv]/\cP_2^+(\Mm)}(d) \neq \max\left(0, \sum_{i=0}^{d} \frac{(-1)^i}{rm+d-i-1}\binom{n-rm}{i} \binom{rm+d-i-1}{d-i+1}\binom{rm+d-i-1}{d-i}\right)$$
where $\cP_2^+(\Mm)$ is the Pfaffian ideal associated to $\CC$.
\end{definition}
Note that in general
\[\HF_{\Fqm[\mv]/\cP_2^+(\Mm)}(1)=\dim_{\Fqm} \Cmat(\cB).
\]
Hence, a different evaluation of the Hilbert function in degree 1 witnesses an unusually large dimension of $\Cmat(\cB)$ and consequently an atypically small dimension of the square code. Indeed, this corresponds to the square distinguisher from \cite{FGOPT11}.
Being $1$-distinguishable is therefore being square-distinguishable. In this sense, this new distinguisher generalizes the square-distinguisher of \cite{FGOPT11}. 

We can readily find examples of codes which are not square-distinguishable (or what is the same $1$-distinguishable), but are distinguishable for higher values of $d$. For instance, Table \ref{table: HF(2)_example} gives examples of generic alternant codes which are not $1$-distinguishable for the lengths $n \leq 124$ but which are $2$-distinguishable in the range $n \in \Iintv{76}{256}$.
Goppa codes are for the same parameters not $1$-distinguishable as soon as $n \leq 96$, but are distinguishable in the range  
$n \in \Iintv{75}{256}$. Note that in the same range we can even distinguish a generic alternant code from a Goppa code.
As was the case for the square distinguisher of \cite{FGOPT11}, Goppa codes are easier to distinguish from random codes 
than generic alternant codes. This also holds for our new distinguisher. We give in Table \ref{table: HF(2)_example_binary} an example of this kind. The binary Goppa codes in this table are $2$-distinguishable in the length range 
$n \in \Iintv{59}{64}$, whereas the generic alternant codes are not distinguishable at all. Note that none of the examples
in this table are square-distinguishable.

\begin{table}
  \center
	\begin{tabular}{|c||c|c|c|c|c|c|}
	\hline
		$\HF_{\Fqm[\mv]/\cP_2^+(\Mm)}(2)$ & $256\ge n \ge 77$ & $n=76$ &  $n=75$ & $n=74$ & $n=73$ & $\dots$ \\ \hline \hline
		Random code & 0 & 10 & 71 & 133 & 196 & $\dots$ \\ \hline		
		Alternant code & \textbf{20} & \textbf{20} & 71 & 133 & 196 & $\dots$\\ \hline
		Goppa code & \textbf{80} & \textbf{80} & \textbf{80} & 133 & 196 & $\dots$\\ \hline
	\end{tabular} 
	\caption{Hilbert function at degree 2 with respect to random, alternant and Goppa codes with parameters $q=4, m=4, r=4$. The evaluations in bold correspond to distinguishable lengths. } \label{table: HF(2)_example}
\end{table}
\begin{table}
  \center
	\begin{tabular}{|c||c|c|c|c|c|c|c|c|}
	\hline
		$\HF_{\Fqm[\mv]/\cP_2^+(\Mm)}(2)$ & $n=64$ & $n=63$ &  $n=62$ & $n=61$ & $n=60$ & $n=59$ & $n=58$ &  $\dots$ \\ \hline \hline
		Random code & 2718 & 2826 & 2935 & 3045 & 3156 & 3268 & 3381 & $\dots$ \\ \hline		
		Alternant code & 2718 & 2826 & 2935 & 3045 & 3156 & 3268 & 3381 & $\dots$ \\ \hline	
		Goppa code & \textbf{2971} & \textbf{2971} & \textbf{2971} & \textbf{3048} & \textbf{3158} & \textbf{3269} & 3381 & $\dots$\\ \hline
	\end{tabular}
	\caption{Hilbert function at degree 2 with respect to random, alternant and Goppa codes with parameters $q=2, m=6, r=3$. The evaluations in bold correspond to distinguishable lengths.}  \label{table: HF(2)_example_binary}
\end{table}

For the time being, we have only a limited understanding of how 
 $\HF_{\F[\mv]/\cP_2^+(\Mm)}(d)$ behaves for alternant/Goppa codes. However in the case of binary square-free Goppa code, \ie those used in McEliece's schemes, we can significantly improve upon 
 the $\HF_{\F[\mv]/\cP_2^+(\Mm)}(d)>0$ lower bound as shown by

\begin{restatable}{thm}{lowerboundHFBG} \label{thm : lowerbound_HF_binGoppa}
Let $\Goppa{\xv}{\Gamma}$ be a non distinguishable binary $[n,k=n-rm]$ Goppa code with $\Gamma$ a square-free polynomial of degree $r$ and extension degree $m$. Let $\cP_2^+(\Mm)$ be the corresponding Pfaffian ideal.
  Then, for all $d>0$,
	\[
	\HF_{\F_{2^m}[\mv]/\cP_2^+(\Mm)}(d) \ge m\left( \binom{r+d-2}{d}^2-\binom{r+d-2}{d+1}\binom{r+d-2}{d-1}\right).
	\]
\end{restatable}
The proof is given in Appendix~\ref{sec:app_modeling}.
Theorem~\ref{thm : lowerbound_HF_binGoppa} has some theoretical interest, because it shows that the distinguisher can be further improved by analyzing the matrix code of relations obtained from a Goppa code. 
\subsection{Complexity of computing the distinguisher and comparison with known key and message attacks}

\subsubsection*{Complexity of computing the distinguisher.}
The complexity of computing the distinguisher is upper-bounded by using Proposition \ref{prop:complexity}
\begin{proposition}
The computation of $HF_{\Fqm[\mv]/\cP_2^+(\Mm)}(d)$ for the Pfaffian ideal associated to an $[n,mr]$-code has complexity
\[
\OO{d \left(n-rm+\binom{rm}{4}\right)\binom{\binom{rm}{2}+d-1}{d}^\omega},
\]
where $\omega$ is the linear algebra exponent.
\end{proposition}
\begin{proof}
This proposition follows on the spot from Proposition \ref{prop:complexity}, since the number of variables is $\binom{rm}{2}$ (the number of independent entries in a skew symmetric of size $rm$), the number of independent linear equations is $n-rm$ and the
number of quadratic equations is $\binom{rm}{4}$. \qed
\end{proof}

However, in the case at hand, we can use Wiedemann's algorithm, because 
 (i) we know the Hilbert function for the Pfaffian ideal associated to an $[n,mr]$ random code, and know when it is equal to $0$, namely for $d=\dreg$
(ii) we only have to check whether at degree $d=\dreg$ the Macaulay matrix  $\Mac(F,\dreg)$ has a non zero kernel, (iii) this Macaulay matrix is sparse, since
the Pfaffian equations contain only 3 quadratic monomials, and therefore the number of entries in a row of $\Mac(F,\dreg)$ is upper-bounded with the number of nonzero entries of the polynomial $\mv^\alpha L(\mv)$, where $\mv$ is the variable vector of the matrix entries, $L=0$ is one of the $k$ linear equations and $\alpha$ is a multi-index exponent of multi-degree $\dreg-1$. This quantity clearly coincides with the number of nonzero entries of $L$ itself and can be upper bounded by $\binom{rm}{2}-k+1$ thanks to Gaussian elimination. Therefore the complexity of the sparse linear algebra approach becomes by using Wiedemann's algorithm
\begin{proposition}\label{prop:comp_wiedemann}
Checking whether a code is an alternant code or a generic linear code can be performed with a complexity upper-bounded by
\[
\OO{\left(\binom{rm}{2}-k+1\right)\binom{\binom{rm}{2}+\dreg-1}{\dreg}^2}.
\]
\end{proposition}

\subsubsection*{Complexity of the standard approach for key recovery.}

Recall that it consists in guessing the irreducible Goppa polynomial $\Gamma$ and the support set of coordinates. After that, the Support Splitting Algorithm (SSA) \cite{S00} checks whether the public code is permutation equivalent to the guessed Goppa code. The cost of the SSA on the Goppa code $\CC$ has been estimated with
\[
\OO{n^3+q^h n^2 \log(n)},
\]
where $n$ is the code length and $h\eqdef \dim(\CC \cap \CC^\perp)$. Despite being exponential in the hull dimension, the latter is typically trivial or it has a very small dimension. Therefore the permutation equivalence code verification usually boils down to a polynomial-time subroutine and we ignore its cost in the comparison. The number of possible support coordinate sets is $\binom{q^m}{n}$, while the number of irreducible degree-$r$ polynomials over $\Fqm$ is given by 
\[
\frac{1}{r}\sum_{a\mid r} \mu(a)(q^m)^{\frac{r}{a}},
\]
where $\mu$ is the M\"obius function. Therefore the total complexity of this approach can be estimated as
\begin{equation}\label{eq:keyattack}
\OO{\frac{\binom{q^m}{n}}{r}\sum_{a\mid r} \mu(a)(q^m)^{\frac{r}{a}}}.
\end{equation}

\subsubsection*{Comparison of distinguisher with the key-attack.}
The comparison of all the methods we have just presented is given in Table~\ref{table: cost_comparison} with respect to Classic McEliece parameters. We remark that, using sparse linear algebra, we can improve upon the classical method for all parameters except those for category 5. Note that in this case the Goppa code is full-support and therefore the support coordinates do not need to be guessed, leading to a big improvement upon non-full support instances. However, our distinguisher suffers less than the standard key-recovery algorithm from taking instances that are not full support. Indeed, if we consider the same $r$ and $m$ used in Category 5, but a smaller length $n$, then our distinguisher approach outperforms the previous one. In fact, this can be seen directly from Category 3, which shares the same $r$ and $m$ with Category 5, but it is not full support.

\begin{table} 
	\hspace{-3.5cm} \begin{tabular}{|c||c|c|c|c|c|c|c|c|}
		\hline
		Category & $n$ & $r$ & $m$ & $\dreg$ & $R$ &  \parbox[t]{4cm}{classical key-recovery \\ $\mathbb{C}=\frac{\binom{2^m}{n}}{r}\sum_{a\mid r} \mu(a)(2^m)^{\frac{r}{a}}$} & \parbox[t]{5cm}{ dense linear algebra\\ $\mathbb{C}=\binom{rm}{4}\dreg\binom{\binom{rm}{2}-k+\dreg-1}{ \dreg}^\omega$ }  &  \parbox[t]{5cm}{sparse linear algebra\\ $\mathbb{C}=3(\binom{rm}{2}-k+1)\binom{\binom{rm}{2}+\dreg-1}{\dreg}^2$ } \\ \hline \hline
		1 & 3488 & 64 & 12 & 84 & 0.7798 & $2^{2476}\cdot 2^{762}=2^{3238}$ & $2^{3141}$ & $2^{2231}$ \\ \hline
		2 & 4608 & 96 & 13 & 212 & 0.7292 & $2^{8093}\cdot 2^{1241}=2^{9334}$ & $2^{7931}$ & $2^{5643}$ \\ \hline
		3 & 6688 & 128 & 13 & 229 & 0.7512 & $2^{5629}\cdot 2^{1657}=2^{7286}$ & $2^{9030}$ & $2^{6425}$ \\ \hline
		4 & 6960 & 119 & 13 & 169 & 0.7777 & $2^{4997}\cdot 2^{1540}=2^{6537}$ & $2^{6779}$ & $2^{4822}$ \\ \hline
		5 & 8192 & 128 & 13 & 154 & 0.7969 & $2^{0}\cdot 2^{1657}=2^{1657}$ & $2^{6329}$ & $2^{4501}$ \\ \hline
	\end{tabular}
	\caption{Computational cost comparison between \textit{this} distinguisher and retrieving the permutation equivalence} \label{table: cost_comparison}
\end{table}

We also remark that our distinguishing modeling works for any alternant code, while the classical key-recovery procedure described here is specific for Goppa codes. Indeed, guessing a valid pair of support $\xv$ and multiplier $\yv$ for a generic alternant code is dramatically more costly for two reasons. First of all, the $n$ multiplier coordinates $y_i$'s are independent and do not have a compact representation through a degree-$r$ polynomial. Moreover, in order to guess a correct code permutation, the support and multiplier coordinate indexes must correspond.

In Figure~\ref{figure: r_grows} we show the growth of the degree of regularity $\dreg$ for a random $[n=2^m, n-rm]$ code, for fixed $m$. The graph is defined on the integer interval whose endpoints are given by the smallest value of $r$ for which \cite{FGOPT11} is not able to distinguish a binary Goppa code and the largest value for which this new modeling is able to distinguish respectively. Note that in this case the rate is decreasing. On the other hand, Figure~\ref{figure: fixed_R} provides the degree of regularity $\dreg$ and the complexity estimate using sparse linear algebra, for $m$ fixed, $r$ growing and $n=5rm$, \ie for the fixed rate $R=4/5$. The domain of the graph is computed in the same way as for Figure~\ref{figure: r_grows}.

\begin{figure}[h]
	\centering
	\subfloat[][\centering $\dreg$]{{\includegraphics[width=5.5cm]{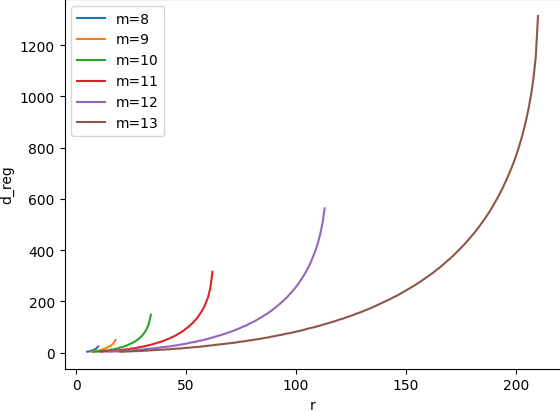} }}	\qquad
	\subfloat[][\centering complexity (logarithmic scale)]{{\includegraphics[width=5.5cm]{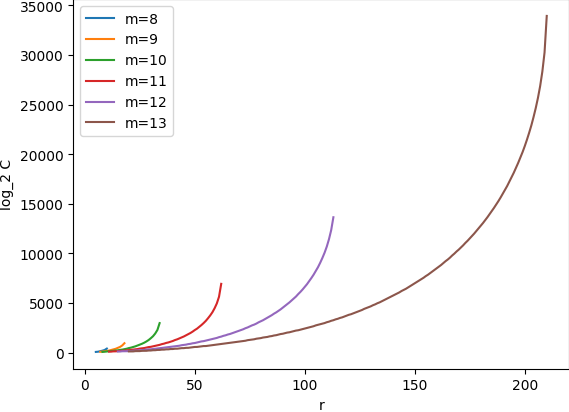} }}	\caption{Growth of the degree regularity in function of $r$ for fixed $m$} \label{figure: r_grows}
\end{figure}

\begin{figure}[h]	\centering
	\subfloat[][\centering $\dreg$]{{\includegraphics[width=5.5cm]{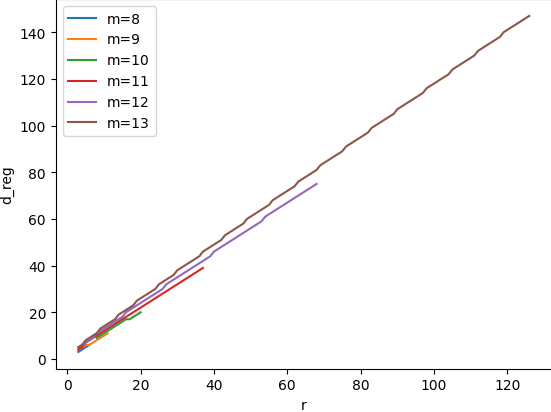} }}	\qquad
	\subfloat[][\centering complexity (logarithmic scale)]{{\includegraphics[width=5.5cm]{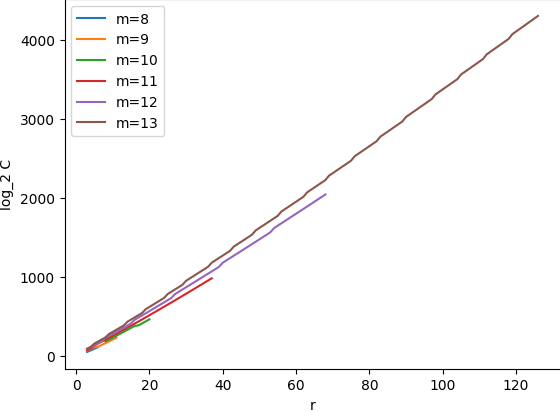} }}	\caption{Degree of regularity and complexity cost with respect to sparse linear algebra for the fixed rate $R=4/5$}	\label{figure: fixed_R}\end{figure}

\subsubsection*{Sublinear regime.}
 It is insightful to study the asymptotic complexity of distinguishing an $[n,rm]$-code in the sublinear regime, when the dimension $rm$
 is sublinear in the codelength $n$ and to compare it with key and message attacks. Assume that $rm = \Th{n^\alpha}$ where $\alpha \in [\frac{1}{2},1)$. 
We will also be interested in the case where the code is a binary  Goppa code. 
To simplify a little bit the discussion and to minimize the complexity of the known key attack, we will assume 
that we have a Goppa code of full support, i.e. $n=2^m$. 

A binary Goppa code of length $n$, extension degree $m$ and degree $r$ allows to correct $r$ errors. 
Because the number of errors to decode is sublinear in the codelength, the complexity $\Cmess$ of message attacks for binary $[n,n-mr]$ Goppa codes (namely that of decoding $r$ errors in an $[n,n-mr]$ code)  is of the form $2^{-r \log_2(1-R)(1+o(1))}$ for the best known generic decoding algorithms by \cite{CS16} where $R$ is the code rate, i.e. $R=\frac{n-mr}{n}$. 
We clearly have $\log_2(\Cmess)=(1-\alpha)rm(1+o(1))$ since $-\log_2(1-R)= -\log_2\left( \frac{rm}{n}\right)=(1-\alpha)\log n(1+o(1))$. 

On the other hand, the complexity $\Ckey$ of key attacks is of the form $\OO{2^{rm(1+o(1))}}$ in the full support case. Here we have $\log_2(\Ckey)=rm(1+o(1))$. Our distinguisher has complexity 
$\Cdist$ which can be estimated through Proposition \ref{prop:comp_wiedemann} and $\dreg$ by Conjecture \ref{conj: HF_asymptotic}, from which we readily obtain that 
$\log_2(\Cdist) = 4 \alpha c \frac{(rm)^{2}}{n} \log n(1+o(1))$, where $c$ is the constant appearing in Conjecture \ref{conj: HF_asymptotic}. This whole discussion is summarized in Table 
\ref{tab:complexity_sublinear}. The complexity of key attacks is bigger than the complexity of message attacks, however now asymptotically the complexity of the distinguisher is {\em significantly lower} 
than both attacks: message attacks gain a constant factor $1- \alpha$ in the exponent when compared to key attacks, whereas the distinguisher gains a {\em polynomial factor}
$\Th{\frac{rm}{n} \log n} =o(1)$ in the exponent with respect to both key and message attacks.

\begin{table}
\centering
\begin{tabular}{|c|c|c|c|}
\hline 
type & Key attack& Message attack&  distinguisher \\
\hline
$\log_2 C$&  $rm(1+o(1))$ & $(1-\alpha)rm(1+o(1)) $  & $4 \alpha c \frac{(rm)^{2}}{n} \log n(1+o(1))$ \\
\hline
\end{tabular}
\caption{Logarithm of the complexity $C$ of different attacks for full support $n=2^m$ binary $[n,n-mr]$ Goppa codes in the sublinear codimension regime
$rm = \Th{n^\alpha}$, where  $\alpha \in [\frac{1}{2},1)$. \label{tab:complexity_sublinear}   }
\end{table}

 \section{An attack on distinguishable random alternant codes, without the use of Gr\"obner bases}
\label{sec:attack}

We are going to present now a polynomial time attack on
square-distinguishable generic alternant codes defined over $\Fq$ as
soon as the degree $r$ satisfies $r < q+1$ by using this new notion of
the matrix code of quadratic relations. We also recall that a
square-distinguishable alternant code must have degree $r\geq 3$
\cite{FGOPT11}. If we combine this together with the filtration
technique of \cite{BMT23} which allows to compute from a
square-distinguishable alternant code of degree $r$ satisfying
$r \geq q+1$ an alternant code with the same support but of degree
$r-1$ we obtain an attack on all square-distinguishable generic
alternant codes. This is a big improvement on the attack presented in
\cite{BMT23} which needed two conditions to hold (1) a
square-distinguishable alternant code (2) $q$ is either $2$ or
$3$. Moreover \cite{BMT23} could not handle the subcase where the
alternant code is actually a Goppa code, whereas our new attack is
able to treat this case at least in the case $r < q-1$. We present in
Table \ref{table:summary} a summary of the attacks. In other words,
all square-distinguishable generic alternant codes can now be
attacked. The reason why for the time being the square-distinguishable
Goppa codes are out of reach, is that the filtration technique of
\cite{BMT23} for reducing the degree of the code does not work for the
special case of Goppa codes.

\begin{table}
  \center
\begin{tabular}{|c|c|c|c|}
\hline 
code & technique/paper & $r(\geq 3)$ & $q$ \\
\hline 
(generic) square-distinguishable alternant code & \cite{BMT23} & any & $\in \{2,3\}$ \\
  \hline
(generic) square-distinguishable alternant code & this paper & $< q+1$ & any \\
\hline
(generic) square-distinguishable alternant code & this paper + filtration techn. of \cite{BMT23} & any & any\\
\hline
square-distinguishable Goppa codes & this paper & $<q-1$ & any  \\
\hline
\end{tabular}
\caption{Summary of the attacks against square-distinguishable codes \label{table:summary}. The column $q$ corresponds to the restrictions on $q$ for the attack to work and the column $r$ has the same meaning for the parameter $r$.}
\end{table}

Thus, from now on, we will consider an alternant code
$\Alt{r}{\xv}{\yv}\subseteq \Fq^n$ of extension degree $m$ which is
such that $r<q+1$. For generic alternant codes, this corresponds to
the square-distinguisher case with $e=0$. If instead the alternant
code is also Goppa, then we restrict ourselves to the case of
$r<q-1$. We will show now how to recover $\xv$ and $\yv$ from the
knowledge of a generator matrix of this code by making use of the
matrix code of quadratic relations associated to the extended dual
code over $\Fqm$.

\subsection*{The idea}
We first present the underlying idea by picking the canonical basis $\cA$ \eqref{eq:basis_choice} and the parity-check matrix $\Hm_{\cA}$ of $\ext{\Alt{r}{\xv}{\yv}}{\Fqm}$ whose rows correspond to the elements of $\cA$ in that same order. Recall that this basis can be written as
$$
\cA = (\av_1,\cdots,\av_{r},\av_1^{q},\cdots,\av_r^q,\cdots,\av_1^{q^{m-1}},\cdots,\av_r^{q^{m-1}}).
$$ 
We also assume $q$ is odd for now. 
The crucial point is that, with the assumption of a square-distinguishable generic alternant code (resp. Goppa code) with $r<q+1$ (resp. $r<q-1$), the analysis provided in \cite{FGOPT11} implies that the matrix code is generated by \textit{all and only} relations of the kind
$$
\starp{\yv^{q^l} \xv^{aq^l}}{\yv^{q^l} \xv^{bq^l}}=\starp{\yv^{q^l} \xv^{cq^l}}{\yv^{q^l} \xv^{dq^l}}  
$$ 
where $l$ is arbitrary in $\Iintv{0}{m-1}$ and $a,b,c,d$ in $\Iintv{0}{r-1}$ such that $a+b=c+d$.  This corresponds to the quadratic relation 
$$\starp{\av_{a+1}^{q^l}}{\av_{b+1}^{q^l}}-\starp{\av_{c+1}^{q^l}}{\av_{d+1}^{q^l}}=0.$$
 The related code of relations $\Cmat(\cA)$ has therefore a block diagonal structure with blocks of size $r$, \ie, for each element in $\Cmat(\cA)$, the entries outside the $m$ diagonal blocks of size $r\times r$ are 0. Thus, an element $\Am$ of $\Cmat(\cA)$ has the following block shape:

\begin{equation} \label{eq: A_diagblock}  
\Am= \begin{bmatrix}
\Am_{0,0} &  &  & \\
	 & \Am_{1,1} & & \zerov\\
	 \zerov & & \ddots & \\
	  &  &  & \Am_{m-1,m-1} \\
\end{bmatrix}
\end{equation}
where the diagonal blocks $\Am_{i,i}$ are symmetric and of size $r$. Clearly $\rank(\Am_{i,i})\le r$ and, because of the block diagonal shape, $\rank(\Am)=\sum_i \rank(\Am_{i,i})$. Now assume that $\Am$ happens to be minimally rank defective, i.e.
\[
\rank(\Am) =rm-1.\]
 It means that for exactly one index $j\in \Iintv{0}{m-1}$, $\rank(\Am_{j,j})=r-1$, and for all $i \in \Iintv{0}{m-1}\setminus \{j\}$, $\rank(\Am_{i,i})=r$. We consider the left kernel of (the map corresponding to) the matrix $\Am$, simply denoted by $\ker(\Am)$. Note that, if we identify row vectors with column vectors, left and right kernels are the same in this case, as $\Am$ is symmetric. Since $\rank(\Am)=rm-1$, we have $\dim(\ker(\Am))=1$. Let $\vv=(\vv_0,\dots,\vv_{m-1}) \in \Fqm^{rm}$ be a generator of $\ker(\Am)$, with $\vv_i \in \Fqm^r$. Because of the block diagonal structure of $\Am$, $\vv$ must satisfy
\[
\vv=(\zerov_r,\dots,\zerov_r,\vv_j,\zerov_r,\dots,\zerov_r).
\] 
In other words, the computation of this nullspace provides information
about the position of the vectors generating a single GRS code
$\GRS{r}{\xv^{q^j}}{\yv^{q^j}}$. The key idea is that if enough of
such vectors are found, a basis of the corresponding GRS code can be
retrieved.

\subsection{Choosing $\cB$ with a special shape}

Consider an ordered basis
\begin{equation} \label{eq: AltBasisFrobenius}
  \mathcal{B} = (\bv_1, \dots, \bv_r, \bv_1^q, \dots, \bv_r^q, \dots, \bv_1^{q^{m-1}},\dots, \bv_r^{q^{m-1}})
\end{equation}
of $\ext{\Alt{r}{\xv}{\yv}^\perp}{ \Fqm}$. Such a basis can be computed by drawing $\bv_1, \dots, \bv_r \in \ext{\Alt{r}{\xv}{\yv}^\perp}{\Fqm}$ at
random, applying the Frobenius map $m-1$ times and checking if the obtained family generates $\ext{\Alt{r}{\xv}{\yv}^\perp}{\Fqm}$, or equivalently if its dimension is $rm$.
If not, draw another $r$-tuple $\bv_1, \dots, \bv_r$ at random until the
construction provides a basis. We remark that even sampling a basis as in \eqref{eq: AltBasisFrobenius} does not provide a basis with the same properties of $\cA$, \ie $(\bv_1,\dots,\bv_r)$ is not an ordered basis of $\GRS{r}{\xv}{\yv}$, except with negligible probability.

When $\cB$ is chosen as in \eqref{eq: AltBasisFrobenius}, the transition matrix $\Pm$ has a special shape.
\begin{lemma}
  The matrix $\Pm$ is blockwise Dickson. That is to say, there
  exist $\Pm_0, \dots, \Pm_{m-1} \in \Fqm^{r \times r}$ such that
  \begin{equation} \label{eq: matDickson}
    \Pm =
    \begin{pmatrix}
      \Pm_0 & \Pm_1 & \cdots & \Pm_{m-1} \\
      \Pm_{m-1}^{(q)} & \Pm_0^{(q)} & \cdots & \Pm_{m-2}^{(q)}\\
      \vdots & \vdots & \ddots & \vdots \\
       \Pm_{1}^{(q^{m-1})} & \Pm_2^{(q^{m-1})} & \cdots & \Pm_{0}^{(q^{m-1})}\\
    \end{pmatrix}.
  \end{equation}
\end{lemma}

\begin{proof}
  This is a direct consequence of the structure of the bases $\mathcal A$
  and $\mathcal B$. \qed
\end{proof}
Let $\Sm\in \mathbf{GL}_{mr}(\Fqm)$ be the right $r$-cyclic shift matrix, \ie
\begin{equation} \label{eq: matF}
  \Sm \eqdef
  \begin{pmatrix}
   & \mat{I}_r &  & & \\
    && \mat{I}_r & &\zerov \\
    &&\zerov & \ddots & \\
    && & & \mat{I}_r\\
     \mat{I}_r &  &  && \\
  \end{pmatrix}.
\end{equation}
Note that $\Sm^{-1}=\trsp{\Sm}$ is the left $r$-cyclic shift matrix. The block-wise Dickson structure of $\Pm$ can be re-interpreted
as follows:
\begin{proposition} \label{prop: P^q}
Let $\Sm$ be defined as in \eqref{eq: matF} and $\Pm$ satisfy the blockwise Dickson structure of \eqref{eq: matDickson}. Then
$  \Pm = \trsp{\Sm} \Pm^{(q)} \Sm$.
\end{proposition}
\begin{proof}
Direct computation. \qed
\end{proof}
The following result will also be used frequently in what follows

\begin{restatable}{proposition}{propstablefrobenius} \label{prop: stable_frobenius}
Whenever a basis $\cB$ has the form given in \eqref{eq: AltBasisFrobenius}, 
  $\Cmat (\mathcal B)$ is stable
  by the operation
  \[
    \Mm \longmapsto \trsp{\Sm} \Mm^{(q)} \Sm.
  \]
\end{restatable}
The proof is given in Appendix~\ref{sec:app-attack}.
Note that $\Sm^{(q^i)}=\Sm$ for any $i$. Therefore, by applying $i$ times the map $\Mm \longmapsto \trsp{\Sm} \Mm^{(q)} \Sm$, we obtain
$\Mm \longmapsto   (\trsp{\Sm})^i \Mm^{(q^i)} (\Sm)^i.$
We say that $\Mm$ and $(\trsp{\Sm})^i \Mm^{(q^i)} (\Sm)^i$ are {\em blockwise Dickson shift} of each other.

\subsection{The full algorithm with respect to a public basis $\cB$}
Algorithm~\ref{alg: attack} provides a sketch of the attack in the
case of odd chacteristic field size. We will then justify why this
algorithm is supposed to work with non-negligible probability,
elaborate on some subroutines (as sampling matrices of rank $rm-1$)
and adapt it to the even characteristic case.
\begin{algorithm}[t]
    \hspace*{\algorithmicindent} \textbf{Input:} (a basis of) an alternant code $\Alt{r}{\xv}{\yv}$\\
    \hspace*{\algorithmicindent} \textbf{Output:} a pair $(\xv', \yv')$ of support and multiplier for $\Alt{r}{\xv}{\yv}$ 
\begin{algorithmic}[1]
\State Choose a basis $\cB=(\bv_1, \dots, \bv_r, \bv_1^q, \dots, \bv_r^q, \dots, \bv_1^{q^{m-1}},\dots, \bv_r^{q^{m-1}})$ for $\ext{\Alt{r}{\xv}{\yv}^\perp}{ \Fqm}$.
\State $\SC_{aux}\gets \{0\}$
 \Repeat 
\State Sample $\Bm\in \Cmat(\cB)$ of rank $rm-1$ at random
\State $\vv \gets \text{ generator of }\ker(\Bm)$ \label{row:v}
\State $\SC_{aux} \gets \SC_{aux}+\Fqmspan{\vv,\vv^{q}\Sm,\dots, \vv^{q^{m-1}}\Sm^{m-1}}$ \label{row:at_once}
\Until $\dim_{\Fqm} \SC_{aux}=(r-1)m$
\State Sample $\Bm_1\in \Cmat(\cB)$ of rank $rm-1$ at random 
\State $\uv_1 \gets \text{ generator of }\ker(\Bm_1)$ \label{row:u1}
\State $\VC \gets \langle \uv_1 \rangle$
\For{$j \in \Iintv{2}{r}$} 
\State Sample $\Bm_j\in \Cmat(\cB)$ of rank $rm-1$ at random 
\State $\uv_j \gets \text{ generator of }\ker(\Bm_j)$ \label{row:uj}
\Repeat 
\State $\uv_j\gets \uv_j^{q}\Sm$
\Until $\dim_{\Fqm} \SC_{aux}+\langle \uv_1,\uv_j \rangle=(r-1)m+1$
 \State $\VC \gets \VC+ \langle \uv_j \rangle$
 \EndFor \label{row:V}
 \State $\DC \gets \VC^\perp$
 \State $\GC \gets \DC$
 \For{$j \in \Iintv{1}{m-2}$} 
 \State $\DC \gets \DC^{(q)}\Sm$
 \State $\GC \gets \GC \cap \DC$
 \EndFor
 \State Apply the Sidelnikov-Shestakov attack \cite{SS92}
 on $\GC\cdot  \Hm_{\cB}$
  \State Return the support-multiplier pair $(\xv', \yv')$ found from Sidelnikov-Shestakov attack
   \end{algorithmic} \caption{Sketch of the attack in odd characteristic}\label{alg: attack}
\end{algorithm}
We now show the
structure of the attack. Starting form a public basis, compute a basis
as in (\ref{eq: AltBasisFrobenius}), {\em i.e.}, a basis of the
following form
\[
\mathcal{B} = (\bv_1, \dots, \bv_r, \bv_1^q, \dots, \bv_r^q, \dots, \bv_1^{q^{m-1}},\dots, \bv_r^{q^{m-1}}).
\]
How to compute such a basis has already been explained in a previous
section.  Similarly to $\Hm_{\cA}$, we  define
 $\Hm_{\cB}$ as the parity-check matrix of $\ext{\Alt{r}{\xv}{\yv}}{\Fqm}$ whose rows correspond to the elements of $\cB$ in that same order. 
 The correctness of the whole algorithm follows immediately
from the following propositions whose proofs can be found in
Appendix~\ref{sec:app_modeling}.  The first one explains why when we
have one kernel element in Algorithm \ref{alg: attack} at line
\ref{row:at_once} we can find $m-1$ other ones.
\begin{restatable}{proposition}{propatonce}\label{prop:at_once}
Let $\vv$ be in the kernel of a matrix $\Bm$ in $\Cmat(\cB)$ of rank $rm-1$. Then 
$\vv^{q}\Sm,\dots, \vv^{q^{m-1}}\Sm^{m-1}$ are $m-1$ elements that are also kernel elements of matrices in 
$\Cmat(\cB)$ of rank $rm-1$ which are respectively $\trsp{\Sm} \Bm^{(q)} \Sm,\cdots$, $(\trsp{\Sm})^{m-1} \Bm^{(q^{m-1})} \Sm^{m-1}$.
\end{restatable}
Then we are going to give a description of the space $\VC$ produced in
line \ref{row:V}. Basically this a vector space of elements that
correspond to a similar GRS code, in the following sense.
\begin{definition}
Let $\cA,\cB$ be the two bases introduced before and $\Pm$ the change of basis, \ie $\Hm_{\cB}=\Pm \Hm_{\cA}$. Let $\uv_1,\uv_2 \in \Fqm^{rm}$ be two vectors such that
\[
  \forall t \in \{1,2\}, \quad
\uv_t  \trsp{(\Pm^{-1})} \Pm^{-1} \Hm_{\cB} \in \GRS{r}{\xv}{\yv}^{q^{j_t}}
\]  
for some values $j_t\in \Iintv{0}{m-1}$. We say that $\uv_1$ and $\uv_2$ \textbf{correspond to the same GRS code with respect to the basis $\cB$} if and only if $j_1=j_2$.
\end{definition}
Two vectors $\uv_1$ and $\uv_2$ obtained by computing the nullspaces of rank $rm-1$ matrices may or may not correspond to the same GRS code. In any case, from them, we can easily exhibit two vectors corresponding to the same GRS code by choosing among their shifts $\uv_t^{q^i}\Sm^{i}$. More precisely, we have
\begin{restatable}{proposition}{propcorrGRS}\label{prop:corrGRS}
Let $\cA,\cB$ be the two bases introduced before and $\Pm$ the change of basis, \ie $\Hm_{\cB}=\Pm \Hm_{\cA}$. Let $\uv_1,\uv_2 \in \Fqm^{rm}$ be two vectors such that
\[
  \forall t \in \{1,2\},\quad
  \uv_t  \trsp{(\Pm^{-1})} \Pm^{-1} \Hm_{\cB} \in \GRS{r}{\xv}{\yv}^{(q^{j_t})}
\]  
for some values $j_t\in \Iintv{0}{m-1}$. There exists a unique $l\in \Iintv{0}{m-1}$ such that $\uv_1$ and $\uv_2^{q^l}\Sm^{l}$ correspond to the same GRS code.
\end{restatable}
To detect which shift of $\uv_2$ corresponds to the same GRS code of $\uv_1$, we rely on the following proposition.
\begin{restatable}{proposition}{propSCaux} \label{prop: SC_aux}
Let $\vv_1,\dots,\vv_{r-1},\uv_1,\uv_2\in \Fqm^{rm}$ be the generators of the kernels of $\Bm_1,\dots,\Bm_{r-1},\Bm',\Bm''\in \Cmat(\cB)$ respectively, for randomly sampled matrices of rank $rm-1$. Define 
\[
	\SC_{aux} \eqdef \Fqmspan{\vv_j^{q^l}\Sm^l \mid j \in \Iintv{1}{r-1}, l \in \Iintv{0}{m-1}}.
\] If the following conditions are satisfied:
\begin{itemize}
\item $\dim_{\Fqm} \SC_{aux} =(r-1)m$ \quad (\ie the $(r-1)m$ vectors that generate $\SC_{aux}$ are linearly independent);
\item $\dim_{\Fqm} \SC_{aux}+\Fqmspan{\uv_t} =(r-1)m+1$, $\quad t=1,2$;
\end{itemize}
then the two following statements are equivalent:
\begin{enumerate}
\item $\dim_{\Fqm} \SC_{aux}+\Fqmspan{\uv_1,\uv_2^{q^l} \Sm^l} =(r-1)m+1$;
\item $\uv_1$ and $ \uv_2^{q^l} \Sm^l$ correspond to the same GRS code with respect to $\cB$.
\end{enumerate}
\end{restatable}
We are therefore able to construct a space of dimension $r$ whose elements all correspond to a same GRS code. Then we use
\begin{restatable}{proposition}{propVperp} \label{prop:Vperp}
  Let $j \in \Iintv{0}{m-1}$.
Let $\VC_j$ be the $[rm, r]$ linear code generated by $r$ linearly independent vectors corresponding to the same GRS code $\GRS{r}{\xv}{\yv}^{(q^j)}$ with respect to $\cB$. Then the linear space $\VC_j^\perp$ orthogonal to $\VC_j$ is such that
\begin{equation} \label{eq:Vperp}
\VC_j^\perp \Hm_{\cB}=  \sum_{i \in \Iintv{0}{m-1}\setminus \{j\}} \GRS{r}{\xv}{\yv}^{(q^i)}.
\end{equation}
\end{restatable}
Given $\VC_j^\perp$, the other codes $\VC_i^\perp \Hm_{\cB}$ that are sums of $m-1$ GRS codes can be obtained according to the the following chain of equalities
\begin{align*}
&\sum_{i \in \Iintv{0}{m-1}\setminus \{j+l \mod m\}} \GRS{r}{\xv}{\yv}^{(q^i)}\\
=&\left(\sum_{i \in \Iintv{0}{m-1}\setminus \{j\}} \GRS{r}{\xv}{\yv}^{(q^i)}\right)^{(q^l)}=(\VC_j^\perp \Hm_{\cB})^{(q^l)}=(\VC_j^\perp)^{(q^l)} \Hm_{\cB}^{(q^l)}=(\VC_j^\perp)^{(q^l)} \Sm \Hm_{\cB}.
\end{align*}

After this, we are ready to compute a basis of a GRS code.
\begin{restatable}{proposition}{propcapVperp} \label{prop:capVperp}
Let $\VC_j^\perp$ be a linear space satisfying Equation~\eqref{eq:Vperp}, for all $j\in \Iintv{0}{m-1}$. Then with the standard assumption that all $\GRS{r}{\xv}{\yv}^{(q^j)}$ are in direct sum, we obtain, for any $j\in \Iintv{0}{m-1}$,
\[
\GRS{r}{\xv}{\yv}^{(q^j)}= \bigcap_{i \in \Iintv{0}{m-1}\setminus \{j\}} \VC_i^\perp \Hm_{\cB}.
\]
\end{restatable}

 \begin{remark}
 In the $q$ odd case, the only exception to what was said until now occurs for $r=3$. In this case a non-full rank diagonal block $\Bm_{j,j}$ becomes the null block, because there are no matrices of rank 1 or 2. In this case, the kernel of a rank $r(m-1)=3m-3$ matrix is a three-dimensional subspace, which immediately provides the subspace $\VC_j$ from which to recover the associated GRS codes.
 \end{remark}

\subsection*{How to sample matrices in $\Cmat(\cB)$ of rank $rm-1$}
This is the most costly part of the algorithm. We first address the case of odd characteristic, as the case of even characteristic needs an ad hoc discussion. 
It is not too difficult to estimate that the
density of rank $rm-1$ matrices inside $\Cmat(\cB)$ is of order $q^{-m}$ (see \ref{ss:app-estimate_rank}) and therefore it is desirable to have a better technique than just a brute force approach.
 We proceed instead as follows.
We take two matrices $\Dm_1,\Dm_2$ at random in $\Cmat(\cB)$ and solve over $\Fqm$ the equation
\[
\det(w\Dm_1+ \Dm_2)=0.
\]
The determinant $\det(w\Dm_1+ \Dm_2)$ is a univariate polynomial of degree $rm$ and since $w$ is taken over $\Fqm$ we can expect to have solutions with non-negligible probability. A root $w_0$ of $\det(w\Dm_1+ \Dm_2)$ determines a matrix $w_0 \Dm_1 + \Dm_2$ whose rank is strictly smaller than $rm$ but not necessarily equal to $rm-1$. However, 
the rank $rm-1$ is by far the most likely outcome.
 Repeating the process enough times ($\Theta(1)$ times on average) then provides a matrix of rank $rm-1$.

\subsection{Complexity}
The bottelneck of the attack is the
computation of rank $rm-1$ matrices in $\Cmat(\cB)$ which is explained
in the previous paragraph. The computation of the polynomial
$\det (w \Dm_1 + \Dm_2)$ can be done by choosing $rm$ distinct
elements $\alpha_1, \dots, \alpha_{rm}$ of $\Fqm$, compute the values
$\det (\alpha_1 \Dm_1 + \Dm_2), \dots, \det (\alpha_{rm} \Dm_1 +
\Dm_2)$ and then recover the polynomial $\det (w \Dm_1 + \Dm_2)$ by
interpolation.  This represents the calculation of
$rm = \mathcal{O}(n)$ determinants of $rm \times rm$ matrices and
hence a cost \( O(n^{\omega + 1}), \) where $\omega$ is the complexity
exponent of linear algebra. Once this polynomial (in the variable $w$)
is computed, the cost of the root--finding step is negligible compared
to that of the previous calculation.

Since the latter process should be repeated $\mathcal{O}(n)$ times, we get an
overall complexity of
\[
  \mathcal{O}(n^{\omega + 2})\  \text{operations\ in\ } \Fqm.
\]

\subsection{Even characteristic}
This case is treated in
Appendices~\ref{ss:app_attack_char_2},~\ref{ss:app-qeven-rodd}
and~\ref{ss:app-Vj}.

 \section{Conclusion}
\label{sec:conclusion}
\subsubsection*{A general methodology for studying the security of the McEliece cryptosystem with respect to key--recovery attacks.}
Trying to find an attack on the key of the McEliece scheme based on
Goppa codes,
has turned out over the years to be a formidable problem. The progress
on this issue has basically been non existent for many years and it
was for a long time judged that the McEliece scheme was immune against
this kind of attacks. This changed a little bit when many variants of
the original McEliece came out, either by turning to a slightly larger
class of codes namely the alternant codes which retain the main
algebraic structure of the Goppa code and/or adding additional
structure on it \cite{BCGO09,BBBCDGGHKNNPR17}, changing the alphabet
\cite{BLP10,BLP11a}, or going to extreme parameters \cite{CFS01}. This
has lead to devise many tools to attack these variants such as
algebraic modeling to recover the alternant stucture of a Goppa code
which is basically enough to recover its structure \cite{FOPT10},
using square code considerations \cite{COT14,COT17,BC18}, or trying to
solve a simpler problem which is to distinguish these algebraic codes
from random codes \cite{FGOPT11,FGOPT13,MT22}.  We actually believe
that in order to make further progress on this very hard problem, it
is desirable to move away now from studying particular schemes
proposed in the literature, by exploring and developing systematically
tools for solving this problem and study the region of parameters
(alphabet size $q$, code length $n$, degree $r$ of the code, extension
degree $m$) where these methods work.  We suggest the following
research plan
\begin{itemize}
\item Studying the slightly more general problem of attacking alternant codes might be the right way to go because it retains the essential algebraic features of Goppa codes and it allows to find attacks that might not work in the subcase of Goppa codes where the additional structure can be a nuisance. An example which is particularly enlightening here is the recent work \cite{BMT23} (attack on generic alternant codes in a certain parameter regime which amazingly does not work in the 
particular case of Goppa codes where the additional structure prevents the attack to work).
\item A particularly fruitful research thread is to study the potentially easier problem of finding a distinguisher for alternant/Goppa codes first. 
\item Turn later on this distinguisher into an attack (such as \cite{BMT23} for the distinguisher of \cite{FGOPT11}).
\end{itemize}
This is the research plan we have followed to some extent here.

\subsubsection*{A distinguisher in odd characteristic.} It is clear that any algebraic modeling for solving the symmetric MinRank problem for rank $3$ could be used to attack the problem in odd characteristic. The Support Minors modeling of \cite{BBCGPSTV20} would be for instance a good candidate for this. The difficulty is here to predict the complexity of system solving, since the fact that the matrices are symmetric gives many new linear dependencies that do not happen in the generic MinRank case. This is clearly a promising open problem.

\subsubsection*{Turning the distinguisher of \S \ref{sec:modeling} into an attack.}
The Pfaffian modeling for the distinguisher can be used in principle to attack the key-recovery problem as well. This problem is strictly harder than just distinguishing because of the algebraic structure in the code $\Cmat(\cA)$ that is much stronger than in $\Cmat(\cR)$ (random case). In particular, rank 2 matrices are found at a potentially larger degree than $\bar{d}$ at which the Hilbert function in the random case becomes 0. 
The fact that the solution space is very large, in particular it contains a rather large vector space (see Section \ref{sec:low_rank}), suggests though that we can safely specialize a rather large number of variables to speed up the system solving. Once a rank $2$ matrix is found, the attack is not finished yet, but it is tempting to conjecture that the main bottleneck is to find such a matrix first and that some of the tools developed in the attack given in Section~\ref{sec:attack} might be used to finish the job.

Indeed, since rank 2 matrices in $\Cmat(\cA)$ are identically zero outside the main block diagonal, we can consider a matrix subcode spanned by many of them, obtained by solving the Pfaffian system with different specializations. This subcode will have a block diagonal shape and that is why the attack of the last section is expected to apply on such subspace.

\newpage
\appendix

\section{Gr\"obner bases and Computing the Hilbert Series} \label{app: GB}

Gr\"obner basis techniques are the main tool at hand to solve multivariate polynomial systems and therefore to perform algebraic cryptanalysis. One crucial notion for this kind of computation and for the complexity analysis is the Macaulay matrix \cite{M16}. We give the definition that is relevant in the homogeneous case.

\begin{definition}[Macaulay Matrix \cite{M16}]
	Let $F=\{f_1,\dots,f_m\}\subset\K[\xv]$ be a homogeneous system such that $\deg(f_i)=d_i$. Let $d$ be a positive integer. The (homogeneous) Macaulay matrix $\Mac(F,d)$ of $F$ in degree $d$ is a matrix whose rows are each indexed by a polynomial $m_j f_i$, for any $f_i\in F$ and any monomial $m_j$ of degree $ d-d_i$, and whose columns are indexed by all the monomials of degree $d$. The entry corresponding to the row indexed by $m_j f_i$ and column indexed by $m_l$ is the coefficient of $m_l$ in $m_j f_i$. In particular, if $m_j f_i=\sum_{\alpha \in \N^n} a_\alpha \xv^\alpha$ and $m_l=\xv^\beta$, then the corresponding entry of $\Mac(F,d)$ is $a_\beta$:
	\[
	\begin{array}{lcc}
		& & \textcolor{gray}{m_l} \\
		\Mac(F,d)= & \textcolor{gray}{m_j f_i} & \begin{bmatrix}
			& \vdots & \\ \cdots & a_\beta & \cdots \\ & \vdots &
		\end{bmatrix}.
	\end{array}
	\] 
\end{definition}

A Gr\"obner basis can be computed using linear algebra, in particular \cite{L83} showed that it is enough to perform Gaussian elimination on a Macaulay matrix in degree equal to the degree of regularity $d_{reg}$. Several algorithms and variants are linear-algebra based, for instance F4 \cite{F99}, F5 \cite{F02} or XL \cite{CKPS00}. Differently from methods not exploiting the Macaulay matrix construction, this approach allows to derive complexity estimates for this task.
Computing the Hilbert series can also be done with these methods and this is the only result we need here, which is derived from \cite[Proposition 1]{BFS15} as
\propcomplexity*
Furthermore, these methods can take advantage of algorithms that benefit from matrix sparsity \cite{W86},\cite{CCNY12}. The cost of the XL Wiedemann algorithm to solve a Macaulay matrix in degree $d$ has been evaluated \cite[Proposition 3, p. 219]{DY09} with
\[
3 n_r\binom{n+d-1}{d}^2,
\]
where $n_r$ is the average weight of a row in $\Mac(F,d)$. \section{Proof related to Section \ref{sec:quadratic_form}}
\label{sec:app-quadratic}
Let us recall the proposition we prove.
\propcongr*
\begin{proof}
Let $\cA$ and $\cB$ be related as $\Hm_{\cB}=\Pm\Hm_{\cA}.$ where $\Hm_{\cA}$, resp. $\Hm_{\cB}$ is a matrix whose rows are the basis elements of $\cA$, resp. $\cB$. 
It will be helpful to view an element $\cv=(c_{i,j})_{1 \leq i \leq j \leq k}$ of $\Crel$ as a matrix $\Cm=(C_{i,j})_{\substack{1 \leq i \leq k\\1 \leq j \leq k}}$ where
$C_{ij}=c_{ij}$ for $i \leq j$ and $C_{ij}=0$ otherwise. We can write the matrix $\Mm_{\cv}$ of $\Cmat$ corresponding to $\cv$ as 
$\Mm_{\cv} = \Cm + \trsp{\Cm}$.
Consider an element $\Mm \in\Cmat(\cB)$. By definition of $\Cmat(\cB)$ there is an element $\cv=(c_{i,j})_{1 \leq i \leq j \leq k}$ of $\Crel(\cB)$ such 
that $\Mm=\Mm_{\cv}$. Consider the  matrix $\Cm$ corresponding to $\cv$ that we just introduced.
By definition of $\Crel(\cB)$ we have
\begin{equation}
\label{eq:cij}
 \sum_{1\le i\le j \le k} c_{i,j} \starp{\bv_i}{\bv_j}=0.
 \end{equation}
 We have for all $i$ in $\Iintv{1}{k}$:
 $
 \bv_i = \sum_{s=1}^k p_{is} \av_s,
 $
 where $p_{i,j}$ denotes the entry $(i,j)$ of $\Pm$. Therefore
 \begin{eqnarray*}
  \sum_{1\le i\le j \le k} c_{i,j} \starp{\bv_i}{\bv_j} &= & \sum_{1\le i\le j \le k} c_{i,j} \starp{\left(\sum_{s \in \Iintv{1}{k}} p_{i,s}\av_s\right)}{\left(\sum_{t \in \Iintv{1}{k}} p_{j,t}\av_t\right)}\\
   &= & \sum_{s,t \in \Iintv{1}{k}} \left(\sum_{1\le i\le j \le k} p_{i,s}  p_{j,t} c_{i,j}\right)\starp{\av_s}{\av_t} \\
   & = & \sum_{1\le s< t \le k} \left(\sum_{1\le i\le j \le k} (p_{i,s}  p_{j,t}+p_{i,t}  p_{j,s}) c_{i,j}\right)\starp{\av_s}{\av_t} \\
&&+\sum_{s \in \Iintv{1}{k}} \left(\sum_{1\le i\le j \le k} p_{i,s}  p_{j,s} c_{i,j}\right)\starp{\av_s}{\av_s} \\
 \end{eqnarray*}
 Let $\Dm = (d_{s,t})_{\substack{1 \leq s \leq k\\1 \leq t \leq k}}$ where
 \begin{eqnarray*}
 d_{s,t} & \eqdef &\sum_{1\le i\le j \le k} (p_{i,s}  p_{j,t}+p_{i,t}  p_{j,s}) c_{i,j}\;\;\text{ for $1 \leq s < t \leq k$}\\
 d_{s,s} & \eqdef & \sum_{1\le i\le j \le k} p_{i,s}  p_{j,s} c_{i,j} \;\;\text{ for $s \in \Iintv{1}{k}$}\\
d_{s,t} & \eqdef & 0 \;\;\text{ otherwise.}
 \end{eqnarray*}
 $\dv \eqdef (d_{i,j})_{1 \leq i \leq j \leq k}$ is because of \eqref{eq:cij} an element of $\Crel(\cA)$. 
 Now from the definition of $\Dm$ is clear that we have
 $\Dm + \trsp{\Dm} = \trsp{\Pm} \left( \Cm + \trsp{\Cm} \right) \Pm.$
 In other words, the matrix $\Mm_{\dv}$ in $\Cmat(\cA)$ corresponding to $\dv$ satisfies
 \begin{eqnarray*}
 \Mm_{\dv} & = & \Dm + \trsp{\Dm}\\
 & = & \trsp{\Pm} \left( \Cm + \trsp{\Cm} \right) \Pm\\
 & = & \trsp{\Pm} \Mm_{\cv}  \Pm.
 \end{eqnarray*}
 This holds for any $\cv$ in $\Crel(\cB)$. This leads to $\trsp{\Pm}\Cmat(\cB) \Pm \subseteq \Cmat(\cA)$.
Since $\Pm$ is invertible, this implies
$\Cmat(\cA)=\trsp{\Pm}\Cmat(\cB) \Pm$. \qed
\end{proof}
 \section{Proofs of some results given in Section \ref{sec:low_rank}}
\label{app-GV}
\subsection{Proofs of the results given in \S\ref{ss:low-rank}}
\label{ss:app-low-rank}

For all the proofs given here we recall that we have fixed the basis
$$\cA \eqdef \{\yv,\xv\yv,\dots,\xv^{r-1}\yv,\dots,\yv^{q^{m-1}},(\xv\yv)^{q^{m-1}},\dots,(\xv^{r-1}\yv)^{q^{m-1}}\}.$$
We will also consider the following block form of the matrices $\Mm \in \Cmat(\cB)$:
 \[
 \Mm_{\cv}=\begin{pmatrix} \Mm_{0,0} & \Mm_{0,1} & \dots & \Mm_{0,m-1} \\
 \Mm_{1,0} & \ddots & &  \\
 \vdots & & & \vdots \\
 \Mm_{m-1,0} & & \dots & \Mm_{m-1,m-1}
 \end{pmatrix},
 \]
 with $\Mm_{l,u}=(m^{(l,u)}_{i,j})_{\substack{0\le i\le r-1 \\ 0\le j\le r-1}}\in \Fqm^{r\times r}$. 

We are first going to prove Proposition \ref{prop: rel_binGoppa} which is given by
\proprelbinGoppa*
\begin{proof}
Recall from \cite{P75} that, if $\Goppa{\xv}{\Gamma}=\Alt{r}{\xv}{\yv}$ and $\Gamma$ is a square-free polynomial of degree $r$, then
\[
	 \Goppa{\xv}{\Gamma}=\Goppa{\xv}{\Gamma^2}=\Alt{2r}{\xv}{\yv^2}.
\]
Thus 
\[\xv^{i 2^l}\yv^{2^{(l+1 \mod m)}}\in \Goppa{\xv}{\Gamma}^\perp_{\Fqm},\]
 for all $i\in \Iintv{0}{2r-1}, l \in \Iintv{0}{m-1}$. Consequently each equation
 \begin{equation} \label{eq: rank2_squarefree}
 (\xv^a \yv)^{2^l}(\xv^b \yv)^{2^{l}}=(\xv^{a+b} \yv^2)^{2^{l-1}}(\xv^{a+b} \yv^2)^{2^{l-1}},
\end{equation}
with $l \in \Iintv{1}{m}, 0\le b < a < r$ corresponds to a codeword $\cv$ in $\Crel(\cA)$. Let us fix $(a,b,l)$. 
Since $(\xv^{a+b} \yv^2)^{2^{l-1}}(\xv^{a+b} \yv^2)^{2^{l-1}}$ is a square and the field characteristic is 2, the matrix $\Mm \in \Cmat(\cA)$ corresponding to the relation \eqref{eq: rank2_squarefree} is such that
 \[
\Mm_{u,v}=\zerov_{r\times r},\quad \text{if } (u,v)\ne (l,l)
\]
and
 \[
m^{(l,l)}_{i,j}= \begin{cases}
1&\text{if }(i,j)\in\{(a,b),(b,a)\}\\
0& \text{otherwise}
\end{cases},
\]
where $\Mm_{u,v}=(m^{(u,v)}_{i,j})\in \Fqm^{r\times r}$ is the block of $\Mm$ with row-column block index $(u,v)$.
Hence
\[
\rank(\Mm)=\rank(\Mm_{l,l})=2.
\]
It is trivial to check that the set of matrices obtained by any possible choice of $a,b$ and $l$ generates the space of all block-diagonal skew-symmetric matrices with $r\times r$ blocks. \qed
\end{proof}

Let us prove Proposition \ref{prop: rank2_subspaces} that we recall here
\propranktwosubspace*

\begin{proof}
	Let us consider the matrix subspace originated by choosing all the matrices corresponding to \eqref{eq: alt_rel} for a fixed $l=u$ and such that $c=d$, $a+b=2c$, $a$ and $b$ are even and one of them equals a fixed even value $j$ (alternatively one can choose $a,b$ both odd and one equal to an odd $j$). Any matrix $\Mm$ in this subspace is zero outside the union of the $(lr+j+1)$-th column and the $(lr+j+1)$-th row. Its rank is therefore upper bounded by 2. In other words, any such matrix $\Mm$ has the following shape
\begin{equation}
\Mm=\begin{bmatrix}
\zerov &  &  & & \\
 &  \ddots &  & \zerov & \\
	& & \Mm_{l,l} & & \\
	 & \zerov & & \ddots & \\
	  &  &  & & \zerov \\
\end{bmatrix}, \quad \text{with} \quad \Mm_{l,l}= \begin{blockarray}{ccccccccc}
\begin{block}{[cccccccc]c}
 \zerov &  & \makecell{\cellcolor{gray!30}*} & & & \zerov & & &\\
  &  & \makecell{\cellcolor{gray!30}0} & &  &  &  &  &\\
\makecell{\cellcolor{gray!30}*} & \makecell{\cellcolor{gray!30}0} & \makecell{\cellcolor{gray!50}0} & \makecell{\cellcolor{gray!30}0} & \makecell{\cellcolor{gray!30}*} & \makecell{\cellcolor{gray!30}0} & \makecell{\cellcolor{gray!30}*} & \makecell{\cellcolor{gray!30}0} &  \hspace{5pt} \leftarrow (j+1)\text{-th row} \\
   &  & \makecell{\cellcolor{gray!30}0} & &  &  &  & & \\
  &  & \makecell{\cellcolor{gray!30}*} &  &  &  &  & & \\
 \zerov &  & \makecell{\cellcolor{gray!30}0} & &  & \zerov &  & & \\
 &  & \makecell{\cellcolor{gray!30}*} &  &  &  &  &  &\\
    &  & \makecell{\cellcolor{gray!30}0} & &  &  &  & & \\
  \end{block}
  \end{blockarray},
\end{equation}	
where all the $*$'s in the $(j+1)$-th row of $\Mm_{l,l}$ can be chosen independently. 
Thus, the subspace dimension is $\floor{\frac{r-1}{2}}$, because each of the $\floor{\frac{r+1}{2}}$ odd entries of the $(j+1)$-th column of $\Mm_{l,l}$ is a $*$, with the exception of the $(j+1,j+1)$ entry, which is 0.

If $\Alt{r}{\xv}{\yv}$ is a Goppa code $\Goppa{\xv}{\Gamma}$, we
consider instead the matrix subspace originated by choosing all the
matrices corresponding to \eqref{eq: rank2_squarefree} for a fixed $l$
and such that one element among $a, b$ equals a fixed value $j$. Any
matrix $\Mm$ in this subspace is null outside the union of the
$(lr+j+1)$-th column and the $(lr+j+1)$-th row. Its rank is therefore
upper bounded by 2. In other words, any such matrix $\Mm$ has the
following shape
\begin{equation}
\Mm=\begin{bmatrix}
\zerov &  &  & & \\
 &  \ddots &  & \zerov & \\
	& & \Mm_{l,l} & & \\
	 & \zerov & & \ddots & \\
	  &  &  & & \zerov \\
\end{bmatrix}, \quad \text{with} \quad \Mm_{l,l}= \begin{blockarray}{ccccccccc}
\begin{block}{[cccccccc]c}
 \zerov &  & \makecell{\cellcolor{gray!30}*} & & & \zerov & & &\\
  &  & \makecell{\cellcolor{gray!30}*} & &  &  &  &  &\\
\makecell{\cellcolor{gray!30}*} & \makecell{\cellcolor{gray!30}*} & \makecell{\cellcolor{gray!50}0} & \makecell{\cellcolor{gray!30}*} & \makecell{\cellcolor{gray!30}*} & \makecell{\cellcolor{gray!30}*} & \makecell{\cellcolor{gray!30}*} & \makecell{\cellcolor{gray!30}*} &  \hspace{5pt} \leftarrow (j+1)\text{-th row} \\
   &  & \makecell{\cellcolor{gray!30}*} & &  &  &  & & \\
  &  & \makecell{\cellcolor{gray!30}*} &  &  &  &  & & \\
 \zerov &  & \makecell{\cellcolor{gray!30}*} & &  & \zerov &  & & \\
 &  & \makecell{\cellcolor{gray!30}*} &  &  &  &  &  &\\
    &  & \makecell{\cellcolor{gray!30}*} & &  &  &  & & \\
  \end{block}
  \end{blockarray},
\end{equation}	
where all the $*$'s in the $(j+1)$-th row of $\Mm_{l,l}$ can be chosen independently. Thus, the subspace dimension is $r-1$, because each of the $r$ entries of the $(j+1)$-th column of $\Mm_{l,l}$ is a $*$, with the exception of the $(j+1,j+1)$ entry, which is 0. \qed
\end{proof}

We will prove now Proposition \ref{prop:count-low-rank-alt} that we recall here
\propcountlowrankalt*
\begin{proof}
It directly follows from Lemmas~\ref{lemma: nrank2_rodd} and \ref{lemma: nrank2_reven} that we will give below. \qed
\end{proof}

To understand what is going on in this case, it is insightful to have a look at some examples first.
Let us fix a value $l=u \in \Iintv{0}{m-1}$ and consider the subspace of $\Cmat(\cA)$ spanned by all the matrices corresponding to a quadratic relation
\[\starp{(\xv^a \yv)^{q^l}}{(\xv^b \yv)^{q^l}}=\starp{(\xv^c \yv)^{q^l}}{(\xv^d \yv)^{q^l}},
\]
for any possible choice of $r-1\ge a>c\ge d>b\ge 0$. It follows from the analysis of the distinguisher in
\cite{FGOPT13} and \cite{MT22} that this space has dimension
$\binom{r-1}{2}$. Let $\Mm_{l,l}(\uv)$ be the generic diagonal block
matrix of such subspace, where $\uv=(u_1,\dots,u_{\binom{r-1}{2}})$ is
the vector of coefficients with respect to the basis. We give examples
for some small values of $r$.
\begin{example}
	\begin{itemize}
		\item For $r=3$:
		\begin{equation} \label{prop: countlowrankalt_r3}
		\Mm_{l,l}(\uv)=\begin{bmatrix}
			0 & 0& u_1 \\
			0 & 0& 0 \\
			u_1 & 0 & 0
		\end{bmatrix}.
		\end{equation}
		\item For $r=4$:
		\begin{equation} \label{prop: countlowrankalt_r4}
		\Mm_{l,l}(\uv)=\begin{bmatrix}
			0 & 0& u_1 & u_2 \\
			0 & 0& u_2 & u_3\\
			u_1 & u_2 & 0 & 0 \\
			u_2 & u_3 & 0 & 0
		\end{bmatrix}.
		\end{equation}
		\item For $r=5$:
		\begin{equation} \label{prop: countlowrankalt_r5}
		\Mm_{l,l}(\uv)=\begin{bmatrix}
			0 & 0& u_1 & u_2 & u_4 \\
			0 & 0& u_2 & u_3 & u_5\\
			u_1 & u_2 & 0 & u_5 & u_6 \\
			u_2 & u_3 & u_5 & 0 & 0 \\
			u_4 & u_5 & u_6 & 0 & 0 \\
		\end{bmatrix}.	
		\end{equation}
		\item For $r=6$:
		\begin{equation} \label{prop: countlowrankalt_r6}
		\Mm_{l,l}(\uv)=\begin{bmatrix}
			0 & 0& u_1 & u_2 & u_4 & u_7\\
			0 & 0& u_2 & u_3 & u_5+u_7 & u_8\\
			u_1 & u_2 & 0 & u_5 & u_6 & u_9 \\
			u_2 & u_3 & u_5 & 0 & u_9 & u_{10}\\
			u_4 & u_5+u_7 & u_6 & u_9 & 0 & 0\\
			u_7 & u_8 & u_9 & u_{10} & 0 & 0\\
		\end{bmatrix}.	
		\end{equation}
	\end{itemize}
\end{example}

Table~\ref{table: rank2} illustrates the experimental number of rank 2 matrices in $\Cmat(\cA)$ such that $\rank(M_{1,1})=2$ and all the other blocks are null, for small values of $r$ and over the field $\Fqm$.

\begin{table} 
	\begin{tabular}{ |c||c|c|c|c|c|c|  } 
		\hline
		Size $r\times r$& 3 & 4 & 5 & 6 & 7 & 8\\
		\hline
		n. of rank 2 matrices  & $q^m-1$ & $q^{2m}-1$ & $2q^{3m}-q^{2m}-1$ & $2q^{4m}-q^{2m}-1$ & $3q^{5m}-q^{4m}-q^{2m}-1$ & $3q^{6m}-q^{5m}-q^{2m}-1$\\
		\hline
	\end{tabular} 
	\caption{Number of rank-2 block matrices} \label{table: rank2}
\end{table}
Table~\ref{table: rank2} suggests that these blocks have a number of rank 2 specializations that roughly grows as $\floor{\frac{r-1}{2}}(q^m)^{r-2}$. We are now going to show the shape of a number of rank-2 matrices in the order of $(q^m)^{r-2}$. Despite not being all the rank-2 matrices, this is interesting in order to determine the dimension of the variety corresponding to a determinantal ideal, which indeed can be proved to be at least $r-2$.
The explanation can be split into odd and even matrix sizes.

From the matrices $\Mm_{l,l}(\uv)$ with odd size $r\times r$, by specializing some of the $\uv$ variables and selecting some row/column indexes, we can determine submatrices of size $\ceil{r/2}\times \ceil{r/2}$ that are skew-symmetric but without any other additional relation. Again, we first give examples for some small values of $r$.
\begin{example}
	\begin{itemize}
		\item For $r=3$, the submatrix of \eqref{prop: countlowrankalt_r3} obtained by taking row/column indexes in $\{1,3\}$ is
		\[
		\begin{bmatrix}
			0& u_1 \\
			u_1 & 0
		\end{bmatrix}.
		\]
		\item For $r=5$, the submatrix of \eqref{prop: countlowrankalt_r5} obtained by taking row/column indexes in $\{1,3,5\}$ is
		\[
		\begin{bmatrix}
			0 & u_1 & u_4 \\
			u_1 & 0 & u_6 \\
			u_4 & u_6 & 0 \\
		\end{bmatrix}.	
		\]
	\end{itemize}
\end{example}
More generally, it is enough to build the $\ceil{r/2}\times \ceil{r/2}=\frac{r+1}{2}\times \frac{r+1}{2}$ submatrix selecting the odd row/column indexes. 
By specializing all the $u_i$'s not appearing in the submatrix, this gives a lower bound on the number of matrices of rank 2. In particular

\begin{lemma}  \label{lemma: nrank2_rodd}
	The number of choices of $\uv$ for which $\Mm_{l,l}(\uv)$ ($r$ odd) has rank 2 is lower bounded by $N_0\left(\frac{r+1}{2},2\right)$, where 
	$N_0(s,r)$ stands for the number of skew-symmetric matrices of size $s$ over $\Fq$ and rank $r$.
\end{lemma} 
Since (see Proposition \ref{prop: n_sym_givenrank_ch2})
\[
N_0\left(\frac{r+1}{2},2\right)=(q^m)^{2\frac{r+1}{2}-3}+o((q^m)^{2\frac{r+1}{2}-3})=q^{m(r-2)}+o(q^{m(r-2)}),
\]
and we have $m$ blocks, we expect that the number of solutions is at least in the order of $mq^{m(r-2)}$.
 Analogously for $\Mm_{l,l}(\uv)$ matrices with even size, we do not construct generic skew-symmetric submatrices but we provide specializations of $\Mm_{l,l}(\uv)$ related to such submatrices. We first give the examples for some small values of $r$.
\begin{example}
	\begin{itemize}
		\item For $r=4$:
		\[
		\begin{bmatrix}
			0 & 0& u_1 & \lambda u_1 \\
			0 & 0& \lambda u_1 & \lambda^2 u_1\\
			u_1 & \lambda u_1 & 0 & 0 \\
			\lambda u_1 & \lambda^2 u_1 & 0 & 0
		\end{bmatrix},
		\]
		\ie we take in \eqref{prop: countlowrankalt_r4}the specialization
		\[
		\begin{cases}
		u_2=\lambda u_1\\ u_3 =\lambda^2 u_1
		\end{cases},
		\]
		with the parameter $\lambda \in \Fqm$.
		\item For $r=6$:
		\[
		\begin{bmatrix}
			0 & 0& u_1 & \lambda u_1 &u_4  & \lambda u_4 \\
			0 & 0& \lambda u_1 & \lambda^2 u_1 & \lambda u_4  & \lambda^2 u_4 \\
			u_1 & \lambda u_1 & 0 & 0 & u_6 & \lambda u_6  \\
			\lambda u_1 & \lambda^2 u_1 & 0 & 0 & \lambda u_6  & \lambda^2 u_6 \\
			u_4 & \lambda u_4 & u_6 & \lambda u_6  & 0 & 0\\
			\lambda u_4  & \lambda^2 u_4  & \lambda u_6  & \lambda^2 u_6  & 0 & 0\\
		\end{bmatrix},	
		\]
		\ie we take in \eqref{prop: countlowrankalt_r6}the specialization
		\[
		\begin{cases}
		u_2=\lambda u_1\\ u_3 =\lambda^2 u_1 \\u_5=0 \\ u_7=\lambda u_4 \\ u_8=\lambda ^2 u_4 \\ u_9 = \lambda u_6 \\ u_{10} = \lambda^2 u_6
		\end{cases},
		\]
		with the parameter $\lambda \in \Fqm$.
      \end{itemize}
\end{example}
More generally we can replace each entry $u_i$ of a generic anti-symmetric matrix of size $\frac{r}{2}\times\frac{r}{2}$ with the $2\times 2$ block $\begin{bmatrix} u_i & \lambda u_i \\ \lambda u_i & \lambda^2 u_i\end{bmatrix}$ and each null element of the diagonal with the null $2\times 2$ block. It is clear that if the starting $\frac{r}{2}\times\frac{r}{2}$ matrix has rank 2, then the same occurs for the $r\times r$ block matrix. Moreover, the variable $\lambda$ adds one degree of freedom. Hence we have
\begin{lemma} \label{lemma: nrank2_reven}
	The number of choices for $u_i$'s such that the specialized $r\times r$ matrix $W^{(\uv)}$ ($r$ even) has rank 2 is lower bounded by $q^m \cdot N_0(\frac{r}{2},2)$.
\end{lemma} 
Since 
\[
q^m\cdot N_0\left(\frac{r}{2},2\right)=(q^m)\cdot (q^m)^{2\frac{r}{2}-3}+o((q^m)\cdot (q^m)^{2\frac{r}{2}-3})=q^{m(r-2)}+o(q^{m(r-2)}),
\]
and we have $m$ blocks, we have proved that the number of rank-2 matrices in $\Cmat(\cA)$ is again at least in the order of $mq^{m(r-2)}$.

We are now going to prove a refinement of this counting for binary Goppa codes
\propcountlowrankGoppa*

\begin{proof}
	We have seen that each choice of $(a,b,l)$ from \eqref{eq: rank2_squarefree} leads to a different matrix $\Mm$ in $\Cmat(\cA)$ which is null outside the diagonal block $\Mm_{l,l}$ and such that $\rank(\Mm)=\rank(\Mm_{l,l})=2$. Furthermore, the block submatrix $\Mm_{l,l}$ is such that only one element below the diagonal is nonzero, \ie the entry $(a+1,b+1)$. Hence the set over all possible choices of $(a,b,l)$ of these matrices generates the full subspace of skew-symmetric block diagonal matrices. Therefore, by counting the rank-2 matrices in this subspace, the number of rank-2 matrices in $\Cmat(\cA)$ can be lower bounded by
	\[
m N_0(r,2)=m \frac{(q^{mr}-1)(q^{m(r-1)}-1)}{q^{2m}-1}.
\]
\qed
\end{proof}

\subsection{Proof of Proposition \ref{prop: GV}}
\label{ss:proof_prop: GV}

Let us first recall this proposition.
\propGV*

For this proof, we will need the following results giving the number of symmetric/skew-symmetric matrices of a given rank.
The number of symmetric matrices over a finite field of a given rank can be found in \cite{M69}.
\begin{proposition}[{\cite[Theorem 2]{M69a}}] \label{prop: n_sym_givenrank} 
	Let $N(t, r)$ denote the number of symmetric matrices of size $t\times t$, rank $r$, with entries in $\F_{q}$. Then
	\[
	N(t,2s)=\prod_{i=1}^{s}\frac{q^{2i}}{q^{2i}-1}\prod_{i=0}^{2s-1}(q^{t-i}-1),\quad 2s \le t
	\]
	\[
	N(t,2s+1)=\prod_{i=1}^{s}\frac{q^{2i}}{q^{2i}-1}\prod_{i=0}^{2s}(q^{t-i}-1),\quad 2s+1 \le t.
	\]	
\end{proposition}
When the field characteristic is 2, the number of skew-symmetric matrices has also been computed.
\begin{proposition}[{\cite[Theorem 3]{M69a}}]  \label{prop: n_sym_givenrank_ch2}  
	Let $N_0(t, r)$ denote the number of symmetric matrices of size $t\times t$, rank $r$, with entries in $\F_{q}$, $q=2^n$, and 0 on the main diagonal. Then
	\[
		N_0(t,2s)=\prod_{i=1}^{s}\frac{q^{2i-2}}{q^{2i}-1}\prod_{i=0}^{2s-1}(q^{t-i}-1),\quad 2s \le t
	\]
	\[
		N_0(t,2s+1)=0.
	\]	
\end{proposition}
\begin{remark}
Proposition~\ref{prop: n_sym_givenrank_ch2} implies that skew-symmetric matrices defined over a field of characteristic 2 have always even rank. 
\end{remark}

We are ready now to give a proof of Proposition \ref{prop: GV}.
\begin{proof}[of Proposition  \ref{prop: GV}]
For a random code $\RC$ with basis $\cR$, $\dim(\Cmat(\cR))=\binom{rm+1}{2}-n$ is expected with probability $1-o(1)$ when $\binom{rm+1}{2}>n$ \cite{CCMZ15}. From Propositions~\ref{prop: n_sym_givenrank},\ref{prop: n_sym_givenrank_ch2} we have respectively
\begin{align*}
\card{B_d^{(\Sym)}}\sim &  N(rm, d)\\
 = &\prod_{i=1}^{\floor{d/2}}\frac{(q^m)^{2i}}{(q^m)^{2i}-1}\prod_{i=0}^{d-1}((q^m)^{rm-i}-1)\\
 \sim &\prod_{i=0}^{d-1}(q^m)^{rm-i}\\
 = &(q^m)^{drm-\binom{d}{2}}.
 \end{align*} 
and (in characteristic 2)
\begin{align*}
\card{B_d^{(\Skew)}}\sim &  N_0(rm, 2\floor{d/2})\\
 = &\prod_{i=1}^{\floor{d/2}}\frac{(q^m)^{2i-2}}{(q^m)^{2i}-1}\prod_{i=0}^{2\floor{d/2}-1}((q^m)^{rm-i}-1)\\
 \sim &(q^m)^{-2\floor{d/2}}\prod_{i=0}^{d-1}(q^m)^{rm-i}\\
 = &(q^m)^{2\floor{d/2}rm-\binom{2\floor{d/2}+1}{2}}.
 \end{align*} 

Therefore, from Gilbert-Varshamov bounds \eqref{eq: GV_Sym},\eqref{eq: GV_Skew} we get that rank-$d$ matrices belong to $\Cmat(\cR)$ with non negligible probability iff
\begin{itemize}
\item (for symmetric matrices)
 \begin{align*}
 &(q^m)^{\binom{rm+1}{2}-n} (q^m)^{drm-\binom{d}{2}} \ge (q^m)^{\binom{rm+1}{2}}\\
 \iff & \binom{rm+1}{2}-n+drm-\binom{d}{2} \ge \binom{rm+1}{2}\\
 \iff & n\le drm-\binom{d}{2}.
 \end{align*}
 \item  (for skew-symmetric matrices in characteristic 2)
 \begin{align*}
 &(q^m)^{\binom{rm+1}{2}-n} (q^m)^{d-\binom{d+1}{2}} \ge (q^m)^{\binom{rm}{2}}\\
 \iff & \binom{rm+1}{2}-n+drm-\binom{d+1}{2} \ge \binom{rm}{2}\\
 \iff & n\le (d+1)rm-\binom{d+1}{2}.
 \end{align*}
 \end{itemize} \qed
\end{proof}

 \section{Proofs and experimental evidence corresponding to Section \ref{sec:modeling}}
\label{sec:app_modeling}

\subsection{Proof of Proposition \ref{prop  : dimvar}} \label{sec:dim_var}
Let us first recall the Proposition.
\dimvar*
\begin{proof}
	We recall from Proposition~\ref{prop: HS_pfaffian} that the dimension of the variety of the generic Pfaffian ideal $\cP_2(\Mm)$ is $2s-3$, where $s$ is the matrix size. The result follows from the construction given in Appendix \S\ref{ss:app-low-rank}, for estimating the number of rank 2 matrices, where we have shown that $\Cmat(\cA)$ contains subspaces of matrices that are isomorphic to the full space of skew-symmetric matrices for some smaller size. This allows to lower bound $\dim \Vm(\cP_2^+(\Mm))$ in terms of $\dim \Vm(\cP_2(\Nm))$, where $\Nm$ is the generic skew-symmetric matrix of smaller size. More precisely:
	\begin{itemize}
		\item if $r$ is odd: let $\Nm$ be the generic skew-symmetric matrix of size $\frac{r+1}{2}\times \frac{r+1}{2}$. Then the construction explained before Lemma~\ref{lemma: nrank2_rodd} in Appendix \S\ref{ss:app-low-rank} implies that
		\[
		\dim\Vm(\cP_2^+(\Mm))\ge \dim \Vm(\cP_2(\Nm))=2\frac{r+1}{2}-3=r-2;
		\]
		\item if $r$ is even: this is the most subtle case, because we do not construct generic skew-symmetric matrices. Let $\Nm$ be the generic skew-symmetric matrix of size $\frac{r}{2}\times \frac{r}{2}$ and $\Nm'$ be the skew-symmetric matrix of size $r\times r$ with indeterminates given as in the construction explained before Lemma~\ref{lemma: nrank2_reven} in Appendix \S\ref{ss:app-low-rank}. We identify $n_{i,j}=n'_{2i-1,2j-1}$ and define the function  $f(i)=\begin{cases}
			0 & i \text{ odd} \\
			1& i \text{ even} \\
		\end{cases}. $
		Using the ideintification above, we can rewrite the generators of the Pfaffian ideal for $\Nm'$ in function of $n_{i,j}$'s and $\lambda$. If $i,j,k,l$ are such that there are not two consecutive indexes with the smallest being odd, then
		\begin{align*}
			& n'_{i,j}n'_{k,l}+n'_{i,k}n'_{j,l}+n'_{i,l}n'_{j,k}\\
			=& \lambda^{f(i)+f(j)}n_{\ceil{\frac{i}{2}},\ceil{\frac{j}{2}}}\lambda^{f(k)+f(l)}n_{\ceil{\frac{k}{2}},\ceil{\frac{l}{2}}}+\lambda^{f(i)+f(k)}n_{\ceil{\frac{i}{2}},\ceil{\frac{k}{2}}}\lambda^{f(j)+f(l)}n_{\ceil{\frac{j}{2}},\ceil{\frac{l}{2}}}\\&+\lambda^{f(i)+f(l)}n_{\ceil{\frac{i}{2}},\ceil{\frac{l}{2}}}\lambda^{f(j)+f(k)}n_{\ceil{\frac{j}{2}},\ceil{\frac{k}{2}}}\\
			=&\lambda^{f(i)+f(j)+f(k)+f(l)}(n_{\ceil{\frac{i}{2}},\ceil{\frac{j}{2}}}n_{\ceil{\frac{k}{2}},\ceil{\frac{l}{2}}}+n_{\ceil{\frac{i}{2}},\ceil{\frac{k}{2}}}n_{\ceil{\frac{j}{2}},\ceil{\frac{l}{2}}}+n_{\ceil{\frac{i}{2}},\ceil{\frac{l}{2}}}n_{\ceil{\frac{j}{2}},\ceil{\frac{k}{2}}}).
		\end{align*}
		Otherwise, if for instance $j=i+1$, $i$ odd, then
		\[
		n'_{i,j}n'_{k,l}+n'_{i,k}n'_{j,l}+n'_{i,l}n'_{j,k}= 0\cdot n'_{k,l}+ n'_{i,k}(\lambda n'_{i,l}) +  n'_{i,l}(\lambda n'_{i,k})=0
		\]
		Therefore $\cP_2(\Nm')=\cP_2(\Nm)$ seen as ideals in $\Fqm[(n_{i,j})_{i,j}, \lambda]$.
		Hence
		\[
		\dim\Vm(\cP_2^+(\Mm))\ge \dim\Vm(\cP_2^+(\Nm)')=1+\dim \Vm(\cP_2(\Nm))=1+2\frac{r}{2}-3=r-2,
		\]
		where the summand 1 corresponds to the free parameter $\lambda$ used in the construction.
	\end{itemize}
	\qed
\end{proof}

\subsection{Experiments about the Hilbert function convergence} \label{sec:HF_exps}
In Conjecture~\ref{conj: HF_asymptotic}, we claimed that $d_{0} \sim c\frac{s^2}{k}$ for some constant $c$. We experimentally verified this in the following way. We define $k=\floor{\beta s^\alpha}$ for several positive values of $\beta$ and $\alpha\in (1,2)$. We start from a value $s=2^i$ such that the parameters are above Gilbert-Varshamov bound and not distinguishable and then we let $s$ double each time and update $k$ accordingly. The ratio $\frac{d_{0}k}{s^2}$ is eventually a decreasing function and seems to converge to $c=\frac{1}{4}$ (or something very close to it) from above, even though with a different speed depending on $\alpha$.
In particular, let us choose $\beta=1$ and let us test the convergence for different values of $\alpha$ in Table~\ref{table: exps_conjecture}.
\begin{table} 
	\begin{tabular}{|c||c|c|c|c|c|c|c|c|}
		\hline
		$\alpha$ & 1.2 & 1.3 & 1.4 & 1.5 & 1.6 & 1.7 & 1.8 & 1.9 \\ \hline \hline
		$\frac{d_{0}k}{s^2}<0.28$ starting from & $s=2^{18}$ & $s=2^{14}$ & $s=2^{14}$ & $s=2^{15}$ & $s=2^{18}$ & $s=2^{24}$ & $s=2^{36}$ & $s=2^{72}$ \\ \hline
		$\frac{d_{0}k}{s^2}<0.255$ starting from & & $s=2^{21}$ & $s=2^{20}$ & $s=2^{23}$ & $s=2^{29}$ & $s=2^{38}$ & $s=2^{57}$ & $s=2^{114}$ \\ \hline
		$\frac{d_{0}k}{s^2}<0.252$ starting from  & & & $s=2^{23}$ & $s=2^{28}$ & $s=2^{34}$ & $s=2^{45}$ & $s=2^{67}$ & $s=2^{133}$ \\ \hline
		$\frac{d_{0}k}{s^2}<0.251$ starting from  & & & & & $s=2^{37}$ & $s=2^{50}$ & &  \\ \hline
	\end{tabular}
	\caption{Experiments for the convergence of the Hilbert function} \label{table: exps_conjecture}
\end{table}

\subsection{Proof of Theorem \ref{thm : lowerbound_HF_binGoppa}}
Let us first recall the Theorem.
\lowerboundHFBG*
\begin{proof}
	For our convenience we denote $R\eqdef \F_{2^m}[\mv]$. Define $\mv^{(l)}\eqdef(m_{i,j})_{lr+1\le i<j\le (l+1)r}$ and $\mv^{(\setminus l)}$ the sequence of monomials that are in $\mv$ but not in $\mv^{(l)}$, for all $l\in \Iintv{0}{m-1}$. Moreover, we define the sequence of variables $\mv^{(out)}$ that are not in any of the $\mv^{(l)}$'s. We consider the corresponding polynomial rings $R_l\eqdef \F_{2^m}[\mv^{(l)}]$, $R_{\setminus l}\eqdef\F_{2^m}[\mv^{(\setminus l)}]$ and $R_{out}\eqdef\F_{2^m}[\mv^{(out)}]$ and, with some abuse of notation, the monomial ideals over $\F_{2^m}[\mv]$ generated by these sequences of variables: $\cI^{(l)}=\cI(\mv^{(l)}), \cI^{(\setminus l)}=\cI(\mv^{(\setminus l)}), \cI^{(out)}=\cI(\mv^{(out)})$. Finally, we define the monomial ideal $\cI^{(quad)}$ generated by all possible quadratic monomials with two unknowns belonging to two different diagonal blocks. We recall from Proposition~\ref{prop: rel_binGoppa} that each skew-symmetric block diagonal matrix belongs to $\Cmat(\cA)$. Therefore, the homogeneous linear relations $L_i$'s such that $\cP_2^+(\Mm)=\cP_2(\Mm)+\langle L_i\rangle_i$ can be chosen in such a way that only the variables in $\mv^{(out)}$ can appear in them, \ie $L_j \in \cI^{(out)}$. 
	
	Let us take an element in the basis of $\cP_2(\Mm)$ as in \eqref{eq: pfaffian_ideal}:
	\[ Q_{a,b,c,d}=m_{a,b}m_{c,d}+m_{a,c}m_{b,d}+m_{a,d}m_{b,c}.\]
	We analyze two cases:
	\begin{itemize}
		\item If there exists $l\in \Iintv{0}{m-1}$ such that $lr+1\le a<b<c<d\le (l+1)r$, then $Q_{a,b,c,d}\in \cP_2(\Mm_{l,l})$, \ie the Pfaffian ideal corresponding to the $(l+1)$-th diagonal block submatrix. 
		\item Otherwise, the monomials $m_{a,b}m_{c,d},m_{a,c}m_{b,d}$ and $m_{a,d}m_{b,c}$ belong to either $\cI^{(out)}$ or $\cI^{(quad)}$.
	\end{itemize}
	In both cases we obtain
	\[
	Q_{a,b,c,d}\in \left(\sum_{i=0}^{m-1} \cP_2(\Mm_{l,l})\right) + \cI^{(out)}+\cI^{(quad)}.
	\]
	Hence
	\[
	\cP_2^+(\Mm)=\cP_2(\Mm)+\langle L_1,\dots , L_k\rangle \subseteq  \left(\sum_{l=0}^{m-1} \cP_2(\Mm_{l,l})\right) + \cI^{(out)}+\cI^{(quad)}.
	\]
	One can readily verify that, for any $l \in \Iintv{0}{m-1}$, the monomial ideal $\cI^{(\setminus l)}$ contains:
	\begin{itemize}
		\item  $\cI^{(out)}$;
		\item $\cI^{(quad)}$;
		\item $\cP_2(\Mm_{l',l'})$, for all $l' \in \Iintv{0}{m-1} \setminus \{l\}$.
	\end{itemize}
	This results in
	\[
	\left(\sum_{l=0}^{m-1} \cP_2(\Mm_{l,l})\right) + \cI^{(out)}+\cI^{(quad)}\subseteq \bigcap_{l \in \Iintv{0}{m-1}} \left( \cP_2(\Mm_{l,l})+\cI^{(\setminus l)}\right).
	\]
	Note now that, for any $\bar{l}\in \Iintv{1}{m-1}$,
	\begin{equation} \label{eq: rec_HSblocks}
		\bigcap_{l \in \Iintv{0}{\bar{l}-1}} \left( \cP_2(\Mm_{l,l})+\cI^{(\setminus l)}\right)+\cP_2(\Mm_{\bar{l},\bar{l}})+\cI^{(\setminus \bar{l})}=\langle \mv \rangle
	\end{equation}
	and 
	\[
	\HS_{R/\langle \mv \rangle}(z)=\HS_{\F_{2^m}}(z)=1.
	\]
	By applying recursively relation \eqref{eq: rec_HSblocks} on the quotient rings, we obtain
	\begin{align*}
		\HF_{R/\cP_2^+(\Mm)}(d)\ge & \HF_{R/\bigcap_{l \in \Iintv{0}{m-1}} \left( \cP_2(\Mm_{l,l})+\cI^{(\setminus l)}\right)}(d)\\
		=& \HF_{R/\bigcap_{l \in \Iintv{0}{m-2}} \left( \cP_2(\Mm_{l,l})+\cI^{(\setminus l)}\right)}(d)+\HF_{R/\left( \cP_2(\Mm_{m-1,m-1})+\cI^{(\setminus m-1)}\right)}(d)-\HF_{\F_{2^m}}(d)\\
		=& \HF_{R/\bigcap_{l \in \Iintv{0}{m-3}} \left( \cP_2(\Mm_{l,l})+\cI^{(\setminus l)}\right)}(d)+\HF_{R/\left( \cP_2(\Mm_{m-2,m-2})+\cI^{(\setminus m-2)}\right)}(d)\\ &+\HF_{R/\left( \cP_2(\Mm_{m-1,m-1})+\cI^{(\setminus m-1)}\right)}(d)-2\HF_{\F_{2^m}}(d)\\
		=&\dots\\
		=&\sum_{l=0}^{m-1}\HF_{R/\left( \cP_2(\Mm_{l,l})+\cI^{(\setminus l)}\right)}(d)-(m-1)\HF_{\F_{2^m}}(d)\\
		=&\sum_{l=0}^{m-1}\HF_{R_l/ \cP_2(\Mm_{l,l})}(d)-(m-1)\HF_{\F_{2^m}}(d)\\
		=&\begin{cases}
			m-(m-1)=1 &\quad \text{if } d=0\\
			m\left( \binom{r+d-2}{d}^2-\binom{r+d-2}{d+1}\binom{r+d-2}{d-1}\right) &\quad \text{if } d>0
		\end{cases}.
	\end{align*}
	\qed
\end{proof}

 \section{Proofs for some of the Results of Section \ref{sec:attack}}
\label {sec:app-attack}

\subsection{Proof of Proposition \ref{prop: stable_frobenius}}
Let us recall first the proposition

\propstablefrobenius*

\begin{proof}
  Let $\Bm = (b_{i,j})_{i,j} \in \Cmat(\mathcal B) \subseteq \Fqm^{rm \times rm}$.
  Then, by definition,
  \[
    \sum_{i< j} 2b_{i,j} \starp{\bv_i}{\bv_j}+\sum_i b_{i,i} \starp{\bv_i}{\bv_i}= 0.
  \]
    Then, by applying the Frobenius map $z\mapsto z^q$ component-wise,
  \[
   0=\sum_{i< j} 2^q b_{i,j}^q \starp{\bv_i^q}{\bv_j^q}+\sum_i b_{i,i}^q \starp{\bv_i^q}{\bv_i^q} =\sum_{i< j} 2 b_{i,j}^q \starp{\bv_i^q}{\bv_j^q}+\sum_i b_{i,i}^q \starp{\bv_i^q}{\bv_i^q}.
  \]
  From now on, the indexes are considered modulo $rm$. The structure
  of the basis $\mathcal B$ yields
  \[
   \sum_{i< j} 2 b_{i,j}^q \starp{\bv_{i+r}^q}{\bv_{j+r}^q}+\sum_i b_{i,i}^q \starp{\bv_{i+r}^q}{\bv_{i+r}^q}= 0  
  \]
  and hence
  \[
   \sum_{i< j} 2 b_{i-r,j-r}^q \starp{\bv_i^q}{\bv_j^q}+\sum_i b_{i-r,i-r}^q \starp{\bv_i^q}{\bv_i^q}= 0  
  \]
  The matrix $(b_{i-r,j-r}^q)_{i,j}$ is nothing but $\trsp{\Sm} \Bm^{(q)} \Sm$
  and hence $\trsp{\Sm} \Bm^{(q)} \Sm \in \Cmat (\mathcal B)$. \qed
\end{proof}

\subsection{Proof of Proposition \ref{prop:at_once}}
Let us recall this proposition
\propatonce*
\begin{proof}
Given $\Bm \in \Cmat(\cB)$, it follows from Proposition~\ref{prop: stable_frobenius} that
\[
  (\trsp{\Sm})^i \Bm^{(q^i)} (\Sm)^i \in \Cmat(\cB),\quad \text{for any } i \in \Iintv{0}{m-1}
\]
and all these matrices have the same rank, namely $rm-1$. Moreover, if $\vv$ generates the nullspace of $\Bm$, then 
$\vv^{q^{i}}\Sm^{i}$ is in the kernel of $(\trsp{\Sm})^i \Bm^{(q^i)} \Sm^i$ since
\begin{align*}
	& (\vv^{q^i}  \Sm^i ) \cdot (\trsp{\Sm})^i \Bm^{(q^i)} \Sm^i \\
	=& \vv^{q^i}\Bm^{(q^i)} \Sm^i \\
	=& (\vv\Bm)^{(q^i)} \Sm^i\\
	=&0.
\end{align*}
\qed
\end{proof}

\subsection{Proof of Proposition \ref{prop:corrGRS}}
The proposition states that
\propcorrGRS*
\begin{proof}
Let $\Bm\in \Cmat(\cB)$ such that $\uv_2$ generates $\ker(\Bm)$. By Proposition~\ref{prop:at_once}, we know that $\uv_2^{q^l}\Sm^{l}$ generates the kernel of $(\trsp{\Sm})^{l} \Bm^{(q^{l})} \Sm^l$.
We get
\begin{align*}
0=&\uv_2^{q^l} \Bm^{(q^{l})}&\\
=&\uv_2^{q^l} (\trsp{\Pm^{(q^l)}})^{-1} \Am^{(q^{l})} (\Pm^{(q^l)})^{-1}&\\
=&\uv_2^{q^l} (\trsp{\Pm^{(q^l)}})^{-1} (\Sm^l (\trsp{\Sm})^l) \Am^{(q^{l})} (\Sm^l (\trsp{\Sm})^l)(\Pm^{(q^l)})^{-1}&\\
=&(\uv_2^{q^l} (\trsp{\Pm^{(q^l)}})^{-1} \Sm^l) ((\trsp{\Sm})^l \Am^{(q^{l})} \Sm^l) (\trsp{\Sm})^l(\Pm^{(q^l)})^{-1}\\
=&(\uv_2^{q^l} ((\trsp{\Sm})^l\trsp{\Pm^{(q^l)}})^{-1}) ((\trsp{\Sm})^l \Am^{(q^{l})} \Sm^l) (\trsp{\Sm})^l(\Pm^{(q^l)})^{-1} \quad \text{by Proposition~\ref{prop: P^q}}\\
=&(\uv_2^{q^l} \Sm^l\trsp{\Pm}^{-1}) ((\trsp{\Sm})^l \Am^{(q^{l})} \Sm^l) (\trsp{\Sm})^l(\Pm^{(q^l)})^{-1},
\end{align*}
which implies
\begin{align*}
(\uv_2^{q^l} \Sm)\trsp{\Pm}^{-1}\Pm^{-1} \Hm_{\cB}=(\uv_2^{q^l} \Sm\trsp{\Pm}^{-1}) \Hm_{\cA}\in\GRS{r}{\xv}{\yv}^{(q^{j_2+l})},
\end{align*}
since the diagonal block of rank $r-1$ in $(\trsp{\Sm})^l \Am^{(q^{l})} \Sm^l$ is the one indexed by $j_2+l \mod m$. Therefore, $\uv_1$ and $\uv_2^{q^l}\Sm^{l}$ correspond to the same GRS code with respect to $\cB$ for the unique value $l\in \Iintv{0}{m-1}$ such that $j_1=j_2+l \mod m$. \qed
\end{proof}

\subsection{Proof of Proposition \ref{prop: SC_aux}}
This proposition says that
\propSCaux*
\begin{proof}
Let $j \in \Iintv{0}{r-1}$. For each $i\in \Iintv{0}{m-1}$, there exists a unique $\vv_j^{(q^l)} \Sm^l$, $l \in \Iintv{0}{m-1}$, such that
\[
\vv_j^{q^l} \Sm^l\trsp{(\Pm^{-1})} \Pm^{-1} \Hm_{\cB}\subseteq \GRS{r}{\xv}{\yv}^{(q^i)}. 
\]
As this holds for all $j \in \Iintv{0}{r-1}$, we obtain that
\[
\SC_{aux}  \trsp{(\Pm^{-1})} \Pm^{-1} \Hm_{\cB} = \sum_{i=0}^{m-1} \GC_i,
\]
where $\GC_i$ is an $[n,r-1]$ linear code contained into $\GRS{r}{\xv}{\yv}^{(q^i)}.$ From the standard assumption that all the codes $\GRS{r}{\xv}{\yv}^{(q^i)}$'s are in direct sum, we get $\SC_{aux}  \trsp{(\Pm^{-1})} \Pm^{-1} \Hm_{\cB} = \bigoplus_{i=0}^{m-1} \GC_i$. Analogously for $\uv_t$, $t=1,2$, we have
\[
\uv_t \trsp{(\Pm^{-1})}\Pm^{-1} \Hm_{\cB}\in \GRS{r}{\xv}{\yv}^{(q^{i_t})}. 
\]
The condition $\dim_{\Fqm} \SC_{aux}+\Fqmspan{\uv_t} =(r-1)m+1=\dim_{\Fqm} \SC_{aux}+1$ implies that
\[
\left(\SC_{aux}+\Fqmspan{\uv_t}\right)\trsp{(\Pm^{-1})} \Pm^{-1} \Hm_{\cB} = \left(\bigoplus_{i\in \Iintv{0}{m-1}\setminus \{i_t\}} \GC_i\right) \oplus \GC_{i_t}',
\]
with $\GC_{i_t}' \subseteq \GRS{r}{\xv}{\yv}^{(q^{i_t})}$. But
\[
\dim_{\Fqm} \GC_{i_t}'=\dim_{\Fqm} \SC_{aux}+\Fqmspan{\uv_t}-\dim_{\Fqm}\bigoplus_{i\in \Iintv{0}{m-1}\setminus \{i_t\}} \GC_i= (r-1)m+1-(r-1)(m-1)=r,
\]
hence $\GC_{i_t}'=\GRS{r}{\xv}{\yv}^{(q^{i_t})}$. Note that, with the same argument,
\begin{align*}
&\left(\SC_{aux}+\Fqmspan{\uv_2^{q^l}\Sm^l}\right)\trsp{(\Pm^{-1})} \Pm^{-1} \Hm_{\cB} \\=& \left(\bigoplus_{i\in \Iintv{0}{m-1}\setminus \{(i_2+l) \mod m\}} \GC_i\right) \oplus \GRS{r}{\xv}{\yv}^{(q^{i_2+l})}.
\end{align*}
We can conclude that 
\begin{align*}
&\left(\SC_{aux}+\Fqmspan{\uv_1,\uv_2^{q^l}\Sm^l}\right)\trsp{(\Pm^{-1})} \Pm^{-1} \Hm_{\cB} \\=& \left(\bigoplus_{i\in \Iintv{0}{m-1}\setminus \{i_1,(i_2+l) \mod m\}} \GC_i\right) \oplus \GRS{r}{\xv}{\yv}^{(q^{i_1})} \oplus \GRS{r}{\xv}{\yv}^{(q^{i_2+l})}.
\end{align*}
Hence
\begin{align*}
\dim_{\Fqm} \SC_{aux}+\Fqmspan{\uv_1, \uv_2^{q^l}\Sm^l} = \begin{cases} (r-1)m+1 &\quad \text{if } i_1=i_2+l \mod m\\ (r-1)m+2 &\quad \text{otherwise}\end{cases}.
\end{align*}
and the first case is equivalent to say that
\[
\Fqmspan{\uv_1, \uv_2^{q^l}\Sm^l} \trsp{(\Pm^{-1})} \Pm^{-1} \Hm_{\cB}\subseteq \GRS{r}{\xv}{\yv}^{(q^{i_1})},
\]
\ie $\uv_1$ and $ \uv_2^{q^l}\Sm^l$ correspond to the same GRS code with respect to $\cB$. \qed
\end{proof}

\subsection{Proof of Proposition \ref{prop:Vperp}}
Let us recall this proposition
\propVperp*
\begin{proof}
 Since $\dim_{\Fqm}(\VC_j)=r$ and each of its elements correspond to the $j$-th GRS code, a generator matrix of $\VC_j \trsp{(\Pm^{-1})}$ is
\[
[\zerom_{r\times r} \mid \dots \mid \zerom_{r\times r}  \mid \underbrace{\Im_r }_{j\text{-th block}} \mid \zerom_{r\times r}  \mid \dots \mid \zerom_{r\times r} ].
\]
Let us pick $\vv^\perp \in \VC_j^\perp$. For any $\vv \in \VC_j$, we can write
\begin{align*}
	0=& \langle \vv, \vv^\perp \rangle \\
	=&  \langle \vv \Im_{rm}, \vv^\perp \rangle \\
	=&  \langle \vv (\trsp{\Pm})^{-1} \trsp{\Pm}, \vv^\perp \rangle \\
	=&  \langle \vv\trsp{(\Pm^{-1})} , \vv^\perp \Pm \rangle .
\end{align*}
Therefore $\vv^\perp \Pm$ is zero on the $j$-th block. Hence
\[
 \VC_j^\perp \Hm_{\cB}= (\VC_j^\perp \Pm) \Hm_{\cA} \subseteq \sum_{i \in \Iintv{0}{m-1}\setminus \{j\}} \GRS{r}{\xv}{\yv}^{(q^i)},
\]
and since $\dim_{\Fqm}(\VC_j^\perp)=rm-\dim_{\Fqm}\VC_j=(r-1)m$,
\[
\VC_j^\perp \Hm_{\cB}=  \sum_{i \in \Iintv{0}{m-1}\setminus \{j\}} \GRS{r}{\xv}{\yv}^{(q^i)}.
\]
\qed
\end{proof}

\subsection{Proof of Proposition \ref{prop:capVperp}}
Let us recall this proposition
\propcapVperp*
\begin{proof}
Since $\GRS{r}{\xv}{\yv}^{(q^j)}\subset  \VC_i^\perp$ for all $i \ne j$, it follows that
\[
\GRS{r}{\xv}{\yv}^{(q^j)}\subseteq \bigcap_{i \in \Iintv{0}{m-1}\setminus \{j\}} \VC_i^\perp \Hm_{\cB}.
\]
On the other hand, since the GRS codes are in direct sum, we get
\[
\dim_{\Fqm}\bigcap_{i \in \Iintv{0}{m-1}\setminus \{j\}} \VC_i^\perp=r(m-1)-r(m-2)=r,
\]
which leads to the equality.
\qed
\end{proof}

\subsection{Estimate of matrices of rank $rm-1$ in $\Cmat(\cB)$}\label{ss:app-estimate_rank}
 We start by recalling that $\Cmat(\cA)$ and $\Cmat(\cB)$ have the same weight distribution, thus it is convenient to consider the block diagonal structure of $\Cmat(\cA)$. 
Let us also define the matrix space containing all possible block diagonal (with $m$ blocks of size $r\times r$) symmetric matrices $\DC\subset \Sym(rm, \Fqm)$. The ratio of rank $rm-1$ matrices in $\Cmat(\cD)$ is given by 
\[
\frac{N(r,r-1)\cdot N(r,r)^{m-1}}{(q^m)^{m\binom{r+1}{2}}},
\]
where $N$ is defined as in Proposition~\ref{prop: N}.
Note that, for $q^m\to \infty$, 
\[
N(t,s)\to\prod_{i=0}^{s-1} (q^m)^{t-i}=(q^m)^{\sum_{i=0}^{s-1} t-i}=(q^m)^{\binom{t+1}{2}-\binom{t-s+1}{2}}=(q^m)^{ts-s^2/2+s/2}.
\]
Hence the ratio above tends to
\[
\frac{N(r,r-1) N(r,r)^{m-1}}{(q^m)^{m\binom{r+1}{2}}}\to \frac{(q^m)^{r(r-1)-(r-1)^2/2+(r-1)/2}(q^m)^{(m-1)(r^2-r^2/2+r/2)}}{(q^m)^{m\binom{r+1}{2}}}=\frac{1}{q^m}.
\] The ratios of matrices of a given rank in $\Cmat(\cB)$ is not the same as for $\Cmat(\cD)$, and a more detailed analysis would be useful to derive the exact probability of sampling matrices of rank $rm-1$. However, we expect the distribution not to deviate too much from this behavior.
 We provide in Table~\ref{table: givenrank_block} the number of different diagonal blocks in $\Cmat(\cB)$ of a given rank in for small values of $q^m$ (odd case) and $r$. Note that the total number is given by $(q^m)^{\binom{r-1}{2}}$. 
 \begin{table}
   \center
 \begin{tabular}{|c|c|c|}
 \hline
 $r$ & $q^m$ & [rank 0, rank 1, ..., rank $r$] \\
 \hline \hline
 3 & 3 & [1, 0, 0, 2] \\
 3 & 5 & [1, 0, 0, 4] \\
 3 & 7 & [1, 0, 0, 6] \\
 3 & 9 & [1, 0, 0, 8] \\
 3 & 11 & [1, 0, 0, 10] \\
 \hline
 4 & 3 & [1, 0, 0, 8, 18] \\
 4 & 5 & [1, 0, 0, 24, 100] \\
 4 & 7 & [1, 0, 0, 48, 294] \\
 4 & 9 & [1, 0, 0, 80, 648] \\
 4 & 11 & [1, 0, 0, 120, 1210] \\
 \hline
 5 & 3 & [1, 0, 0, 44, 378, 306] \\
 5 & 5 & [1, 0, 0, 224, 5500, 9900] \\
 5 & 7 & [1, 0, 0, 636, 30870, 86142] \\
 5 & 9 & [1, 0, 0, 1376, 110808, 419256] \\
 5 & 11 & [1, 0, 0, 2540, 306130, 1462890] \\
 \hline
 6 & 3 & [1, 0, 0, 152, 4374, 18072, 36450] \\
 6 & 5 & [1, 0, 0, 1224, 157500, 1919400, 7687500] \\
 \hline
 7 & 3 & [1, 0, 0, 638, 55566, 587502, 4754538, 8950662] \\
 \hline
 \end{tabular}
 \caption{Experimental number of different diagonal blocks in $\Cmat(\cB)$ of a given rank ($q^m$ odd case).} \label{table: givenrank_block}
 \end{table}

\subsection{Characteristic 2}\label{ss:app_attack_char_2}
Recall that skew-symmetric matrices in characteristic 2 can only have even rank. This immediately invalidates the search arguments explained before: either rank $rm$ or $rm-1$ do not exist in $\Cmat(\cA)$ and $\Cmat(\cB)$. The same constraint occurs for the $r\times r$ diagonal blocks with respect to the canonical basis $\cA$, on which we focus now. However, the previous strategy can be adapted to even characteristic by limiting the search to even-rank matrices. Indeed, in our setting, the maximum rank achievable in $\Cmat(\cA)$ is $2\floor{\frac{r}{2}}m$, because for each $r\times r$ diagonal block the rank is at most the largest even integer bounded by $r$, \ie $2\floor{\frac{r}{2}}$. Consequently, the second largest rank achievable by a matrix in $\Am\in\Cmat(\cA)$ with the block diagonal structure as in \eqref{eq: A_diagblock} is 
\[
2\floor{\frac{r}{2}}m-2.
\]

\begin{itemize}
\item In the case where $r$ is even, $2\floor{\frac{r}{2}}m-2=rm-2$, and $\Am$ as in \eqref{eq: A_diagblock} is such that there exists a unique $j \in \Iintv{1}{m}$ for which $\rank(\Am_{j,j})=r-2$ and $\rank(\Am_{j,j})=r$ otherwise. Indeed, the parity constraint on the skew-symmetric matrices prohibits having two diagonal blocks of rank $r-1$. This time, the nullspace of $\Am$ is generated by two linearly independent vectors $\uv$ and $\vv$, and these vectors are zero outside the same $j$-th length-$r$ blocks:
\[
\vv=(\zerov,\dots,\zerov,\vv_i,\zerov,\dots,\zerov)
\] 
and
\[
\uv=(\zerov,\dots,\zerov,\uv_i,\zerov,\dots,\zerov).
\] 
With similar arguments to the $q$ odd case, it is therefore possible to retrieve a basis of a GRS code $\GRS{r}{\xv}{\yv}^{q^j}$. We only need to give an estimate of the ratio of rank $rm-2$ matrices in $\Cmat(\cA)$, to ensure that we can find them with non-negligible probability.

We consider the $rm\times rm$ matrix space $\DC\subset \Skew(rm,\Fqm)$ containing all possible block diagonal (with $m$ blocks of size $r\times r$) skew-symmetric matrices. The ratio of rank $rm-2$ matrices in $\DC$ is given by 
\begin{equation} \label{eq: ratio_reven}
\frac{N_0(r,r-2)\cdot N_0(r,r)^{m-1}}{(q^m)^{m\binom{r}{2}}}.
\end{equation}
Note that for $q\to \infty$, 
\[
N_0(t,2s)\to\prod_{i=0}^{s} \frac{1}{q^2} \prod_{i=0}^{2s-1} (q^m)^{t-i}=(q^m)^{-2s+\sum_{i=0}^{2s-1} t-i}=(q^m)^{\binom{t+1}{2}-\binom{t-2s+1}{2}-2s}=(q^m)^{s(2t-2s-1)}.
\]
Hence the ratio above tends to
\[
\frac{N_0(r,r-2)\cdot N_0(r,r)^{m-1}}{(q^m)^{m\binom{r}{2}}}\to \frac{(q^m)^{\frac{(r+1)(r-2)}{2}}(q^m)^{(m-1)\binom{r}{2}}}{(q^m)^{m\binom{r}{2}}}=\frac{1}{q^m},
\]
\ie the same as for rank $rm-1$ matrices in the $q$ odd case. The approach for finding matrices of rank $rm-1$ described above is therefore expected to work with high probability in this case as well. 
\begin{remark}
In the case of a binary Goppa code with a square-free Goppa polynomial, we have shown in Proposition~\ref{prop: rel_binGoppa} that $\Cmat(\cA)$ contains the space of block-diagonal skew-symmetric matrices with $r\times r$ blocks. Under the condition that $r<q-1$, these matrices generates $\Cmat(\cA)$, \ie $\Cmat(\cA)=\DC$. Therefore, in this special case, \eqref{eq: ratio_reven} provides the exact ratio of rank $rm-2$ matrices in $\Cmat(\cA)$ (or $\Cmat(\cB)$). 
\end{remark}

Similarly to what done for the $q$ odd case, we provide in Table~\ref{table: givenrank_block_qeven_reven} the number of different diagonal blocks in $\Cmat(\cB)$ (when the latter is not originated by a binary Goppa code with square-free polynomial)  of a given rank in for small values of $q^m$ (even case) and $r$ (even case). The total number is given by $(q^m)^{\binom{r-1}{2}}$, as before. 
\begin{table}
  \center
\begin{tabular}{|c|c|c|}
\hline
$r$ & $q^m$ & [rank 0, rank 1, ..., rank $r$] \\
\hline \hline
4 & 2 &  [ 1, 0, 3, 0, 4 ] \\
4 & 4 & [ 1, 0, 15, 0, 48 ] \\
4 & 8 & [ 1, 0, 63, 0, 448] \\
\hline
6 & 2 & [ 1, 0, 27, 0, 612, 0, 384 ]\\
6 & 4 & [ 1, 0, 495, 0, 286224, 0, 761856 ] \\
\hline
8 & 2 &  [ 1, 0, 171, 0, 51348, 0, 1181376, 0, 864256 ]\\
\hline
\end{tabular}
\caption{Experimental number of different diagonal blocks in $\Cmat(\cB)$ of a given rank ($q^m$ even, $r$ even case).} \label{table: givenrank_block_qeven_reven}
\end{table}
\begin{remark}
In all instances where a filtration has been initially applied, $r=q$ is even, therefore they fall in this case.
\end{remark}
The case $q$ even and $r$ even requires only small changes in Algorithm~\ref{alg: attack}. At lines~\ref{row:v},\ref{row:u1} and \ref{row:uj}, the vectors $\vv$, $\uv_1$ and $\uv_j$ respectively are defined as generators of kernels of rank $rm-1$ matrices. However, the nullspace of a square matrix of rank $rm-2$ and size $rm$ is generated by two linearly independent elements. In this case, it suffices to define such vectors as any non-zero element in the kernel and the algorithm still works correctly. It is even possible to exploit the knowledge that two linearly independent generators of a kernel correspond to the same GRS code and roughly halve the number of matrices of rank $rm-2$ that need to be sampled. 
\item In the case where $r$ is even, $2\floor{\frac{r}{2}}m-2=(r-1)m-2$ and a similar computation shows that the ratio of rank $r(m-1)-2$ matrices in $\DC$ is given by
\begin{equation} \label{eq: ratio_rodd}
\frac{N_0(r,r-3)\cdot N_0(r,r-1)^{m-1}}{(q^m)^{m\binom{r}{2}}}
\end{equation}
and, for $q\to \infty$,
\[
\frac{N_0(r,r-3)\cdot N_0(r,r-1)^{m-1}}{(q^m)^{m\binom{r}{2}}}\to \frac{(q^m)^{\frac{(r+2)(r-3)}{2}}(q^m)^{(m-1)\binom{r}{2}}}{(q^m)^{m\binom{r}{2}}}= \frac{1}{(q^m)^3}.
\]
\begin{remark}
Similarly to the $r$ even case, for binary Goppa codes with square-free Goppa polynomials, \eqref{eq: ratio_rodd} provides the exact ratio of rank $(r-1)m-2$ matrices in $\Cmat(\cA)$ (or $\Cmat(\cB)$). 
\end{remark}
We provide in Table~\ref{table: givenrank_block_qeven_rodd} the number of different diagonal blocks in $\Cmat(\cB)$ (when the latter is not originated by a binary Goppa code with square-free polynomial) of a given rank in for small values of $q^m$ (even case) and $r$ (odd case). The total number is given by $(q^m)^{\binom{r-1}{2}}$, as before. 

\begin{table}
  \center
\begin{tabular}{|c|c|c|}
\hline
$r$ & $q^m$ & [rank 0, rank 1, ..., rank $r$] \\
\hline \hline
3 & 2 & [ 1, 0, 1, 0 ] \\
3 & 4 & [ 1, 0, 3, 0 ] \\
3 & 8 & [ 1, 0, 7, 0 ] \\
\hline
5 & 2 & [ 1, 0, 11, 0, 52, 0 ] \\
5 & 4 & [ 1, 0, 111, 0, 3984, 0 ] \\
5 & 8 & [ 1, 0, 959, 0, 261184, 0 ] \\
\hline
7 & 2 &  [ 1, 0, 75, 0, 5748, 0, 26944, 0 ] \\
\hline
\end{tabular}
\caption{Experimental number of different diagonal blocks in $\Cmat(\cB)$ of a given rank ($q^m$ even, $r$ odd case).} \label{table: givenrank_block_qeven_rodd}
\end{table}
Rank $rm-3$ matrices are therefore less probable to be sampled. This issue can be overcome at an asymptotic cost of a factor $q^{3m}$. Furthermore, a Gr\"obner basis approach leads in practice to an even better complexity. More specifically, we can generalize the argument for the previous cases by sampling at random $\Bm_1,\Bm_2,\Bm_3,\Bm_4 \in \Cmat(\cB)$ and solving the trivariate affine polynomial $\det(w_1\Bm_1+w_2 \Bm_2+w_3 \Bm_3+\Bm_4)$ with Gr\"obner basis techniques. As the number of variables is small and constant, this approach seems to be much more efficient than brute force. However, there is another more problematic issue. The nullspace of a matrix $\Am \in\Cmat(\cA)$ has in this case dimension $rm-((r-1)m-2)=m+2$ and its generators are not all zero outside a length-$r$ block. Therefore the strategy explained before does not apply directly here. We treat this case in \ref{ss:app-qeven-rodd}.
\end{itemize}

\subsection{The attack for $q$ even and $r$ odd}\label{ss:app-qeven-rodd}
As already mentioned, the case where $q$ is even and $r$ is odd raises the additional problem that the nullspace of a matrix $\Am \in\Cmat(\cA)$ of rank $(r-1)m-2$ is not zero outside a length-$r$ block. The key idea to adapt the attack is that such nullspace is still ``unbalanced'' with respect to the $m$ blocks. Indeed, let us consider $\Bm=\trsp{(\Pm^{-1})} \Am (\Pm^{-1}) \in\Cmat(\cB)$ of rank $(r-1)m-2$. Since $\rank(\Bm)=\rank(\Am)=\sum_{l=0}^{m-1}\rank(\Am_{l,l})$ and for any $l$, $\rank(\Am_{l,l})\le r-1$ and it is even, we have that
\[
\exists! l \in \Iintv{0}{m-1} \text{ s.t. } \rank(\Am_{l,l})= r-3 \land \forall i \in \Iintv{0}{m-1}\setminus \{l\},\; \rank(\Am_{i,i})= r-1.
\]
Therefore, the kernel can be written as
\[
\ker \Bm= \langle \vv_0,\dots,\vv_{l-1},\vv_{l,1},\vv_{l,2},\vv_{l,3},\vv_{l+1},\dots,\vv_{m-1}\rangle
\]
so that for all $i \in \{1,2,3\}$
\[
\vv_{l,i} \trsp{(\Pm^{-1})} \Pm^{-1} \Hm_{\cB}=\vv_{l,i} \trsp{(\Pm^{-1})} \Hm_{\cA}\in \GRS{r}{\xv}{\yv}^{(q^l)}, 
\]
	and for all $j \in \Iintv{0}{m-1}\setminus \{l\}$
\[
	\vv_j \trsp{(\Pm^{-1})} \Pm^{-1} \Hm_{\cB}= \vv_j \trsp{(\Pm^{-1})} \Hm_{\cA} \in \GRS{r}{\xv}{\yv}^{(q^j)}.	 
\]
We do not know how to identify such vectors, though. Assume, however, that we are able to determine different matrices $\Bm_1,\dots,\Bm_{s}$ of rank $(r-1)m-2$ in $\Cmat(\cB)$ such that their counterparts in $\Cmat(\cA)$ have the rank-$(r-3)$ block indexed by the same $j\in \Iintv{0}{m-1}$, for some value $s<r$ that we are going to determine later. We will see how to achieve this in \ref{ss:app-Vj}. We define the $[rm, \le s(m-1)+\min(3s,r)]$ linear code $\VC_j\eqdef\sum_{i=1}^s \ker \Bm_i$. This construction can be seen as an adaptation of the definition given in Proposition~\ref{prop:Vperp}, where $\VC_j$ is spanned by $r$ vectors, each generating the nullspace of a matrix of rank $rm-1$. If the matrices $\Bm_i$'s have been sampled independently, as is the case, we expect, with a non-negligible probability, that a generator matrix of the code $\VC_j \trsp{(\Pm^{-1})}$ is the block diagonal matrix
\[
\begin{bmatrix}
\Gm_{0,0} &  &  & \\
	 & \Gm_{1,1} & & \zerov\\
	 \zerov & & \ddots & \\
	  &  &  & \Gm_{m-1,m-1} \\
\end{bmatrix}
\] 
where $\Gm_{j,j}$ has $\min(3s,r)$ rows while $\Gm_{i,i}$ has $s$ rows (and they all have $r$ columns). This is equivalent to say that $\dim_{\Fqm}\VC_j=s(m+2)$. Hence, by sampling $s\ge \ceil{r/3}$, we ensure $\Gm_{j,j}=\Im_r$ with non-negligible probability. From now on, we then assume that $\dim_{\Fqm} \VC_j=s(m-1)+r$. We define the $[rm, (r-s)(m-1)]$ dual code $\VC_j^\perp$ and, by repeating the computation made in the proof of Proposition~\ref{prop:Vperp}, we get that, for any $\vv^\perp \in \VC_j^\perp$, $\vv^\perp \Pm$ is zero on the $j$-th block. However, this time we can only assert that
\[
 \VC_j^\perp \Hm_{\cB}= (\VC_j^\perp \Pm) \Hm_{\cA} \subseteq \sum_{i \in \Iintv{1}{m}\setminus \{j\}} \GRS{r}{\xv}{\yv}^{(q^i)}
\]
with $\dim_{\Fqm}(\VC_j^\perp\Hm_{\cB})=\dim_{\Fqm}(\VC_j^\perp)=(r-s)(m-1)$. In order to obtain the code $\sum_{i \in \Iintv{1}{m}\setminus \{j\}} \GRS{r}{\xv}{\yv}^{(q^i)}$ it is then enough to repeat the process and analogously compute other linear codes $\VC_j',\VC_j'',\dots$ such that $(\VC_j')^\perp \Hm_{\cB}\subseteq \sum_{i \in \Iintv{1}{m}\setminus \{j\}} \GRS{r}{\xv}{\yv}^{(q^i)}$ as well (by sampling each time different matrices of rank $(r-1)m-2$). Since all these codes are constructed independently, we expect at some point
\[
 (\VC_j+\VC_j'+\VC_j''+\dots)^\perp \Hm_{\cB}=  \sum_{i \in \Iintv{1}{m}\setminus \{j\}} \GRS{r}{\xv}{\yv}^{(q^i)}.
\]
Since $(\VC_j+\VC_j'+\VC_j''+\dots)^\perp \Hm_{\cB}\subseteq  \sum_{i \in \Iintv{1}{m}\setminus \{j\}} \GRS{r}{\xv}{\yv}^{(q^i)}$, one can put $\dim_{\Fqm}(\VC_j+\VC_j'+\dots)^\perp=(r-1)m$ as an exit condition for the construction of such codes.
\begin{remark}
A good choice for $s$ is $\frac{r-1}{2}$. In this way, $\Gm_{j,j}=\Im_r$ with very high probability and at the same time, since $2(rm-s(m-1)-\min(3s,r))\ge 2(r-s)(m-1)\ge r(m-1)$, computing just two codes $\VC_j$ and $\VC_j'$ is typically enough to recover $\sum_{i \in \Iintv{1}{m}\setminus \{j\}} \GRS{r}{\xv}{\yv}^{(q^i)}$.
\end{remark}
Once the codes $\sum_{i \in \Iintv{1}{m}\setminus \{j\}} \GRS{r}{\xv}{\yv}^{(q^i)}$ have been retrieved for any $j\in \Iintv{0}{m-1}$, a  GRS block code can be obtained from intersections as done previously according to Proposition~\ref{prop:capVperp}.
\subsection{Computing $\VC_j$}\label{ss:app-Vj}
In this technical subsection, we tackle the problem of determining, given two matrices $\Bm_1,\Bm_2\in\Cmat(\cB)$ of rank $(r-1)m-2$, which blockwise Dickson shift of $\trsp{\Pm} \Bm_2 \Pm$ has the diagonal block of rank $r-3$ for the same index $l$ as $\Bm_1$. This represents the basic step to produce elements in $\VC_j$ in this case.
\begin{remark} Note that, in the $q$ odd case, this would be equivalent to determining which shift of $\vv_2$ corresponds to the same GRS code of $\vv_1$ where $\vv_1$ and $\vv_2$ are the generators of the kernels of two matrices $\Bm_1$ and $\Bm_2$ respectively of rank $rm-1$. In the case we examine now, however, the dimension of the nullspace is larger than 1 and not all its elements belong to the same GRS code. This explains why we need to move to the matrix formalism. Analogously, in the $q$ odd case, vectors corresponding to the same GRS code were identified by making use of an auxiliary linear code $\SC_{aux}$ spanned by kernel generators of a set of matrices. We will see that here we directly employ a set of auxiliary matrices $\Bm_{aux,i}$'s instead.
\end{remark}
 Analogously to what shown in Proposition~\ref{prop:at_once}, we still have that if $\vv \Bm=0$, then
\[
(\vv^{q^i}  \Sm^i ) \cdot (\trsp{\Sm})^i \Bm^{(q^i)} \Sm^i =(\vv \Bm)^{(q^i)} \Sm^i =0,
\]
therefore the nullspaces of blockwise Dickson shift matrices can be easily computed from the others.
Let $r_1,r_2$ be the unique integers such that $r=r_1(m+2)+r_2$ with $r_2\in \Iintv{1}{m+2}$. We split the analysis into different cases:
\begin{itemize}
\item \textbf{Case $4\le r_2\le m+2$.} Let $\Bm_1,\Bm_2\in\Cmat(\cB)$ of rank $(r-1)m-2$. Let us first consider the case $r_1=0$. Consider the linear code
\[
\left(\sum_{i=0}^{r-4} \ker \left((\trsp{\Sm})^i \Bm_2^{(q^i)} \Sm^i\right)\right).
\]
A generator matrix of  $\left(\sum_{i=0}^{r-4} \ker \left((\trsp{\Sm})^i \Bm_2^{(q^i)} \Sm^i\right)\right)\trsp{(\Pm^{-1})}$ can be written as
\begin{equation} \label{eq: genmat_sum B2^qi}
\begin{bmatrix}
\Gm_{0,0} &  &  & \\
	 & \Gm_{1,1} & & \zerov\\
	 \zerov & & \ddots & \\
	  &  &  & \Gm_{m-1,m-1} \\
\end{bmatrix}
\end{equation}
where $r-3$ cyclically consecutive diagonal blocks have $r-1$ rows while the others have $r-3$ rows (and they all have $r$ columns). 
A generator matrix of $\ker(\Bm_1)\trsp{(\Pm^{-1})}$ is instead given by
\begin{equation} \label{eq: genmat_B1}
\begin{bmatrix}
\vv_1 &  & \zerov  & & \zerov \\
	 & \ddots & & \\
 & & \vv_{l,1} & \\
	 \zerov & & \vv_{l,2} & & \zerov \\
 & & \vv_{l,3} & & \\
	 	 & & & \ddots &  \\
	\zerov &  & \zerov &  & \vv_m \\
\end{bmatrix}
\end{equation}
for some $l\in \Iintv{0}{m-1}$. 
If $\rank(\Gm_{l,l})=r-1$, then 
\begin{align*}
\dim_{\Fqm} \left(\ker(\Bm_1)+\sum_{i=0}^{r-4} \ker \left((\trsp{\Sm})^i \Bm_2^{(q^i)} \Sm^i\right)\right)=& \dim_{\Fqm} \left(\left(\ker(\Bm_1)+\sum_{i=0}^{r-4} \ker \left((\trsp{\Sm})^i \Bm_2^{(q^i)} \Sm^i\right)\right)\trsp{(\Pm^{-1})}\right) \\ \le & (r-4)((r-1)+1)+(r)+(m-r+3)((r-3)+1)\\
=& rm+2r-2m-6.
\end{align*}
On the other hand, if $\rank(\Gm_{l,l})=r-3$, then we expect with good probability that the dimension of $\ker(\Bm_1)+\sum_{i=0}^{r-4} \ker \left((\trsp{\Sm})^i \Bm_2^{(q^i)} \Sm^i\right)$ attains
\begin{align*}
&(r-3)((r-1)+1)+(m-r+2)((r-3)+1)+((r-3)+3)\\
=& rm+2r-2m-4.
\end{align*}
Therefore, by computing the dimension of $\ker(\Bm_1)+\sum_{i=0}^{r-4} \ker \left((\trsp{\Sm})^i \Bm_2^{(q^i)} \Sm^i\right)$ we determine whether the rank-$(r-3)$ block of $\Bm_1$ corresponds to a rank-$(r-3)$ block of some of the $\left((\trsp{\Sm})^i \Bm_2^{(q^i)} \Sm^i\right)$'s for some $i \in \Iintv{0}{r-4}$. By replacing $\Iintv{0}{r-4}$ with other subsets of $\Iintv{0}{m-1}$ of cardinality $r-3$ and repeating the process, we finally detect the sought $\left((\trsp{\Sm})^i \Bm_2^{(q^i)} \Sm^i\right)$. 
In the case where $r_1>0$, we need to sample independent rank-$((r-1)m-2)$ matrices $\Bm_{aux,1},\dots,\Bm_{aux,r_1}\in \Cmat(\cB)$. In this case a generator matrix of 
\[
\left(\sum_{j=1}^{r_1} \sum_{i=0}^{m-1} \ker \left((\trsp{\Sm})^i \Bm_{aux,j}^{(q^i)} \Sm^i\right)+\sum_{i=0}^{r_2-4} \ker \left((\trsp{\Sm})^i \Bm_2^{(q^i)} \Sm^i\right)\right)\trsp{(\Pm^{-1})}
\]
is with non-negligible probability as in \eqref{eq: genmat_sum B^qi}, where $r_2-3$ cyclically consecutive diagonal blocks have $r_1(m+2)+r_2-1=r-1$ rows while the others have $r_1(m+2)+r_2-3=r-3$ rows (and they all have $r$ columns). Hence, the computation of
\begin{align*}
& \dim_{\Fqm} \ker(\Bm_1)+\sum_{j=1}^{r_1} \sum_{i=0}^{m-1} \ker \left((\trsp{\Sm})^i \Bm_{aux,j}^{(q^i)} \Sm^i\right)+\sum_{i=0}^{r_2-4} \ker \left((\trsp{\Sm})^i \Bm_2^{(q^i)} \Sm^i\right)\\
& \begin{cases}
\le & (r_2-4)((r-1)+1)+(r)+(m-r_2+3)((r-3)+1)\\ &=rm+2r_2-2m-6 \quad \text{ if }\rank(\Gm_{l,l})=r-1,\\
= & (r_2-3)((r-1)+1)+(m-r_2+2)((r-3)+1)+((r-3)+3)\\& = rm+2r_2-2m-4 \quad \text{otherwise (with high probability)}
\end{cases}
\end{align*}
reveals again whether the rank-$(r-3)$ block of $\Bm_1$ corresponds or not to one of the rank-$(r-3)$ blocks of the $\ker \left((\trsp{\Sm})^i \Bm_2^{(q^i)} \Sm^i\right)$'s.
\item \textbf{Case $r_2=1$ and $r_1\ge 1$.} The reasoning is very similar to the one in the previous case. An analogous computation shows that 
\begin{align*}
&\dim_{\Fqm} \ker(\Bm_1)+\sum_{j=1}^{r_1-1} \sum_{i=0}^{m-1} \ker \left((\trsp{\Sm})^i \Bm_{aux,j}^{(q^i)} \Sm^i\right)+\sum_{i=0}^{m-2} \ker \left((\trsp{\Sm})^i \Bm_2^{(q^i)} \Sm^i\right)\\
& \begin{cases}
\le & (m-2)((r-2)+1)+(r)+((r-4)+1)\\ &=rm-m-1 \quad \text{ if }\rank(\Gm_{l,l})=r-2,\\
= & (m-1)((r-2)+1)+((r-4)+3)\\& = rm-m \quad \text{otherwise (with high probability)}
\end{cases},
\end{align*}
for the index $l\in \Iintv{0}{m-1}$ such that a generator matrix of $\ker(\Bm_1)\trsp{(\Pm^{-1})}$ is as in \eqref{eq: genmat_B1} and a generator matrix of \\$\left(\sum_{j=1}^{r_1-1} \sum_{i=0}^{m-1} \ker \left((\trsp{\Sm})^i \Bm_{aux,j}^{(q^i)} \Sm^i\right)+\sum_{i=0}^{m-2} \ker \left((\trsp{\Sm})^i \Bm_2^{(q^i)} \Sm^i\right)\right)\trsp{(\Pm^{-1})}$ is as in \eqref{eq: genmat_sum B2^qi}, with $m-1$ cyclically consecutive diagonal blocks having $r-2$ rows while the other having $r-4$ rows (and they all have $r$ columns). Hence we can distinguish the case $\rank(\Gm_{l,l})=r-2$ from $\rank(\Gm_{l,l})=r-4$.
\item \textbf{Case $r_2=2$ and $r_1\ge 1$.} In this case, we take two consecutive blockwise Dickson shifts of $\Bm_1$, \ie $\Bm_1$ and  $\trsp{\Sm} \Bm_1^{(q)} \Sm$. Therefore a generator matrix of $\left(\ker(\Bm_1)+\ker(\trsp{\Sm} \Bm_1^{(q)} \Sm)\right)\trsp{(\Pm^{-1})}$ is given by either
\[
\begin{bmatrix}
\Vm_{0,0} &  &  & \\
	 & \Vm_{1,1} & & \zerov\\
	 \zerov & & \ddots & \\
	  &  &  & \Vm_{m-1,m-1} \\
\end{bmatrix}
\]
where $2$ cyclically consecutive diagonal blocks have $4$ rows while the others have $2$ rows (and they all have $r$ columns). Let us say that the two blocks with 4 rows are indexed by $l$ and $l+1 \mod m$, for some $l \in \Iintv{0}{m-1}$. Then we get
\begin{align*}
&\dim_{\Fqm} \ker(\Bm_1)+\ker(\trsp{\Sm} \Bm_1^{(q)} \Sm)+\sum_{j=1}^{r_1-1} \sum_{i=0}^{m-1} \ker \left((\trsp{\Sm})^i \Bm_{aux,j}^{(q^i)} \Sm^i\right)+\sum_{i=0}^{m-2} \ker \left((\trsp{\Sm})^i \Bm_2^{(q^i)} \Sm^i\right)\\
& \begin{cases}
\le & (m-3)((r-3)+2)+(2r)+((r-5)+2)\\ &=rm-m \quad \text{ if }\rank(\Gm_{l,l})=r-3 \land \rank(\Gm_{l+1 \mod m,l+1 \mod m})=r-3,\\
= & (m-2)((r-3)+2)+(r)+((r-5)+3+1)\\& = rm-m+1 \quad \text{otherwise (with high probability)}
\end{cases}
\end{align*}
where a generator matrix of \\$\left(\sum_{j=1}^{r_1-1} \sum_{i=0}^{m-1} \ker \left((\trsp{\Sm})^i \Bm_{aux,j}^{(q^i)} \Sm^i\right)+\sum_{i=0}^{m-2} \ker \left((\trsp{\Sm})^i \Bm_2^{(q^i)} \Sm^i\right)\right)\trsp{(\Pm^{-1})}$ is as in \eqref{eq: genmat_sum B2^qi}, with $m-1$ cyclically consecutive diagonal blocks having $r-3$ rows while the other having $r-5$ rows (and they all have $r$ columns). Hence we can distinguish the case $\rank(\Gm_{l,l})=r-3 \land \rank(\Gm_{l+1 \mod m,l+1 \mod m})=r-3$ from $\rank(\Gm_{l,l})=r-5 \lor \rank(\Gm_{l+1 \mod m,l+1 \mod m})=r-5$. Repeating the process at most $m$ times for different pairs of consecutive diagonal block shifts of $\Bm_1$, solves our problem.
\item \textbf{Case $r_2=3$.} If $r_1=0$, the kernel of a single matrix $\Bm$ already defines $\VC_j$ (note that we can choose $s=\ceil{\frac{r}{3}}=\frac{r-1}{2}=1$). Otherwise, let $\Bm_1$ and $\Bm_2$ such that generator matrices of $\ker(\Bm_1) \trsp{(\Pm^{-1})}$ and $\ker(\Bm_2) \trsp{(\Pm^{-1})}$ respectively are given by
\begin{equation} 
\begin{bmatrix}
\vv_1 &  & \zerov  & & \zerov \\
	 & \ddots & & \\
 & & \vv_{l_1,1} & \\
	 \zerov & & \vv_{l_1,2} & & \zerov \\
 & & \vv_{l_2,3} & & \\
	 	 & & & \ddots &  \\
	\zerov &  & \zerov &  & \vv_m \\
\end{bmatrix}
\quad \text{and} \quad 
\begin{bmatrix}
\uv_1 &  & \zerov  & & \zerov \\
	 & \ddots & & \\
 & & \uv_{l_2,1} & \\
	 \zerov & & \uv_{l_2,2} & & \zerov \\
 & & \uv_{l_2,3} & & \\
	 	 & & & \ddots &  \\
	\zerov &  & \zerov &  & \uv_m \\
\end{bmatrix}.
\end{equation}
A generator matrix of $\left(\sum_{j=1}^{r_1} \sum_{i=0}^{m-1} \ker \left((\trsp{\Sm})^i \Bm_{aux,j}^{(q^i)} \Sm^i\right)\right)\trsp{(\Pm^{-1})}$ is expected to be as in \eqref{eq: genmat_sum B2^qi}, with all the block have $r-3$ rows (and $r$ columns). Therefore
\begin{align*}
&\dim_{\Fqm} \ker(\Bm_1)+\ker((\trsp{\Sm})^l \Bm_2^{(q^l)} (\Sm)^l)+\sum_{j=1}^{r_1} \sum_{i=0}^{m-1} \ker \left((\trsp{\Sm})^i \Bm_{aux,j}^{(q^i)} \Sm^i\right)\\
& \begin{cases}
\le & (m-1)((r-3)+2)+(r)\\ &=rm-m+1 \quad \text{ if }l_1=l_2+l \mod m,\\
= & (m-2)((r-3)+2)+(2r)+((r-5)+3+1)\\& = rm-m+2 \quad \text{otherwise (with high probability)}
\end{cases}
\end{align*}
Repeating the process at most $m$ times for different values of $l$ solves our problem.
\end{itemize}

\newcommand{\etalchar}[1]{$^{#1}$}

\end{document}